\documentclass[journal,onecolumn,12pt]{IEEEtran}
\usepackage{times,amsmath,amsthm,epsfig,amssymb,graphicx,amsfonts}
\usepackage{hyphenat}
\usepackage{psfrag}
\usepackage{ifpdf}
\usepackage{cite}
\usepackage{array}
\usepackage{multirow}
\usepackage{graphicx}
\usepackage{color}
\usepackage{subfigure}
\usepackage{epsfig}
\usepackage{citesort}
\usepackage{setspace}
\usepackage{latexsym}
\usepackage{amssymb}
\usepackage{amsmath}

\title{On the Capacity and Degrees of Freedom Regions of MIMO Interference Channels with Limited Receiver Cooperation}
\author{Mehdi Ashraphijuo, Vaneet Aggarwal and Xiaodong Wang \thanks{ M. Ashraphijuo and X. Wang are with the Electrical Engineering Department, Columbia University, New York, NY 10027, email: \{mehdi,wangx\}@ee.columbia.edu. V. Aggarwal is with AT\&T Labs-Research, Florham Park, NJ 07932, email: vaneet@research.att.com }}
\newtheorem{theorem}{Theorem}
\newtheorem{definition}{Definition}

\newtheorem{lemma}{Lemma}
\newtheorem{corollary}{Corollary}

\begin{document}
\maketitle
\begin{abstract}

This paper gives the approximate capacity region of a two-user MIMO interference channel with limited receiver cooperation, where the gap between the inner and outer bounds is in terms of the total number of receive antennas  at the two receivers and is independent of the actual channel values. The approximate capacity region is then used to find the degrees of freedom region. For the special case of symmetric interference channels, we also find the amount of receiver cooperation in terms of the backhaul capacity beyond which the degrees of freedom do not improve. Further, the generalized degrees of freedom are found for MIMO interference channels with equal number of antennas at all nodes. It is shown that the generalized degrees of freedom improve gradually from a ``W'' curve to a ``V'' curve with increase in cooperation in terms of the backhaul capacity.
\end{abstract}

\newpage
\section{Introduction}

Wireless networks with multiple users are interference-limited rather than noise-limited. Interference channel (IC) is a good starting point for understanding the performance limits of interference-limited communications. In spite of research spanning over three decades, the capacity of the IC has been characterized only for some special cases \cite{Etkin,araja,sason,Varanasi,gdof,gdof.tarokh,Jafar}.

ICs model practical cellular networks. However, since the cellular base stations are connected via backhaul links, making efficient use of the backhaul is an important practical problem. This backhaul can lead to cooperation between transmitters in the downlink and cooperation between the receivers in the uplink \cite{C1,C2,C3,C4,C5,C6}. Cooperation between transmitters or receivers can help mitigate interference by forming a distributed MIMO system which provides performance gain. The rate at which they cooperate, however, is limited, due to physical constraints. In this paper, we tackle the fundamental problem of efficient use of limited-capacity backhaul for multiple-input multiple-output (MIMO) uplink ICs (with receiver cooperation).  Recently, many results have shown that transmitter and receiver cooperation can be employed in ICs to achieve an improvement in data rates \cite{Tse,Jafarr,TC,Maric,Prabha2,Prabha,CoMP1,CoMP2,CoMP3}. However, most of the existing works on ICs with cooperation are limited to discrete memoryless channels or to single-input single-output (SISO) channels. This paper analyzes two-user MIMO Gaussian ICs with limited receiver cooperation.

The authors of \cite{Tse} considered a two-user SISO Gaussian IC with limited receiver cooperation where there are links with fixed capacities between the two receivers and they found the capacity region of the channel within two bits. In this paper, we find an outer bound and an inner bound for the capacity region  that are within $N_1+N_2$ bits for a general MIMO IC with limited receiver cooperation (i.e., limited backhaul capacity), where $N_1$ and $N_2$ are the numbers of receive antennas at the two receivers, respectively. We use an achievability scheme based on that for the discrete memoryless channel in \cite{Tse}. In this scheme, receivers do not decode the messages immediately upon receiving the signals from the transmitters. One of the receivers quantizes its received signal at an appropriate distortion, bins the quantization codeword  and sends the bin index to the other receiver. For quantizing the received signal, a novel distortion function for MIMO IC is given in this paper. The other receiver decodes its own information based on its own received signal and the received bin index. After decoding, it bin-and-forwards the decoded common messages back to the other receiver and helps it decode. This paper uses the signal distributions and auxiliary variables that are different from those in \cite{Tse} and in such a way that can be used for a MIMO IC to achieve a constant gap to the capacity region.

We note that the achievability strategy for the SISO IC in \cite{Tse} is to split the transmit signal into public and private messages using the Han-Kobayashi message splitting where the private message is received at the unintended receiver below the noise floor. For a MIMO IC with limited receiver cooperation, we proposed a counterpart of Han-Kobayashi splitting in \cite{J1} where the covariance matrices for the public and the private messages were properly designed. In this paper, we give an achievability scheme based on the splitting scheme in \cite{J1}. Further, the authors of \cite{Tse} proposed different choices of power splits between the public and the private messages for  three different regions of SISO IC corresponding to: weak, mixed and strong interferences. In this paper, for MIMO IC, we propose a single choice of covariance matrices for the public and private messages for all regimes rather than considering different regimes separately. For the special case of SISO IC, the achievability scheme used in this paper reduces to a different one from that given in \cite{Tse}. The achievability scheme uses the convex hull of the regions formed by two strategies corresponding to different decoding orders. The convex hull of the two regions eliminates two constraints in each region resulting in a constant gap between the inner and the outer bounds. 


Having characterized the outer and inner bounds within a constant gap, and as a result having the approximate capacity region, we also find the degrees of freedom (DoF) region for the two-user MIMO IC with limited receiver cooperation. We find that the DoF region improves with the increase in cooperation in terms of the backhaul capacity. For the case of symmetric number of antennas in both the transmitters and the receivers, we find that the DoF improves up to a certain point in terms of the backhaul capacity, and beyond which the DoF does not improve anymore.

The symmetric DoF region formed when both each transmitter has $M$ antennas and each receiver has $N$ antennas, is a pentagon with bounds only on individual DoF $(d_1,d_2)$ and sum DoF $(d_1+d_2)$ for all cases except when $N<M<2N$. Thus, when the number of antennas at all the nodes are the same i.e., $M=N$, the DoF region is a pentagon. However, when $N<M<2N$, the DoF region also has constraints on $2d_1+d_2$ and $d_1+2d_2$. These constraints are known to not hold for IC with no cooperation \cite{gdof}, and for ICs with infinite cooperation which corresponds to a multiple-access channel (MAC)\cite{Tse.book}. In this paper, we find that the extra bounds on $2d_1+d_2$ and $d_1+2d_2$ are dominant for a finite non-zero limited cooperation (when the backhaul capacity is less than a certain value) for $N<M<2N$. We note that this result shows that the role of transmit and receive antennas cannot be interchanged to get the reciprocity result which exists in the case of no cooperation \cite{gdof}.

Finally, we also characterize the generalized degrees of freedom (GDoF) for a MIMO IC with limited receiver cooperation, when the cooperation links are of the same capacity which is increasing with a  base signal-to-noise ratio (SNR) parameter, say $\mathsf{SNR}$ (as $\beta\log \mathsf{SNR}$). Note that even though the DoF region is found in general, we find the GDoF only in a limited setting  when the number of antennas at all the nodes are the same (say $M$). We assume that the channel strengths of the direct links have values of the order $\mathsf{SNR}$ while the cross links have channel strengths of the order $\mathsf{SNR}^\alpha$.  We find that the increase in the cooperation leads to improvement in GDoF. For a given $M$ and $\alpha$, the GDoF increases till $\beta = M\alpha$ at which point the GDoF with limited cooperation is the same as that with full cooperation. Without any receiver cooperation, the GDoF  is a ``W" curve (a curve between the GDoF and $\alpha$). We note that the  ``W" curve changes to a ``V" curve and then to an increasing function as the backhaul capacity increases.


The remainder of the paper is organized as follows. Section II introduces the model for a MIMO IC with limited receiver cooperation and the capacity region. Sections III describes our results on the capacity region. Section IV describes our results on the DoF region and the GDoF. Section V concludes the paper. The detailed proofs of various results are given in Appendices A-D.

\section{Channel Model and Preliminaries}
In this section, we describe the channel model that is used in this paper. A two-user MIMO IC consists of two transmitters and two receivers. Transmitter $i$ is labeled as $\mathsf{T}_i$ and receiver $j$ is labeled as $\mathsf{D}_j$ for $i,j \in \{1,2\}$. Further, we assume $\mathsf{T}_i$ has $M_i$ antennas and $\mathsf{D}_j$ has $N_j$ antennas. Henceforth, such a MIMO IC will be referred to  as the $(M_1,N_1,M_2,N_2)$ MIMO IC. The channel matrix between transmitter $\mathsf{T}_i$ and receiver $\mathsf{D}_j$ is denoted by  $H_{ij}\in \mathbb{C}^{N_j\times M_i}$. We shall consider a time-invariant channel where the channel matrices remain fixed for the entire duration of communication.  At time $t$, transmitter $\mathsf{T}_i$ transmits a vector $X_{i}(t)\in \mathbb{C}^{M_i\times 1}$ over the channel with a power constraint ${{\rm tr}(\mathbb{E}(X_{i}X^{\dagger }_{i}))}\le 1$ ($A^\dagger$ is the conjugate transpose of the matrix $A$).

Let $Q_{ij}=\mathbb{E}(X_{i}X^{\dagger }_{j})$ for $i,j\in \{1,2\}$. We say $A\preceq B$ if $B-A$ is a positive semi-definite (p.s.d.) matrix and we say $A\succeq B$ if $B\preceq A$. The identity matrix of size $s\times s$ is denoted by $I_{s}$. Further, we define $x^+\triangleq \max\{x,0\}$. We also note that $0 \preceq Q_{ii}\preceq I$, and $0\preceq Q_{ij}Q_{ij}^\dagger\preceq I$.

We also incorporate a non-negative power attenuation factor, denoted as ${\rho}_{ij}$, for the signal transmitted from $\mathsf{T}_i$ to $\mathsf{D}_j$. The received signal at receiver $\mathsf{D}_i$ at time $t$ is denoted as $Y_{i}(t)$ for $i\in \{1,2\}$, and can be written as
\begin{eqnarray}
Y_{1}(t)&=&\sqrt{{\rho }_{11}}H_{11}X_{1}(t)+\sqrt{{\rho }_{21}}H_{21}X_{2}(t)+Z_{1}(t),\\
Y_{2}(t)&=&\sqrt{{\rho }_{12}}H_{12}X_{1}(t)+\sqrt{{\rho }_{22}}H_{22}X_{2}(t)+Z_{2}(t),
\end{eqnarray}
where $Z_{i}(t)\in \mathbb{C}^{N_i\times 1}\sim\mathsf{CN}(0,I_{N_i})$ is the i.i.d. complex Gaussian noise, ${\rho }_{ii}$ is the received SNR at receiver $\mathsf{D}_i$ and ${\rho }_{ij}$ is the received interference-to-noise-ratio at receiver $\mathsf{D}_j$ for $i, j\in \{1,2\},$ $i \ne j$. A MIMO IC with limited receiver cooperation is fully described by four parameters. The first is the numbers of antennas at each transmitter and receiver, namely $(M_1,N_1,M_2,N_2)$. The second is the set of channel gains, $\overline{H}=\{H_{11},H_{12},H_{21},H_{22}\}$. The third is the set of average link qualities, ${\overline{\rho }}=\{{\rho }_{11}{,{\rho }_{12},\rho }_{21},{\rho }_{22}\}$. The fourth parameter is ${\overline{C}}=\{{C }_{12},{C}_{21}\}$ where ${C}_{ji}$ is the capacity of the cooperation (backhaul) link from receiver $\mathsf{D}_j$ to $\mathsf{D}_i$. We assume that these parameters are known to all transmitters and receivers. Also, the cooperation channels are orthogonal to each other and they are orthogonal to the data channels.

The receiver-cooperation links are noiseless with finite capacities. Encoding is causal in the sense that the signal transmitted from  $\mathsf{D}_j$ at time $t$ is a function of whatever is received over the data channel, or on the cooperation link up to time $t-1$.
%
%
In addition, the decoded signal at $\mathsf{D}_i$, $\widehat{m_i}$, is a function of the received signal from the channel, $Y_{i}(t)$, and the cooperation signal transmitted from $\mathsf{D}_j$ to $\mathsf{D}_i$, $\mathsf{\Gamma}_{ji}$, for $i \in \{1,2\}$. Thus, the decoding functions of the two receivers are given as
\begin{eqnarray}
\hat{m_i}&=&f_{it}(\mathsf{\Gamma}_{ji},Y_{i}(t)),
\end{eqnarray}
where $f_{it}$ is the decoding function of $\mathsf{D}_i$ at time $t$. Let us assume that $\mathsf{T}_i$ transmits information at a rate of $R_i$ to receiver $\mathsf{D}_i$ using the codebook ${\mathcal C}_{i,n}$ of length-$n$ codewords with $|{\mathcal C}_{i,n}|=2^{nR_i}$. Given a message $m_i\in \{1,\dots ,2^{nR_i}\}$, the corresponding codeword $X^n_i(m_i)\in {\mathcal C}_{i,n}$ satisfies the power constraint mentioned before. From the received signal $Y^n_i$ and the cooperation message from $\mathsf{D}_j$, i.e. $\Gamma_{ji}$, $\mathsf{D}_i$ obtains an estimate $\hat{m_i}$ of the transmitted message $m_i$ using a decoding function. Denote the average probability of decoding error by $e_{i,n}={\rm Pr}({\rm \ }\hat{m_i}\ne m_i)$.

A rate pair $(R_1,R_2)$ is achievable if there exists a family of codebooks ${\mathcal C}_{i,n}$ and decoding functions such that ${{\rm max}}_i\{e_{i,n}\}$ goes to zero as the block length $n$ goes to infinity. The capacity region ${\mathbb C}(\overline{H},\overline{\rho },{\overline{C}})$ of the IC with parameters $\overline{H}$, $\overline{\rho }$ and ${\overline{C}}$ is defined as the closure of the set of all achievable rate pairs.

Consider a two-dimensional rate region ${\mathbb C}(\overline{H},\overline{\rho },{\overline{C}})$. Then, the region ${\mathbb C}(\overline{H},\overline{\rho },{\overline{C}})\ominus ([0,a]\times [0,b])$ denotes the region formed by $\{(R_1,R_2): R_1, R_2\ge 0, (R_1+a, R_2+b)\in {\mathbb C}(\overline{H},\overline{\rho },{\overline{C}})\}$  for some $a,b\ge 0$. Further, we define the notion of an achievable rate region that is within a constant number of bits of the capacity region as follows.

\begin{definition}
An achievable rate region ${\mathbb A}(\overline{H},\overline{\rho },{\overline{C}})$ is said to be within $b$ bits of the capacity region if ${\mathbb A}(\overline{H},\overline{\rho },{\overline{C}}) \subseteq {\mathbb C}(\overline{H},\overline{\rho },{\overline{C}})$ and ${\mathbb A}(\overline{H},\overline{\rho },{\overline{C}})\oplus ([0,b]\oplus[0,b]) \supseteq {\mathbb C}(\overline{H},\overline{\rho },{\overline{C}})$.
\end{definition}

\section{Inner and Outer Bounds on Capacity Region}\label{main_results}

In this section, we give the outer and inner bounds on the capacity region of a two-user MIMO IC with limited receiver cooperation. Let $\mathcal{R}_o$ be the convex hull of the region formed by $(R_1,R_2)$ satisfying the following constraints.

\begin{eqnarray}
R_1&\le& {\log  {\det  \left(I_{N_1}+{\rho}_{11}H_{11}H^{\dagger}_{11}\right)}}+\min \{\log \det \left(I_{N_2}+{\rho }_{12}H_{12}H^{\dagger }_{12}-{\rho }_{12}{\rho }_{11}H_{12}H^{\dagger }_{11}\right. \nonumber\\
&&\left.{\left(I_{N_1}+{\rho}_{11}H_{11}H^{\dagger}_{11}\right)^{-1}}H_{11}H^{\dagger }_{12}\right),C_{21} \},\label{ro0eq1}\\
R_2&\le& {\log  {\det  \left(I_{N_2}+{\rho}_{22}H_{22}H^{\dagger}_{22}\right)}}+\min\{\log \det \left(I_{N_1}+{\rho }_{21}H_{21}H^{\dagger }_{21}-{\rho }_{21}{\rho }_{22}H_{21}H^{\dagger }_{22}\right.\nonumber\\
&&\left.{\left(I_{N_2}+{\rho}_{22}H_{22}H^{\dagger}_{22}\right)^{-1}}H_{22}H^{\dagger }_{21}\right),C_{12}\},\label{ro0eq2}\\
R_1+R_2 &\le& \log  \det \left(I_{N_1}+{\rho}_{11}H_{11}H^{\dagger}_{11}+
{\rho}_{21}H_{21}H^{\dagger}_{21}-
{\rho}_{11}{\rho}_{12}H_{11}H^{\dagger}_{12}
{\left(I_{N_2}+{\rho}_{12}H_{12}H^{\dagger}_{12}\right)^{-1}}H_{12}H^{\dagger}_{11}\right)+\nonumber\\
&&\log  \det \left(I_{N_2}+{\rho}_{22}H_{22}H^{\dagger}_{22}+
{\rho}_{12}H_{12}H^{\dagger}_{12}-
{\rho}_{22}{\rho}_{21}H_{22}H^{\dagger}_{21}
{\left(I_{N_1}+{\rho}_{21}H_{21}H^{\dagger}_{21}\right)^{-1}}H_{21}H^{\dagger}_{22}\right)+\nonumber\\
&&C_{12}+C_{21},\label{ro0eq3}\\
R_1+R_2 &\le& \log  \det \left(I_{N_1}+{\rho}_{11}H_{11}H^{\dagger}_{11}-
{\rho}_{11}{\rho}_{12}H_{11}H^{\dagger}_{12}
{\left(I_{N_2}+{\rho}_{12}H_{12}H^{\dagger}_{12}\right)^{-1}}H_{12}H^{\dagger}_{11}\right)+\nonumber\\
&&\log  \det \left(I_{N_2}+{\rho}_{22}H_{22}H^{\dagger}_{22}+
{\rho}_{12}H_{12}H^{\dagger}_{12}\right)+C_{12},\label{ro0eq4}\\
R_1+R_2 &\le& \log  \det \left(I_{N_2}+{\rho}_{22}H_{22}H^{\dagger}_{22}-
{\rho}_{22}{\rho}_{21}H_{22}H^{\dagger}_{21}
{\left(I_{N_1}+{\rho}_{21}H_{21}H^{\dagger}_{21}\right)^{-1}}H_{21}H^{\dagger}_{22}\right)+\nonumber\\
&&\log  \det \left(I_{N_1}+{\rho}_{11}H_{11}H^{\dagger}_{11}+
{\rho}_{21}H_{21}H^{\dagger}_{21}\right)+C_{21},\label{ro0eq5}\\
R_1+R_2 &\le& \log  \det \left(I_{N_1+N_2}+
\left[\begin{array}{c}
\sqrt{\rho_{11}}H_{11} \\
\sqrt{\rho_{12}}H_{12} \end{array}\right]
\left[\sqrt{\rho_{11}}H_{11}^{\dagger}\ \sqrt{\rho_{12}}H_{12}^{\dagger}\right]+\right.\nonumber\\
&&\left.\left[\begin{array}{c}
\sqrt{\rho_{21}}H_{21} \\
\sqrt{\rho_{22}}H_{22} \end{array}\right]
\left[\sqrt{\rho_{21}}H_{21}^{\dagger}\ \sqrt{\rho_{22}}H_{22}^{\dagger}\right]\right),\label{ro0eq6}
\end{eqnarray}

\begin{eqnarray}
2R_1+R_2 &\le& \log  \det \left(I_{N_1}+{\rho}_{11}H_{11}H^{\dagger}_{11}-
{\rho}_{11}{\rho}_{12}H_{11}H^{\dagger}_{12}
{\left(I_{N_2}+{\rho}_{12}H_{12}H^{\dagger}_{12}\right)^{-1}}H_{12}H^{\dagger}_{11}\right)+\nonumber\\
&&\log  \det \left(I_{N_2}+{\rho}_{22}H_{22}H^{\dagger}_{22}+
{\rho}_{12}H_{12}H^{\dagger}_{12}-
{\rho}_{22}{\rho}_{21}H_{22}H^{\dagger}_{21}
{\left(I_{N_1}+{\rho}_{21}H_{21}H^{\dagger}_{21}\right)^{-1}}H_{21}H^{\dagger}_{22}\right)+\nonumber\\
&&\log  \det \left(I_{N_1}+{\rho}_{11}H_{11}H^{\dagger}_{11}+
{\rho}_{21}H_{21}H^{\dagger}_{21}\right)+C_{12}+C_{21},\label{ro0eq7}\\
R_1+2R_2 &\le& \log  \det \left(I_{N_2}+{\rho}_{22}H_{22}H^{\dagger}_{22}-
{\rho}_{22}{\rho}_{21}H_{22}H^{\dagger}_{21}
{\left(I_{N_1}+{\rho}_{21}H_{21}H^{\dagger}_{21}\right)^{-1}}H_{21}H^{\dagger}_{22}\right)+\nonumber\\
&&\log  \det \left(I_{N_1}+{\rho}_{11}H_{11}H^{\dagger}_{11}+
{\rho}_{21}H_{21}H^{\dagger}_{21}-
{\rho}_{11}{\rho}_{12}H_{11}H^{\dagger}_{12}
{\left(I_{N_2}+{\rho}_{12}H_{12}H^{\dagger}_{12}\right)^{-1}}H_{12}H^{\dagger}_{11}\right)+\nonumber\\
&&\log  \det \left(I_{N_2}+{\rho}_{22}H_{22}H^{\dagger}_{22}+
{\rho}_{12}H_{12}H^{\dagger}_{12}\right)+C_{21}+C_{12},\label{ro0eq8}\\
2R_1+R_2 &\le& \log  \det \left(I_{N_1+N_2}+
\left[\begin{array}{c}
\sqrt{\rho_{22}}H_{22} \\
\sqrt{\rho_{21}}H_{21} \end{array}\right]
\left(I_{M_2}-{\rho}_{21}H^{\dagger}_{21}{\left(I_{N_1}+{\rho}_{21}H_{21}H^{\dagger}_{21}\right)^{-1}}H_{21}\right)\right.\nonumber\\
&&\left.\left[\sqrt{\rho_{22}}H_{22}^{\dagger}\ \sqrt{\rho_{21}}H_{21}^{\dagger}\right]
+\left[\begin{array}{c}
\sqrt{\rho_{12}}H_{12} \\
\sqrt{\rho_{11}}H_{11}\end{array}\right]
\left[\sqrt{\rho_{12}}H_{12}^{\dagger}\ \sqrt{\rho_{11}}H_{11}^{\dagger}\right]\right)+\nonumber\\
&&\log  \det \left(I_{N_1}+{\rho}_{11}H_{11}H^{\dagger}_{11}+
{\rho}_{21}H_{21}H^{\dagger}_{12}\right)+C_{21},\label{ro0eq9}\\
R_1+2R_2 &\le& \log  \det \left(I_{N_1+N_2}+
\left[\begin{array}{c}
\sqrt{\rho_{11}}H_{11} \\
\sqrt{\rho_{12}}H_{12} \end{array}\right]
\left(I_{M_1}-{\rho}_{12}H^{\dagger}_{12}{\left(I_{N_2}+{\rho}_{12}H_{12}H^{\dagger}_{12}\right)^{-1}}H_{12}\right)\right.\nonumber\\
&&\left.\left[\sqrt{\rho_{11}}H_{11}^{\dagger}\ \sqrt{\rho_{12}}H_{12}^{\dagger}\right]
+\left[\begin{array}{c}
\sqrt{\rho_{21}}H_{21} \\
\sqrt{\rho_{22}}H_{22} \end{array}\right]
\left[\sqrt{\rho_{21}}H_{21}^{\dagger}\ \sqrt{\rho_{22}}H_{22}^{\dagger}\right]\right)+\nonumber\\
&&\log  \det \left(I_{N_2}+{\rho}_{22}H_{22}H^{\dagger}_{22}+
{\rho}_{12}H_{12}H^{\dagger}_{12}\right)+C_{12}.\label{ro0eq10}
\end{eqnarray}
The following theorem shows that the capacity region of a two-user MIMO IC with limited receiver cooperation is within $N_1+N_2$ bits of ${\mathcal R}_o$.

\begin{theorem}
The capacity region for an $(M_1,N_1,M_2,N_2)$ two-user MIMO IC with limited receiver cooperation, ${\mathbb C}_{RC}$, is bounded from outside and inside as
\begin{equation}
\mathcal{R}_o \ominus ([0,N_1+N_2]\times [0,N_1+N_2]) \subseteq {\mathbb C}_{RC}\subseteq \mathcal{R}_o.
\end{equation}
Thus, the inner and outer bounds are within $N_1+N_2$ bits. \label{outer_inner_capacity_reciprocal}
\end{theorem}

\textbf{Outer Bound:} The complete proof that $\mathcal{R}_o$ is an outer bound for the capacity region of the two-user MIMO IC with limited receiver cooperation is given in Appendix A.

\eqref{ro0eq1}, \eqref{ro0eq2} and \eqref{ro0eq6} are cut-set based upper bounds. The other bounds are obtained based on genie-aided strategies making use of Fano's inequality, the data processing inequality, the fact that the cooperation messages are functions of $(Y_1,Y_2)$, and the fact that Gaussian distribution maximizes the entropy. The detailed derivations are given in Appendix A.


\textbf{Inner Bound:} Here,  we will give a brief description of the achievability strategy. The complete proof can be found in Appendix B.

The achievability scheme is based on a two-round strategy, similar to that used in \cite{Tse} for SISO interference channels. It consists of two parts: 1) the transmission scheme and 2) the cooperation protocol.

\noindent 1) Transmission Scheme:

Each transmitter $\mathsf{T}_i$ splits its own message into private and common sub-messages and sends
\begin{eqnarray}
X_i=X_{ip}+X_{ic},\label{1t}
\end{eqnarray}
where $X_{ip}\sim \mathsf{CN}\left(0,Q_{{ip}}\right)$ denotes the private message, and $X_{ic}\sim \mathsf{CN}\left(0,Q_{{ic}}\right)$ denotes the common message. We assume that $X_{ip}$ and $X_{ic}$ are independent with
\begin{eqnarray}
Q_{{ip}}
&=&I_{M_i}-
\sqrt{{\rho }_{ij}}H^{\dagger }_{ij}
(I_{N_j}+{\rho }_{ij}H_{ij}H^{\dagger }_{ij})^{-1}
\sqrt{{\rho }_{ij}}H_{ij},\label{2t}\\
{\rm and}\ \ \ Q_{{ic}}&=&I_{M_i}-Q_{{ip}},\label{3t}
\end{eqnarray}
for $i\in \{1,2\}$.

It is shown in Appendix B of \cite{J1} that $Q_{{ip}}\succeq 0$ and $Q_{{ic}}\succeq 0$. Further, this message split is such that a private signal is received at the other receiver with constant power. More specifically, we have $\rho_{ij}H_{ij}Q_{{ip}} H_{ij}^\dagger \preceq I_{N_j}$ thus the received signal at receiver $\mathsf{D}_j$ corresponding to the private signal from transmitter $\mathsf{T}_i$ is below the noise floor.

{\em Remark:} Note that the power allocation in \eqref{2t}-\eqref{3t} is different from that given in \cite{Tse} even for a SISO channel. In \cite{Tse} different achievability schemes were given for weak, mixed, and strong interference regimes. Here what we propose is a single choice of parameters for all interference regimes. For the special case of SISO IC, the above scheme constitutes an alternative choice of variances to those proposed in \cite{Tse}.

\noindent 2) Cooperation Protocol:

We use a two-round cooperation protocol similar to that in \cite{Tse}. In the first round, $\mathsf{D}_j$ quantizes its received signal and sends out the bin index. And then in the second round, $\mathsf{D}_i$ $i\ne j$ $(i,j)\in\{1,2\}$ receives this side information and decodes its desired messages (its own message plus the other's public message). After decoding, $\mathsf{D}_i$ randomly bins the decoded public messages, and sends the bin indices to $\mathsf{D}_j$. Finally, $\mathsf{D}_j$ decodes its message. In this two-round strategy, ${STG}_{j\rightarrow i\rightarrow j}$, the processing order is: $\mathsf{D}_j$ quantize-and-bins, $\mathsf{D}_i$ decode-and-bins and finally $\mathsf{D}_j$ decodes. Its achievable rate region is denoted by ${\mathcal R}_{j\rightarrow i\rightarrow j}$. By time-sharing, the rate region ${\mathcal R}\triangleq conv\{{{\mathcal R}_{2\rightarrow 1\rightarrow 2}}\cup {{\mathcal R}_{1\rightarrow 2\rightarrow 1}}\}$, i.e. the convex hull of the union of the two rate regions is achievable.

For simplicity, we consider strategy ${\mathcal R}_{2\rightarrow 1\rightarrow 2}$. $\mathsf{D}_2$ does not decode messages immediately upon receiving its signal. It first quantizes its signal by a pre-generated Gaussian quantization codebook with certain distortion and then sends out a bin index determined by a pre-generated binning function. It sets the distortion level equal to the aggregate power level of the noise and $\mathsf{T}_2$'s private signal. $\mathsf{D}_1$ decodes the two common messages and its own private message, by searching in transmitters' codebooks for a codeword triplet that is jointly typical with its received signal and some quantization point (codeword) in the given bin after retrieving the receiver-cooperative side information (the bin index). After $\mathsf{D}_1$ decodes, it uses two pre-generated binning functions to bin the two common messages and sends out these two bin indices to $\mathsf{D}_2$. After receiving these two bin indices, $\mathsf{D}_2$ decodes the two common messages and its own private message, by searching in the corresponding bins and $\mathsf{T}_2$'s private codebook for a codeword triplet that is jointly typical with its received signal.

Although the cooperation protocol is similar to that in \cite{Tse}, the distortion function used for the quantization of the received signal needs to be extended to the case of multiple antennas. We here describe the distortion function for ${STG}_{2\rightarrow 1\rightarrow 2}$. For the quantization, we use the quantization codebook satisfying
\begin{equation}\label{hat}
\hat{Y}_2 \triangleq Y_2 + \hat{Z}_2,
\end{equation}

\noindent where the distortion $\hat{Z}_2 \sim \mathsf{CN}(0, \Delta)$ with

\begin{equation}\label{delt}
\Delta = I_{N_2}+\rho_{22}H_{22}Q_{2p}H^{\dagger}_{22}.
\end{equation}

\noindent $\mathsf{D}_2$ then sends the bin index to $\mathsf{D}_1$. The rate loss due to this quantization, $\xi$, is given as
\begin{eqnarray}
\xi &\triangleq& I(\hat{Y_2};Y_2|X_{1c},X_1,X_{2c},Y_1)\nonumber
\end{eqnarray}

\begin{eqnarray}\label{xi}
&=&h\left(\sqrt{{\rho}_{22}}H_{22}X_{2p}+Z_2+\hat{Z}_2|\sqrt{{\rho}_{21}}H_{21}X_{2p}+Z_1\right)-h\left(\hat{Z}_2\right)\nonumber\\
&=&h\left(\sqrt{{\rho}_{22}}H_{22}X_{2p}+Z_2+\hat{Z}_2,\sqrt{{\rho}_{21}}H_{21}X_{2p}+Z_1\right)-h\left(\sqrt{{\rho}_{21}}H_{21}X_{2p}+Z_1\right)-h\left(\hat{Z}_2\right)\nonumber\\
&=& \log \det
\left[\begin{array}{cc}
I_{N_2}+\Delta+{\rho}_{22}H_{22}Q_{{2p}}H^{\dagger}_{22} & \sqrt{{\rho}_{21}{\rho}_{22}}H_{22}Q_{{2p}}H^{\dagger}_{21}\\
\sqrt{{\rho}_{21}{\rho}_{22}}H_{ji}Q_{{2p}}H^{\dagger}_{22} & I_{N_1}+{\rho}_{21}H_{21}Q_{{2p}}H^{\dagger}_{21}\end{array}
\right]\nonumber\\
&&-\log \det\left(I_{N_1}+{\rho}_{21}H_{21}Q_{{2p}}H^{\dagger}_{21}\right)-\log\det\left(\Delta\right)\nonumber\\
&=& \log \det
\left(I_{N_2}+\Delta+{\rho}_{22}H_{22}Q_{{2p}}H^{\dagger}_{22}
-\sqrt{{\rho}_{21}{\rho}_{22}}H_{22}Q_{{2p}}H^{\dagger}_{21}
{\left(I_{N_1}+{\rho}_{21}H_{21}Q_{{2p}}H^{\dagger}_{21}\right)^{-1}}\right.\nonumber\\
&&\left.\sqrt{{\rho}_{21}{\rho}_{22}}H_{21}Q_{{2p}}H^{\dagger}_{22}
\right)-\log\det\left(\Delta\right)\nonumber\\
&=& \log \det
\left(I_{N_2}+\Delta+
{\rho}_{22}H_{22}
\left(Q_{{2p}}-{\rho}_{21}Q_{{2p}}H^{\dagger}_{21}
{\left(I_{N_1}+{\rho}_{21}H_{21}Q_{{2p}}H^{\dagger}_{21}\right)^{-1}}
H_{21}Q_{{2p}}\right)
H^{\dagger}_{22}
\right)\nonumber\\
&&-\log\det\left(\Delta\right)\nonumber\\
&\le& \log \det
\left(I_{N_2}+\Delta+
{\rho}_{22}H_{22}
Q_{{2p}}
H^{\dagger}_{22}
\right)-\log\det\left(\Delta\right)\nonumber\\
&=& \log \det\left(2\Delta\right)-\log\det\left(\Delta\right)\nonumber\\
&=& N_{2}.
\end{eqnarray}

Thus, we see that the rate loss $\xi$ is upper bounded by the constant $N_{2}$. That is, replacing $\hat{Y_2}$ by $Y_2$ incurs at most $N_{2}$ bits.

{\em Remark:} The distortion specified in \eqref{delt} may not be optimal. The achievable rates can be further improved if we optimize over all possible distortions. For instance, if the cooperative link capacity is relatively large, we could lower the distortion level to achieve a better description of the received signals. With the expression of $\Delta$ in \eqref{delt}, however, we can show that the achievable rate region is within a constant number of bits to the capacity region for any channel parameters.

\begin{figure}[htbp]
\centering
	\includegraphics[width=8.5cm]{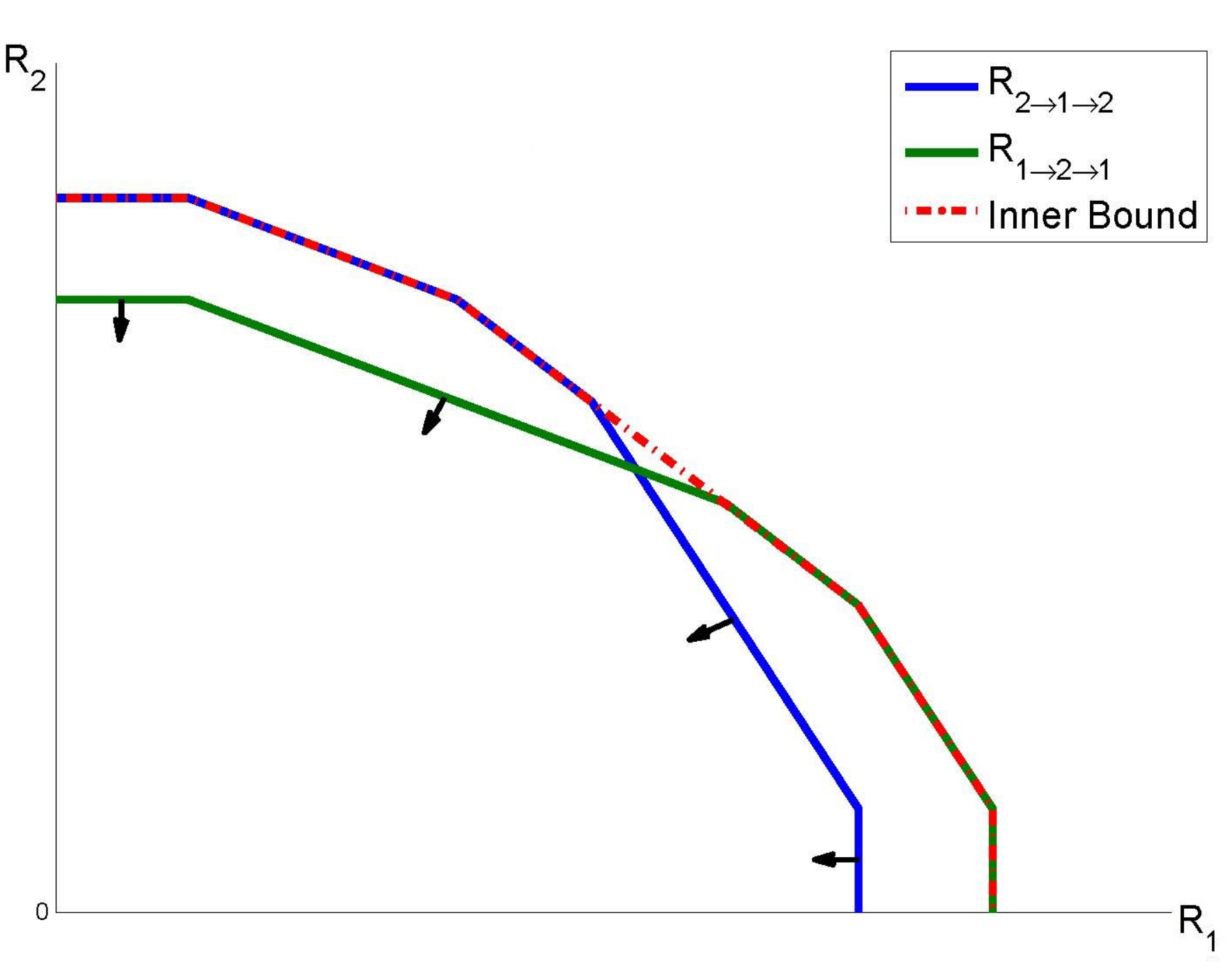}
    \label{fig:subfig1}
\caption{Time sharing of two regions ${\mathcal R}_{2\rightarrow 1\rightarrow 2}$ and ${\mathcal R}_{1\rightarrow 2\rightarrow 1}$. The four lines with arrow marks indicate that the corresponding bounds are not active when the convex hull of the two regions is taken.}
\label{fig:Example111}
\end{figure}

Considering the convex hull of the union of the achievable rate regions by the strategies $STG_{2\rightarrow 1\rightarrow 2}$ and $STG_{1\rightarrow 2\rightarrow 1}$ for MIMO IC, we show in Appendix B that we can get the achievable rate region for the general MIMO IC. Moreover, we will show in Appendix B that two of the bounds in each region will not play a role in the convex hull. This is because if any of these bounds is active, the bound on $R_1+R_2$ is active and thus following the arguments in \cite{Tse} we get that these bounds will not be active when we take the convex hull of the two regions. This is illustrated in Figure \ref{fig:Example111}, where it is seen that two of the bounds in each region are not dominant when a convex hull of the regions is taken.

Having considered the inner and outer bounds for the capacity region of the two-user IC with limited receiver cooperation, we have shown that the inner bound and the outer bound are within $N_1+N_2$ bits, thus finding the capacity region of the two-user IC with limited receiver cooperation, approximately.
%

\begin{figure}[htbp]
\centering
\subfigure[Weak Interference]{
	\includegraphics[width=3.1in]{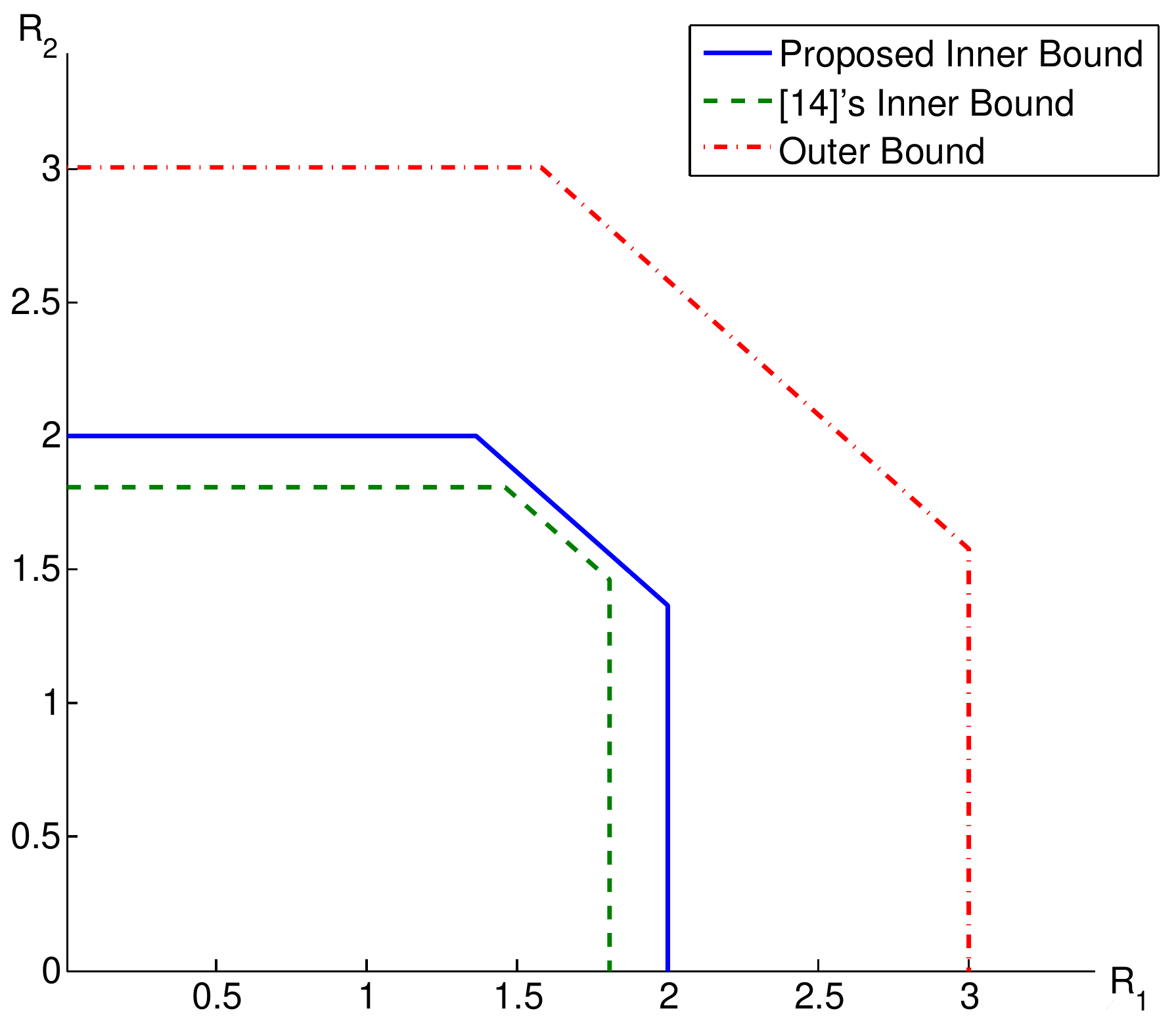}
    \label{fig:subfig3}
}
\subfigure[Strong Interference]{
	\includegraphics[width=3.1in]{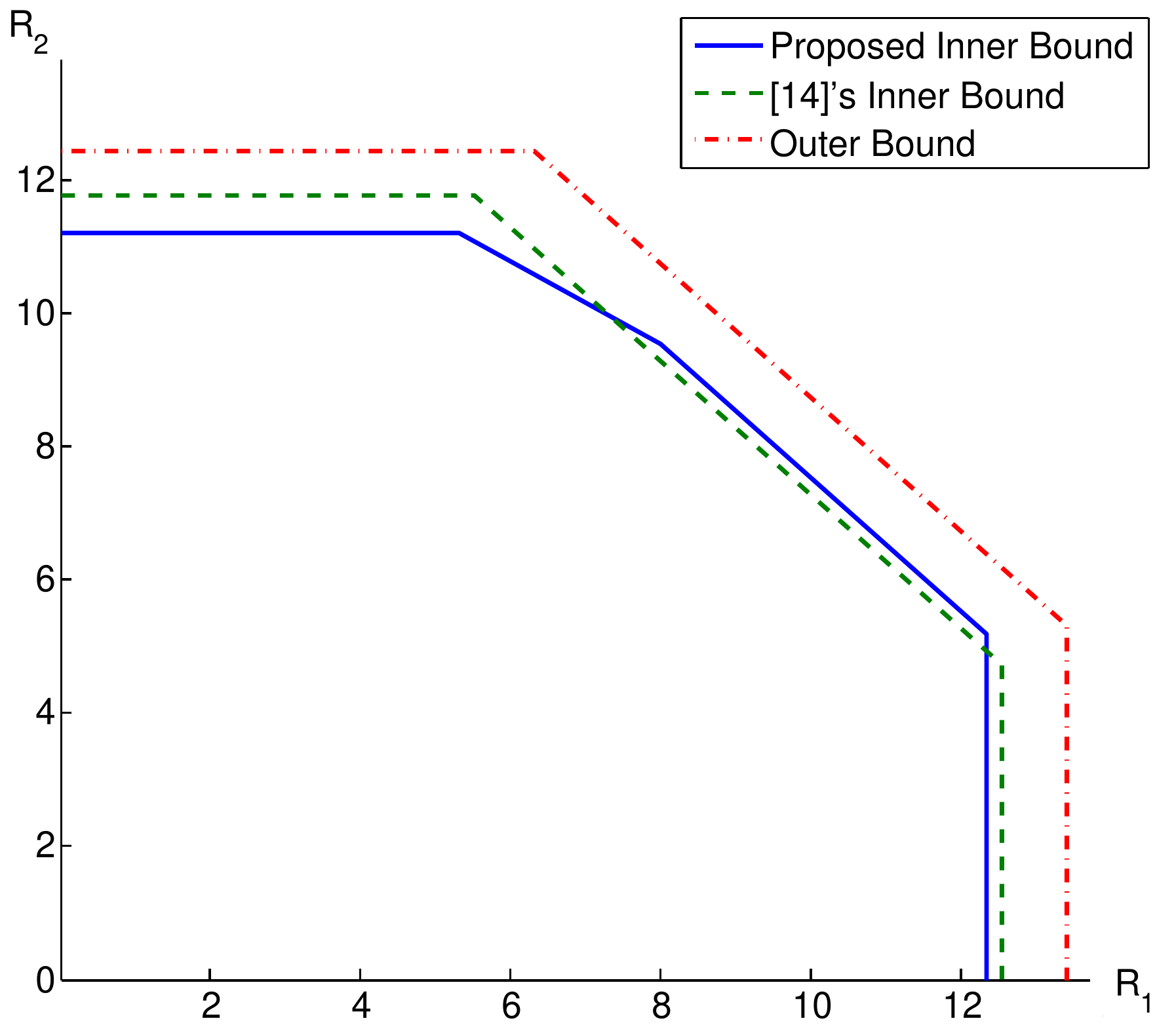}
    \label{fig:subfig2}
}
\subfigure[Mixed Interference]{
	\includegraphics[width=3.1in]{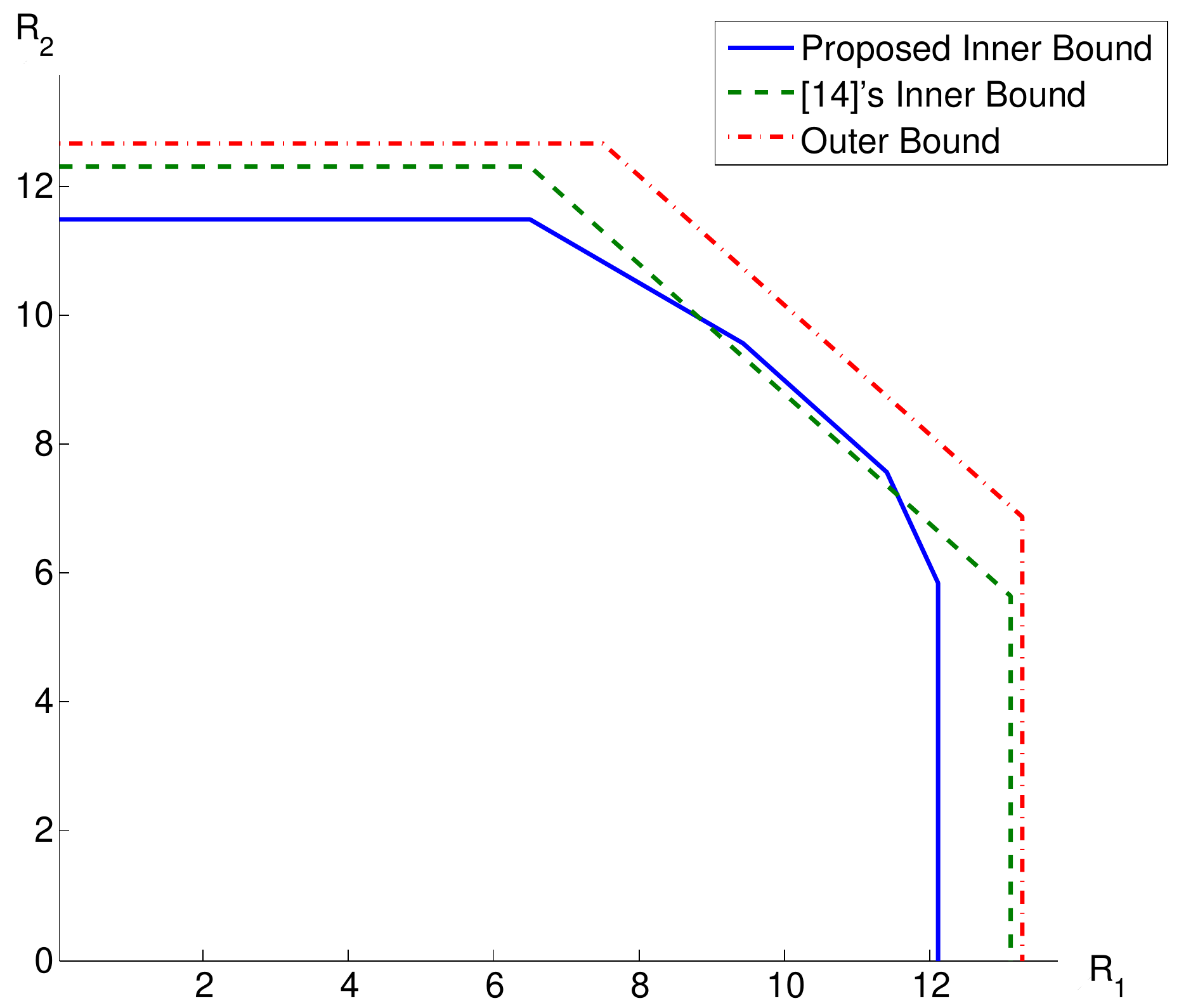}
    \label{fig:subfig1}
}
\caption{Comparison of the inner bound in this paper with that in \cite{Tse} for SISO interference channels. }
\label{fig:Example11}
\end{figure}

The authors of \cite{Tse} found the capacity region for the SISO IC with limited receiver cooperation within 2 bits. Theorem 1 generalizes the result to find the capacity region of MIMO IC with limited receiver cooperation within $N_1+N_2$ bits.

In Figure \ref{fig:Example11}, we compare the  inner bound given in Section V of \cite{Tse} for a SISO IC to that obtained in Lemma \ref{achiv} through some numerical examples. Since the outer bounds are the same for SISO, we only plot one outer bound. Let ${\mathsf {SNR}}_{i}=\rho_{ii}|H_{ii}|^2$ and ${\mathsf {INR}}_{i}=\rho_{ji}|H_{ji}|^2$ for $j\ne i$. In Figure 2(a), we consider a weak interference regime (${\mathsf {SNR}}_{1}\ge{\mathsf {INR}}_{2}$ and ${\mathsf {SNR}}_{2}\ge{\mathsf {INR}}_{1}$) with $C_{21}=1.1$, $C_{12}=1.1$, ${\mathsf {SNR}}_{1}=5$, ${\mathsf {SNR}}_{2}=5$, ${\mathsf {INR}}_{1}=2$ and ${\mathsf {INR}}_{2}=2$. In Figure 2(b), we consider strong interference regime (${\mathsf {SNR}}_{1}\le{\mathsf {INR}}_{2}$ and ${\mathsf {SNR}}_{2}\le{\mathsf {INR}}_{1}$) with $C_{21}=6$, $C_{12}=11$, ${\mathsf {SNR}}_{1}=1000$, ${\mathsf {SNR}}_{2}=1500$, ${\mathsf {INR}}_{1}=4000$ and ${\mathsf {INR}}_{2}=10000$. In Figure 2(c), we consider a  mixed interference regime (${\mathsf {SNR}}_{1}\ge{\mathsf {INR}}_{2}$ and ${\mathsf {SNR}}_{2}\le{\mathsf {INR}}_{1}$) with $C_{21}=6$, $C_{12}=11$, ${\mathsf {SNR}}_{1}=9000$, ${\mathsf {SNR}}_{2}=1500$, ${\mathsf {INR}}_{1}=5000$ and ${\mathsf {INR}}_{2}=1000$. We see from Figure \ref{fig:Example11} that the inner bounds are comparable. In the above example for weak interference channel, the strategy in this paper gives better achievable region than that in Section V.C. of \cite{Tse}.

\begin{figure}[htbp]
\centering{
	\includegraphics[width=3in]{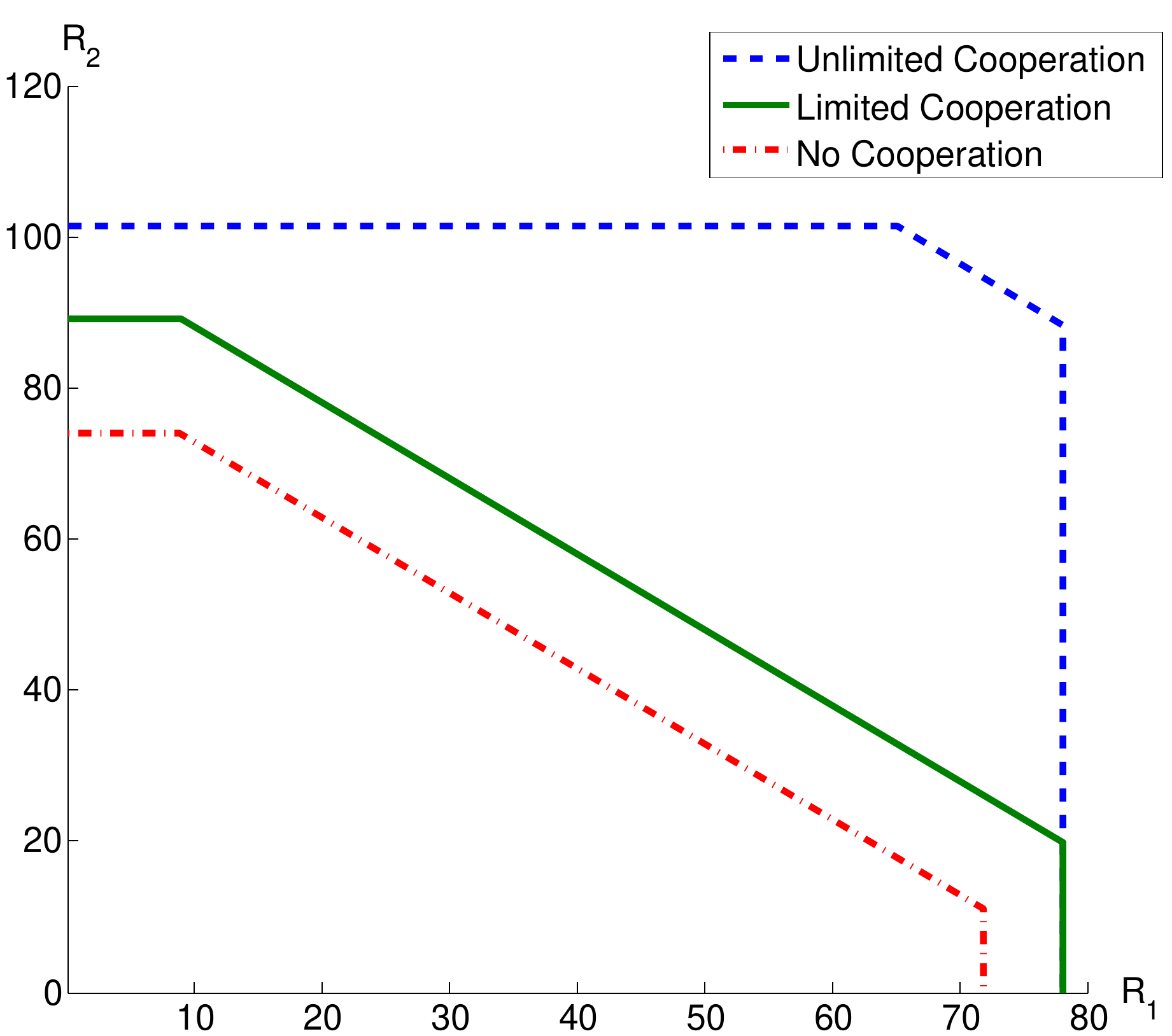}
    \label{fig:subfig3}
}
\caption{The outer bounds for a MIMO IC with no cooperation, limited cooperation and unlimited cooperation.}
\label{fig:Example12}
\end{figure}

In Figure \ref{fig:Example12}, we see the improvement in the capacity region (outer bound) for a MIMO IC with limited receiver cooperation. The parameters chosen for limited cooperation are $M_1=N_2=3$, $M_2=N_1=4$, $\rho_{11}=\rho_{22}=\rho_{12}=\rho_{21}=10^8$, $C_{21}=21$, $C_{12}=15$,
\begin{eqnarray}
H_{11}=\left[\begin{array}{ccc}
0.3096 & 0.1974 & 0.1080\\
0.3066 & 0.4470 & 0.3885\\
0.3595 & 0.6582 & 0.9854\\
0.4595 & 0.6582 & 0.4566
\end{array}\right],
H_{22}=\left[\begin{array}{cccc}
0.9070 & 0.6690 & 0.6854 & 0.6565\\
0.6067 & 0.9480 & 0.6585 & 0.6645\\
0.4465 & 0.6167 & 0.6845 & 0.3685
\end{array}\right],\nonumber\\
H_{21}=\left[\begin{array}{cccc}
0.8660 & 0.9767 & 0.4595 & 0.6582\\
0.8603 & 0.5850 & 0.6582 & 0.9854\\
0.3066 & 0.4470 & 0.6585 & 0.3885\\
0.3066 & 0.6167 & 0.4470 & 0.3885
\end{array}\right], \text{ and }
H_{12}=\left[\begin{array}{ccc}
0.1890 & 0.7650 & 0.3864\\
0.6678 & 0.2880 & 0.3867\\
0.4886 & 0.7904 & 0.2684
\end{array}\right].
\end{eqnarray}


%
%
%
%

\section{DoF and GDoF Regions}
In this section, we will use the DoF and GDoF regions to characterize the capacity region of the MIMO IC with limited receiver cooperation in the limit of high $\mathsf{SNR}$. We first describe our results on the DoF region of the two-user MIMO IC with limited receiver cooperation, and then proceed to the results on GDoF.

\subsection{DoF Region}
The DoF characterizes the simultaneously accessible fractions of spatial and signal-level dimensions (per channel use) by the two users when all the average channel parameters are an exponent of a nominal $\mathsf{SNR}$ parameter. Thus, we assume that
\begin{eqnarray}
&&{\mathop{{\lim}}_{{\log \mathsf{SNR}\to\infty}} \frac{{\log {C }_{ij}}}{{\log \mathsf{SNR}}}\ }={\beta_{ij}}, \text{ and }\nonumber\\
&&{\mathop{{\lim}}_{{\log \mathsf{SNR}\to\infty}} \frac{{\log {\rho }_{ij}}}{{\log \mathsf{SNR}}}}=1,
\end{eqnarray}
where ${\beta_{12}},{\beta_{21}}\in {\mathbb R}^+$.

The DoF region is defined as the region formed by the set of all $(d_1,d_2)$ such that $(d_1 \log\mathsf{SNR}-o(\log\mathsf{SNR}), d_2 \log\mathsf{SNR}-o(\log\mathsf{SNR}))$\footnote{$a=o(\mathsf{\log SNR})$ indicates that $\lim_{\mathsf{SNR}\to\infty}\frac{a}{\log\mathsf{SNR}}=0$.} is inside the capacity region. Further, the DoF is the maximum $d$ such that $(d,d)$ is in the DoF region. We note that since the channel matrices are of full ranks with probability 1, we will have the DoF and GDoF (next subsection) regions with probability 1 over the randomness of channel matrices.

In this subsection, we find the DoF region for the two-user MIMO IC with limited receiver cooperation. We use the approximate capacity region characterization in Theorem 1 to get the DoF region for the two-user MIMO IC as follows.

\begin{theorem}\label{thm_dof}
The DoF region for a general MIMO IC with limited receiver cooperation is given as follows:
\begin{eqnarray}
d_1&\le& {\min \left(M_1,N_1\right)}+\min  \{
\min\{N_2,(M_1-N_1)^+\},\beta_{21}\},\label{dofeq1}\\
d_2&\le& {\min \left(M_2,N_2\right)}+\min  \{
\min\{N_1,(M_2-N_2)^+\},\beta_{12}\},\label{dofeq2}\\
d_1+d_2 &\le&
\min\{N_1,(M_1-N_2)^++M_2\}
+\min\{N_2,(M_2-N_1)^++M_1\}
+\beta_{12}
+\beta_{21},\label{dofeq3}\\
d_1+d_2 &\le&
\min\{N_1,(M_1-N_2)^+\}
+\min\{N_2,M_1+M_2\}
+\beta_{12},\label{dofeq4}\\
d_1+d_2 &\le&
\min\{N_2,(M_2-N_1)^+\}
+\min\{N_1,M_1+M_2\}
+\beta_{21},\label{dofeq5}\\
d_1+d_2 &\le&
\min\{N_1+N_2,M_1+M_2\}
,\label{dofeqx}\\
2d_1+d_2 &\le&
\min\{N_2,(M_2-N_1)^++M_1\}+\min\{N_1,(M_1-N_2)^+\}+\nonumber\\
&&\min\{N_1,M_1+M_2\}+\beta_{12}+\beta_{21}, \label{dofeq6}\\
d_1+2d_2 &\le&
\min\{N_1,(M_1-N_2)^++M_2\}+\min\{N_2,(M_2-N_1)^+\}+\nonumber\\
&&\min\{N_2,M_2+M_1\}
+\beta_{12}+\beta_{21}, \label{dofeq7}\\
2d_1+d_2 &\le&
\min\{N_1+N_2,M_1\}+\min\{N_1,M_1+M_2\}+
\beta_{21}, \label{dofeqy}\\
d_1+2d_2 &\le&
\min\{N_1+N_2,M_2\}+\min\{N_2,M_1+M_2\}+
\beta_{12}. \label{dofeqz}
\end{eqnarray}
\end{theorem}

\begin{proof}
The proof can be found in Appendix C.
\end{proof}



\begin{corollary}\label{coll_gdof}
The symmetric DoF region where $\beta_{12}=\beta_{21}=\beta$, $N_1=N_2=N$, and $M_1=M_2=M$, is given as follows:

For $M\le N$:
\begin{eqnarray}\label{e1.1}
d_1&\le& M,\nonumber\\
d_2&\le& M,\nonumber\\
d_1+d_2&\le& N+\beta;
\end{eqnarray}

For $2N\le M$:
\begin{eqnarray}\label{e1.2}
d_1&\le& N+\beta,\nonumber\\
d_2&\le& N+\beta,\nonumber\\
d_1+d_2&\le& 2N;
\end{eqnarray}

For $N\le M\le 2N$:
\begin{eqnarray}
d_1&\le& \min\{M,N+\beta\},\nonumber\\
d_2&\le& \min\{M,N+\beta\},\nonumber\\
d_1+d_2&\le& \min\{M+\beta,2N\},\nonumber\\
2d_1+d_2&\le& N+M+\beta,\nonumber\\
d_1+2d_2&\le& N+M+\beta.
\end{eqnarray}
\end{corollary}

\begin{figure}[ht]
\centering
\subfigure[$M\le N$.]{
	\includegraphics[width=3.1in]{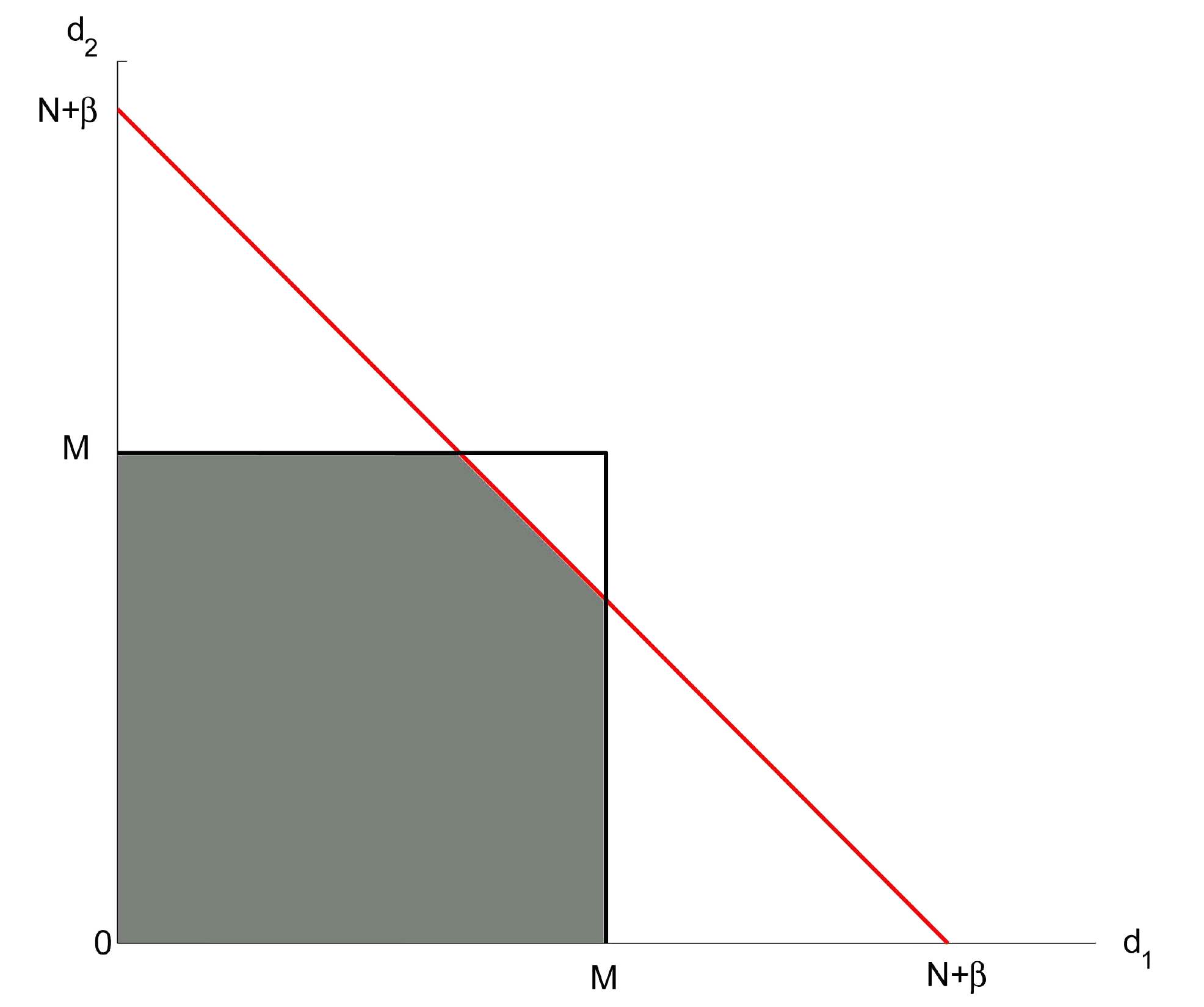}
    \label{fig:subfig1}
}
\subfigure[$2N\le M$.]{
	\includegraphics[width=3.1in]{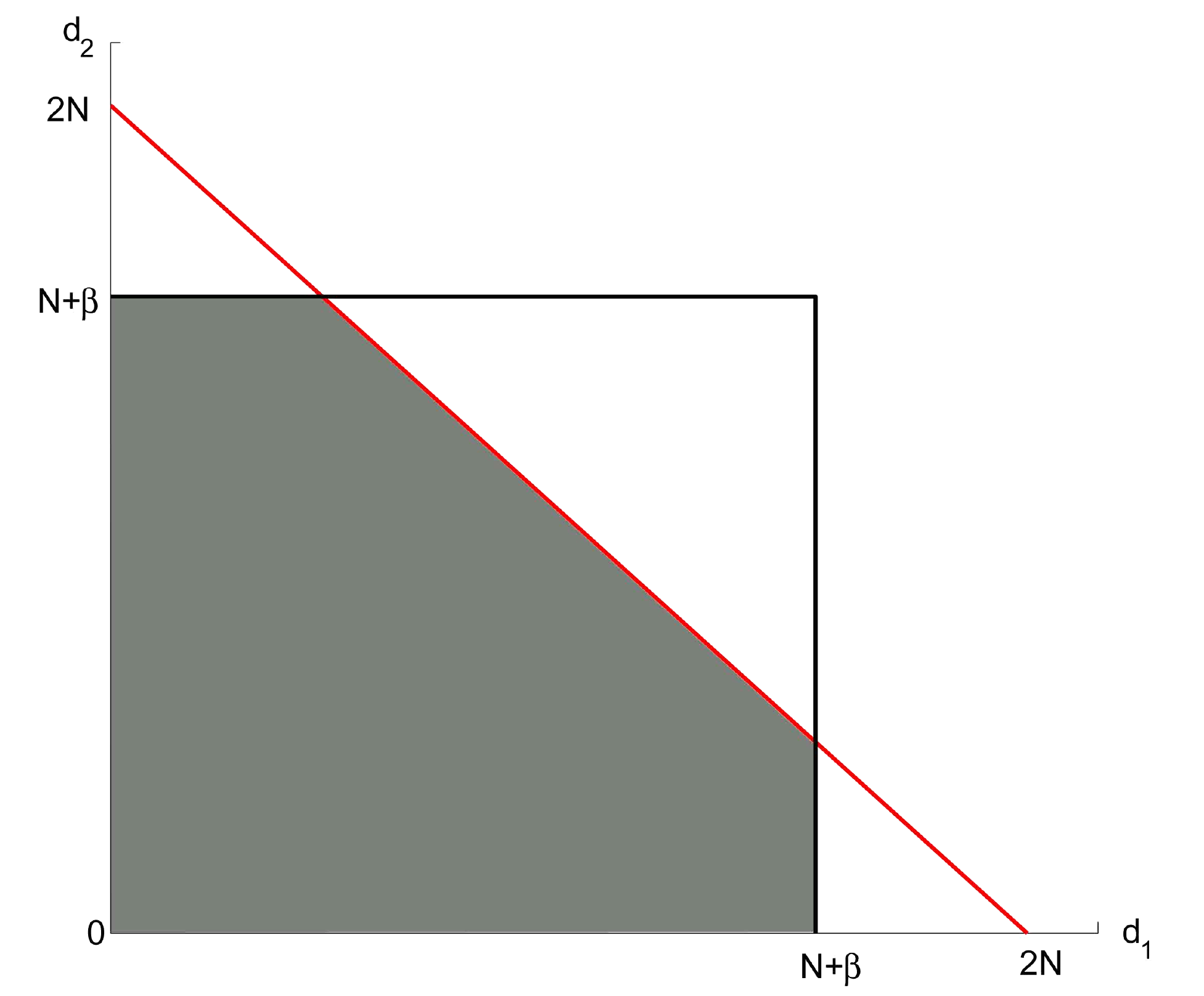}
    \label{fig:subfig2}
}
\subfigure[$N\le M\le 2N$, where $D=\min(M,N+\beta)$, $E=\frac{M+N+\beta}{2}$, and $F=\min(M+\beta,2N)$.]{
	\includegraphics[width=3.1in]{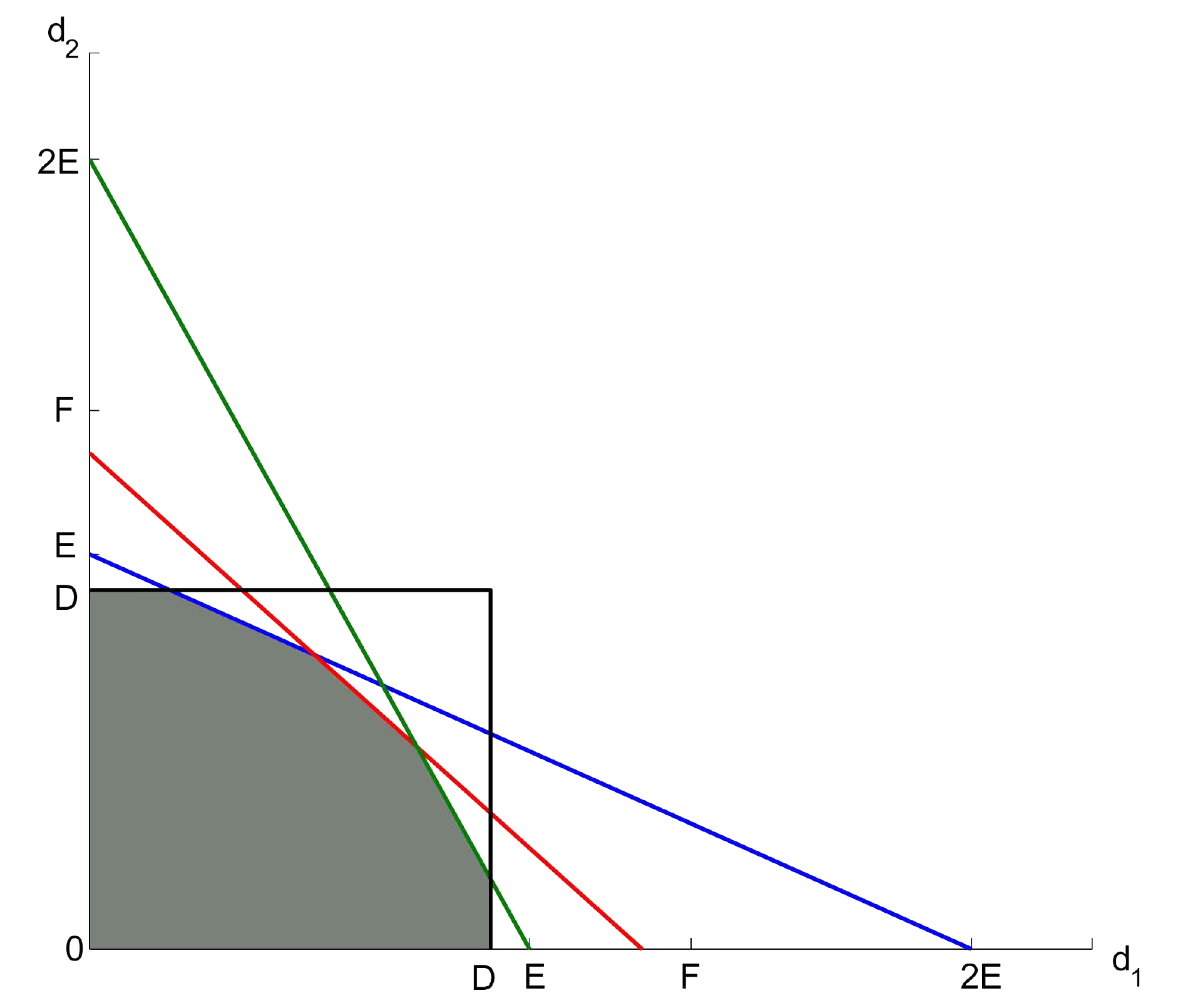}
    \label{fig:subfig3}
}
\caption[Optional caption for list of figures]{The DoF region for symmetric MIMO IC with limited receiver cooperation (grey areas).}
\label{fig:Example}
\end{figure}

These three cases are illustrated in Figure \ref{fig:Example}.

\begin{corollary}
For the symmetric DoF region where $\beta_{12}=\beta_{21}=\beta$, $N_1=N_2=N$, and $M_1=M_2=M$, cooperation improve the DoF region for $\beta \le \min\{N,(2M-N)^+\}$.
\end{corollary}
\begin{proof}
For $M\le N$ it can be seen from \eqref{e1.1} that the cooperation improves the DoF region for $\beta\le (2M-N)^+=\min\{N,(2M-N)^+\}$.

Also, for $2N\le M$ it can be seen from \eqref{e1.2} that the cooperation improves the DoF region for $\beta \le N=\min\{N,(2M-N)^+\}$.

For $N\le M\le 2N$, we consider the following four cases.

Case 1 - $\beta\le M-N$, $\beta \le 2N-M$: In this case, the symmetric DoF region reduces to
\begin{eqnarray}
d_1&\le& N+\beta,\nonumber\\
d_2&\le& N+\beta,\nonumber\\
d_1+d_2&\le& \beta,\nonumber\\
2d_1+d_2&\le& N+M+\beta,\nonumber\\
d_1+2d_2&\le& N+M+\beta.
\end{eqnarray}
In this region, $\beta$ is always less than $\min\{N,(2M-N)^+\}$ because $\beta\le M-N\le N=\min\{N,(2M-N)^+\}$. Hence increasing $\beta$ always enlarges the region.

Case 2 - $\beta \ge M-N$, $\beta\le 2N-M$: In this case, the symmetric DoF region reduces to
\begin{eqnarray}
d_1&\le& M,\nonumber\\
d_2&\le& M,\nonumber\\
d_1+d_2&\le& M+\beta,\nonumber\\
2d_1+d_2&\le& N+M+\beta,\nonumber\\
d_1+2d_2&\le& N+M+\beta.
\end{eqnarray}
In this region, $\beta$ is always less than $\min\{N,(2M-N)^+\}$ because $\beta\le 2N-M\le N=\min\{N,(2M-N)^+\}$. In this case, increasing $\beta$ always enlarges the region. According to Figure 3(c), while $\beta\le 2N-M$, we get $2E\le 3N$ and $F\le 2N$ which indicates none of the red, green and blue lines could include the point $(d_1,d_2)=(M,M)$ below them. Also, increasing $\beta$ leads to the increase of $E$ and $F$ in Figure 3(c) and as a result, enlarges the symmetric DoF region.

Case 3 - $\beta \le M-N$, $\beta \ge 2N-M$: In this case, the symmetric DoF region reduces to
\begin{eqnarray}
d_1&\le& N+\beta,\nonumber\\
d_2&\le& N+\beta,\nonumber\\
d_1+d_2&\le& 2N,\nonumber\\
2d_1+d_2&\le& N+M+\beta,\nonumber\\
d_1+2d_2&\le& N+M+\beta.
\end{eqnarray}
In this region, $\beta$ is always less than $\min\{N,(2M-N)^+\}$ because $\beta\le M-N\le N=\min\{N,(2M-N)^+\}$. In this case, increasing $\beta$ always enlarges the region. According to Figure 3(c), when $\beta\le M-N$, we get $D,E\le M\le 2N=F$ and also, increasing $\beta$ leads to the increase of $D$ and $E$ in Figure 3(c) and as a result, enlarges the symmetric DoF region.

Case 4 - $\beta \ge M-N$, $\beta\ge 2N-M$: In this case, the symmetric DoF region reduces to
\begin{eqnarray}
d_1&\le& M,\nonumber\\
d_2&\le& M,\nonumber\\
d_1+d_2&\le& 2N,\nonumber\\
2d_1+d_2&\le& N+M+\beta,\nonumber\\
d_1+2d_2&\le& N+M+\beta.
\end{eqnarray}
In this region, changing $\beta$ only changes $E$ in Figure 3(c). Also, we can easily see that the black line and red line intersects at $(d_1,d_2)=(M,2N-M)$. The green line includes this intersection point when $\beta\ge N$ and will be below this point when $\beta\le N$ which means increasing $\beta$ improves the DoF region until $\beta\le N=\min\{N,(2M-N)^+\}$.
\end{proof}



\subsection{GDoF Region}
The notion of GDoF generalizes the DoF metric by additionally emphasizing the signal level as a signaling dimension. It characterizes the simultaneously accessible fractions of spatial and signal-level dimensions (per channel use) by the two users when all the average channel parameters vary as exponents of a nominal SNR parameter as follows
\begin{eqnarray}
&&{\mathop{{\lim}}_{{\log \mathsf{SNR}\to\infty}} \frac{{\log {C }_{ij}}}{{\log  \mathsf{SNR}}}\ }=\beta_{ij},\nonumber\\
&&{\mathop{{\lim}}_{{\log \mathsf{SNR}\to\infty}} \frac{{\log {\rho }_{ij}}}{{\log \mathsf{SNR}}}\ }=\begin{cases}
1, & \text{ if } i=j\\
\alpha, & \text{ if } i\neq j
\end{cases},
\end{eqnarray}

where ${\alpha },{\beta}_{12},{\beta}_{21}\in {\mathbb R}^+$.

The GDoF region is defined as the region formed by the set of all $(d_1,d_2)$ such that $(d_1 \log\mathsf{SNR}-o(\log\mathsf{SNR}), d_2 \log\mathsf{SNR}-o(\log\mathsf{SNR}))$ is inside the capacity region. Further, the GDoF is the maximum $d$ such that $(d,d)$ is in the GDoF region. Thus, both the GDoF region and GDoF are functions of link quality scaling exponent $\alpha$.

Next we present our results on the GDoF region for the two-user MIMO IC with limited receiver cooperation. For the general case, the computation of GDoF region is hard and thus we will only consider the case that $M_1=M_2=N_1=N_2=M$. We also assume that $\beta_{21}=\beta_{12}=\beta$. With these assumptions, the GDoF region for the two user MIMO IC with limited receiver cooperation is given in the following Theorem.




\begin{theorem}\label{thm_gdof}
The GDoF region for a two-user symmetric MIMO IC with limited receiver cooperation is equivalent to the convex hull of the:
\begin{eqnarray}
d_1&\le& M+ \min  \{({\alpha}-{1})^+M,\beta\},\label{gdofeq1}\\
d_2&\le& M+ \min  \{({\alpha}-{1})^+M,\beta\},\label{gdofeq2}\\
d_1+d_2&\le& 2M\max\{({1}-{\alpha})^+,{\alpha}\}+2\beta,\label{gdofeq3}\\
d_1+d_2&\le& ({1}-{\alpha})^+M + M\max\{{1},{\alpha}\}+\beta,\label{gdofeq4}\\
d_1+d_2&\le&
{2M}\max\{{1},{\alpha}\},\label{gdofeq6}
\end{eqnarray}

\begin{eqnarray}
d_1+2d_2 &\le&
M\max\{({1}-{\alpha})^+,{\alpha}\}+\nonumber\\
&&({1}-{\alpha})^+M
+M\max\{{1},{\alpha}\}+2\beta,\label{gdofeq7}\\
2d_1+d_2 &\le&
M\max\{({1}-{\alpha})^+,{\alpha}\}+\nonumber\\
&&({1}-{\alpha})^+M
+M\max\{{1},{\alpha}\}+2\beta,\label{gdofeq8}\\
d_1+2d_2 &\le&
M\max\{({2}-{\alpha})^+,{\alpha}\}+ M\max\{{1},{\alpha}\}+\beta,\label{gdofeq9}\\
2d_1+d_2 &\le&
M\max\{({2}-{\alpha})^+,{\alpha}\}+ M\max\{{1},{\alpha}\}+\beta\label{gdofeq10}.
\end{eqnarray}
\end{theorem}

\begin{proof}
The proof can be found in Appendix D.
\end{proof}

\begin{corollary}\label{coll_gdof}
The GDoF for a two-user MIMO IC  with limited receiver cooperation, when $M_1=M_2=N_1=N_2=M$ and $\beta_{21}=\beta_{12}=\beta$ is given as
\begin{eqnarray}
\text{GDOF}_{RC}&=& \min\{ M+ \min  \{({\alpha}-{1})^+M,\beta\}, M\max\{({1}-{\alpha})^+,{\alpha}\}+\beta,\nonumber\\
&& \frac{1}{2}({1}-{\alpha})^+M + \frac{1}{2}M\max\{{1},{\alpha}\}+\frac{1}{2}\beta,
{M}\max\{{1},{\alpha}\},\nonumber\\
&&\frac{1}{3}M\max\{({1}-{\alpha})^+,{\alpha}\}+\frac{1}{3}({1}-{\alpha})^+M
+\frac{1}{3}M\max\{{1},{\alpha}\}+\frac{2}{3}\beta,\nonumber\\
&&\frac{1}{3}M\max\{({2}-{\alpha})^+,{\alpha}\}+ \frac{1}{3}M\max\{{1},{\alpha}\}+\frac{1}{3}\beta\}. \label{eq_gdof}
\end{eqnarray}
\end{corollary}

Since the GDoF in Corollary \ref{coll_gdof} is the minimum of many terms, we evaluate the minimum in \eqref{eq_gdof} to reduce the expression of GDoF as follows.

For $0\le \beta \le \frac{M}{2}$:

\begin{eqnarray}
\text{GDoF}_{RC}=\left\{ \begin{array}{ll}
M,& \text{ if }0\le \alpha \le \frac{\beta}{M}, \\
M{(1-\alpha)^+}+\beta,& \text{ if }\frac{\beta}{M}\le \alpha \le \frac{1}{2}, \\
M{\alpha}+\beta,& \text{ if }\frac{1}{2}\le \alpha \le {{\frac{2}{3}}-{\frac{\beta}{3M}}}, \\
{\frac{1}{2}}(M{(2-\alpha)^+}+\beta),& \text{ if }{{\frac{2}{3}}-{\frac{\beta}{3M}}}\le \alpha \le 1, \\
{\frac{1}{2}}(M{\alpha}+\beta),& \text{ if } 1 \le \alpha \le {2+\frac{\beta}{M}}, \\
M+\beta,& \text{ if } {2+\frac{\beta}{M}} \le \alpha.
\end{array}
\right.
\end{eqnarray}

For $\frac{M}{2}\le \beta\le M$:

\begin{eqnarray}
\text{GDoF}_{RC}=\left\{ \begin{array}{ll}
M,& \text{ if }0\le \alpha \le \frac{\beta}{M}, \\
{\frac{1}{2}}(M{(2-\alpha)^+}+\beta),& \text{ if }\frac{\beta}{M}\le \alpha \le 1, \\
{\frac{1}{2}}(M{\alpha}+\beta),& \text{ if } 1 \le \alpha \le {2+\frac{\beta}{M}}, \\
M+\beta,& \text{ if } {2+\frac{\beta}{M}} \le \alpha.
\end{array}
\right.
\end{eqnarray}

For $M\le \beta$:

\begin{eqnarray}
\text{GDoF}_{RC}=\left\{ \begin{array}{ll}
M,& \text{ if }0\le \alpha \le 1, \\
M\alpha,& \text{ if }1\le \alpha \le \frac{\beta}{M}, \\
{\frac{1}{2}}(M{\alpha}+\beta),& \text{ if } \frac{\beta}{M} \le \alpha \le {2+\frac{\beta}{M}}, \\
M+\beta,& \text{ if } {2+\frac{\beta}{M}} \le \alpha.
\end{array}
\right.
\end{eqnarray}

The authors of  \cite{Etkin} found the GDoF for the two-user symmetric MIMO IC without cooperation as follows
\begin{eqnarray}
\text{GDoF}_{NRC}=\left\{ \begin{array}{ll}
M{(1-\alpha)^+},& \text{ if }0\le \alpha \le \frac{1}{2}, \\
M{\alpha},& \text{ if }\frac{1}{2}\le \alpha \le {{\frac{2}{3}}}, \\
{\frac{1}{2}}(M{(2-\alpha)^+}),& \text{ if }{{\frac{2}{3}}}\le \alpha \le 1, \\
{\frac{1}{2}}{M{\alpha}},& \text{ if } 1 \le \alpha \le {2}, \\
M,& \text{ if } {2} \le \alpha.
\end{array}
\right.
\end{eqnarray}

\begin{figure}[ht]
\centering
\subfigure[$0\le \beta_0\le \frac{M}{2}$.]{
	\includegraphics[width=4.3in]{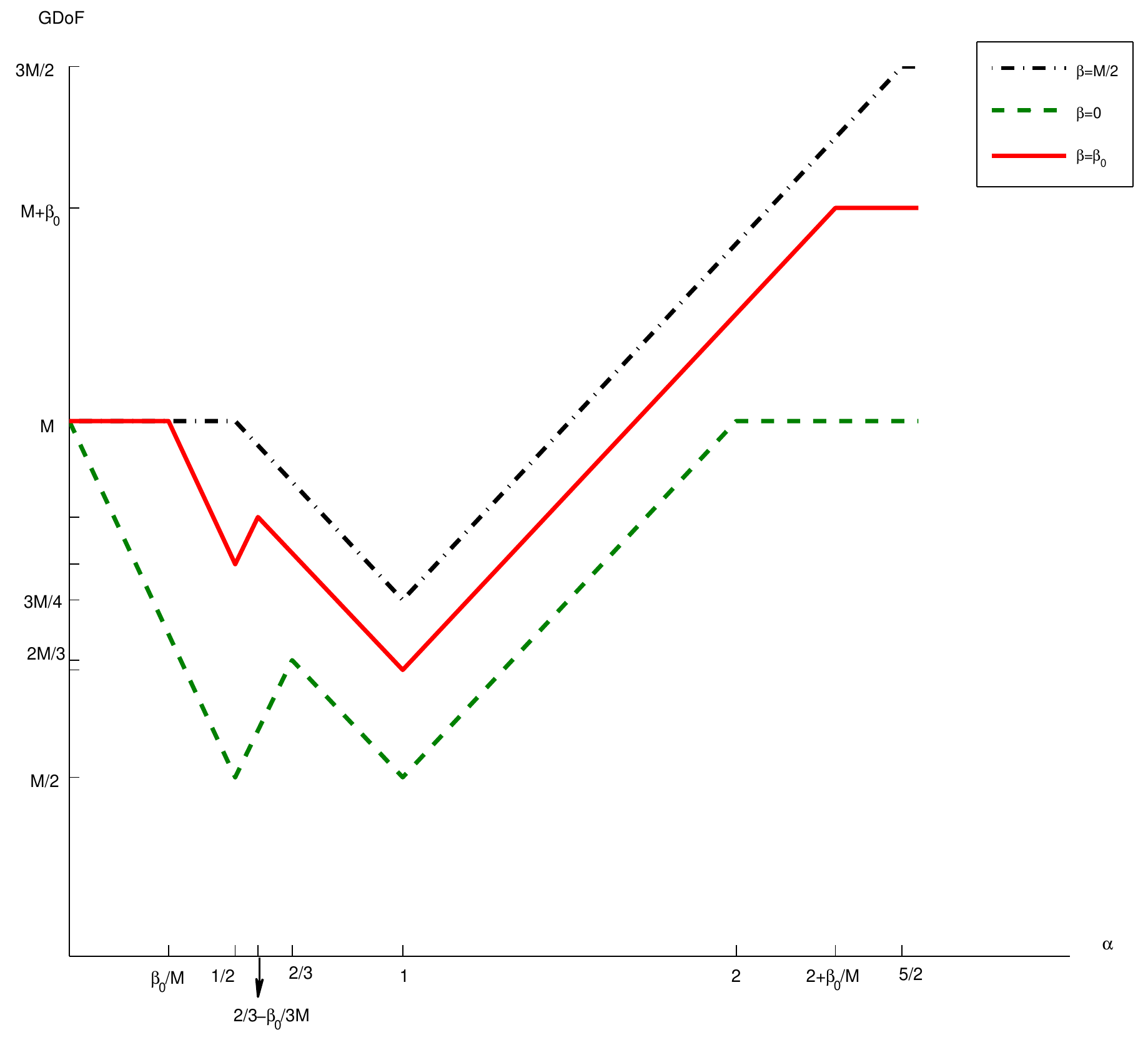}
    \label{fig:subfig1}
}
\subfigure[$\frac{M}{2}\le \beta_0\le M$.]{
	\includegraphics[width=3.1in]{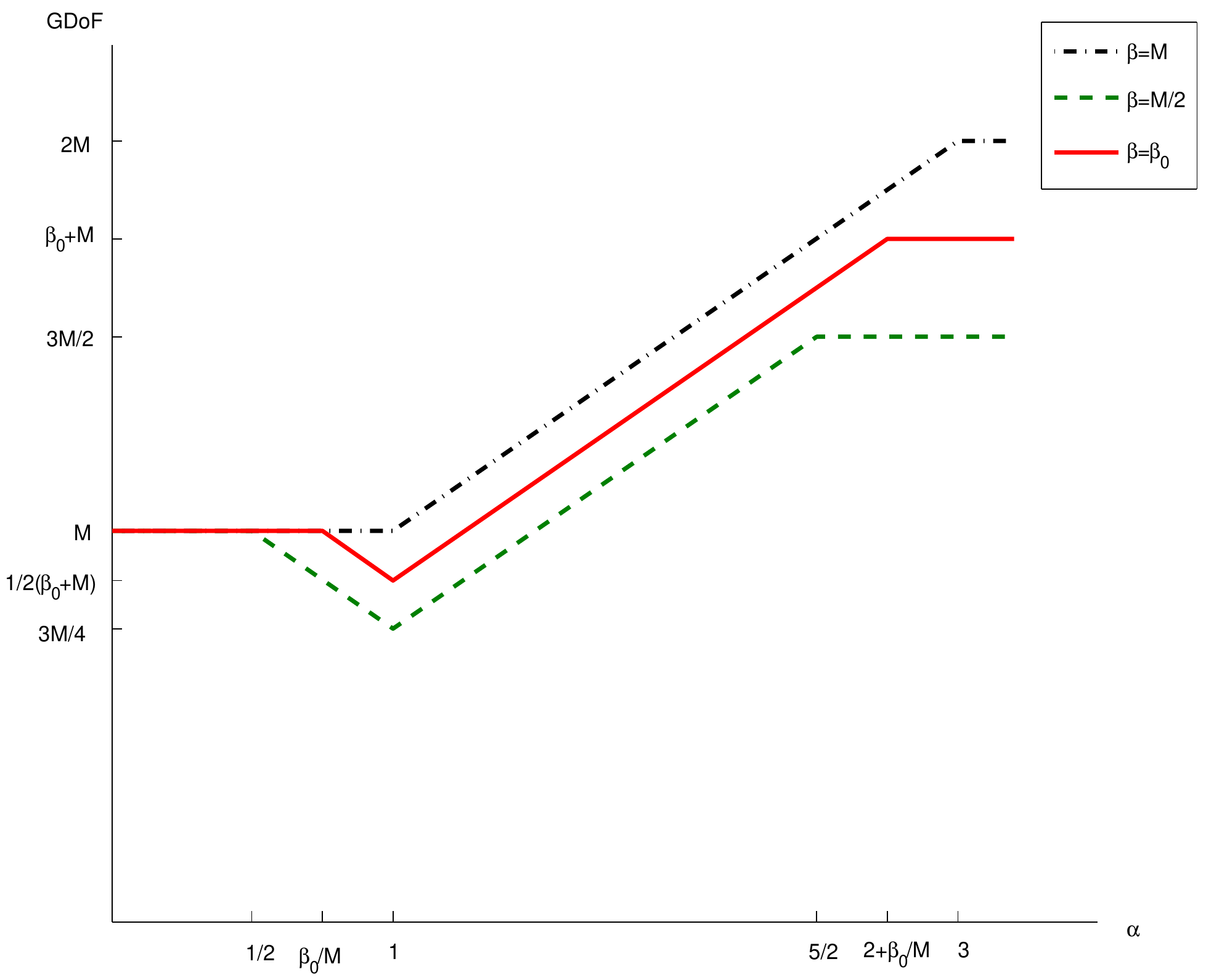}
    \label{fig:subfig2}
}
\subfigure[$\beta_0\ge M$.]{
	\includegraphics[width=3.1in]{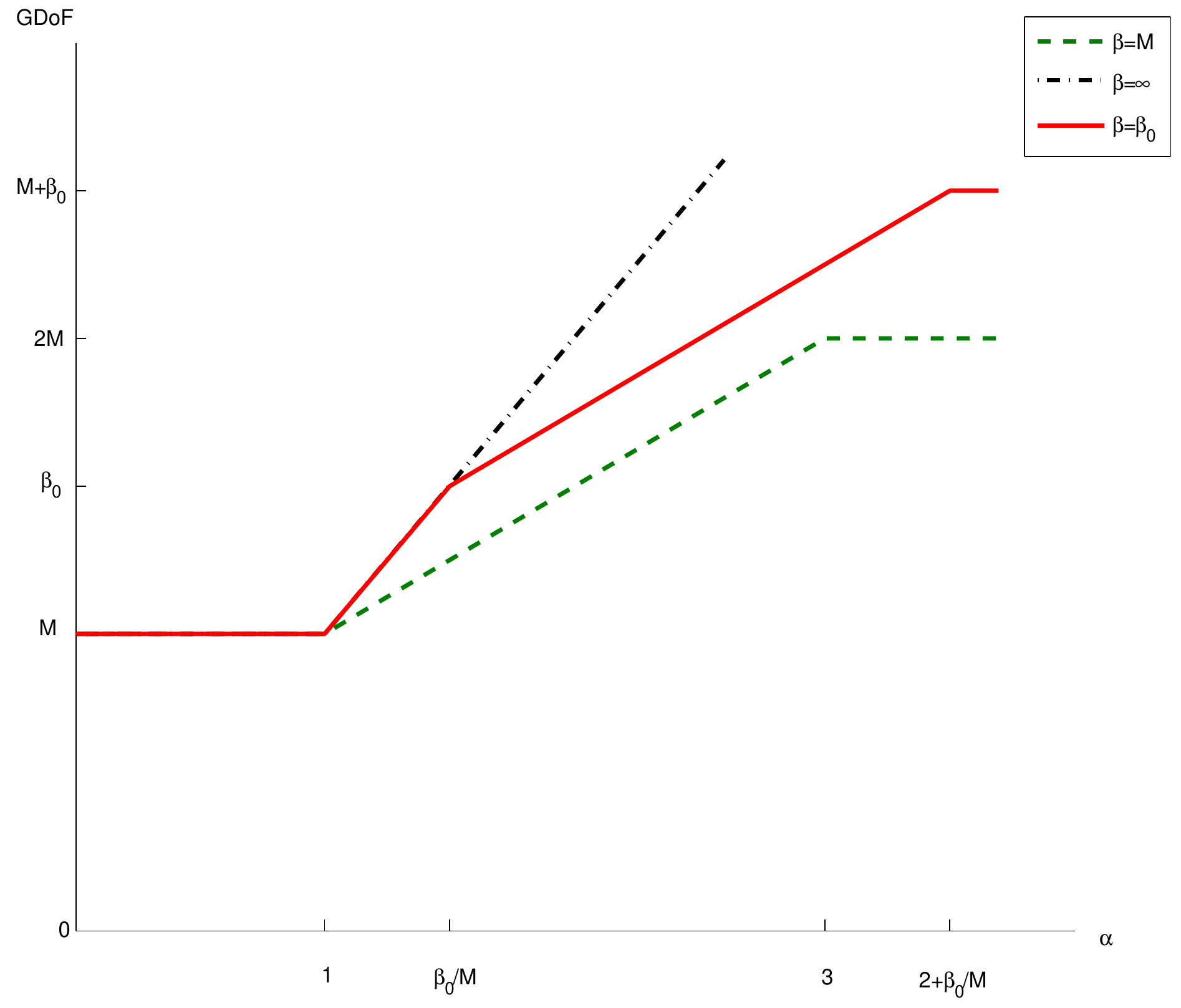}
    \label{fig:subfig3}
}
\caption[Optional caption for list of figures]{GDoF for MIMO IC with limited receiver cooperation when all nodes have the same number of antennas $M$.}
\label{fig:subfigureExample}
\end{figure}

Figure \ref{fig:subfigureExample} compares the GDoF for the two-user symmetric MIMO IC with and without receiver cooperation. In Figure \ref{fig:subfigureExample}(a), the ``W"-curve obtained without cooperation delineates the very weak ($0\le \alpha \le\frac{1}{2}$), weak ($\frac{1}{2}\le\alpha\le\frac{2}{3}$), moderate ($\frac{2}{3}\le\alpha\le1$), strong ($1\le\alpha\le 2$) and very strong ($\alpha\ge 2$) interference regimes. In the presence of weak collaboration ($0\le \beta \le{{\frac{M}{2}}}$), the ``W"-curve improves to another  ``W"-curve which delineates to extremely weak ($0\le\alpha\le\frac{\beta}{M}$), very weak ($\frac{\beta}{M}\le\alpha\le\frac{1}{2}$), weak ($\frac{1}{2}\le\alpha\le{\frac{2}{3}-\frac{\beta}{3M}}$), moderate (${\frac{2}{3}-\frac{\beta}{3M}}\le\alpha\le1$), strong ($1\le\alpha\le{2+{\frac{\beta}{M}}}$) and very strong ($2+{\frac{\beta}{M}}\le\alpha$) interference regimes. In the presence of weak collaboration ($ 0 \le \beta \le {\frac{M}{2}}$), we see that the GDoF  is strictly greater than that without collaboration for every $\alpha>0$. The GDoF improvement indicates an unbounded gap in the corresponding capacity regions as the SNR goes to infinity.

For moderate collaboration (${{\frac{M}{2}}}\le \beta\le M$), the ``W"-curve improves to a ``V"-curve which delineates to the very weak ($0\le\alpha\le\frac{\beta}{M}$), weak ($\frac{\beta}{M}\le\alpha\le 1$), strong ($1\le\alpha\le{2+{\frac{\beta}{M}}}$) and very strong ($2+{\frac{\beta}{M}}\le\alpha$) interference regimes, and we see that the GDoF with collaboration is strictly greater than that without collaboration for $\alpha>0$ similar to the weak collaboration.

For strong collaboration ($ \beta \ge M$), the ``W"-curve improves to an increasing curve which delineates to the very weak ($0\le\alpha\le 1$), weak ($1\le\alpha\le \frac{\beta}{M}$), strong ($\frac{\beta}{M}\le\alpha\le{2+{\frac{\beta}{M}}}$) and very strong ($2+{\frac{\beta}{M}}\le\alpha$) interference regimes. The slopes of increase of GDoF with $\alpha$ changes at the border of these regimes.

We note that for a given $M$ and $\alpha$, increasing $\beta$ improves the GDoF till $\beta = M\alpha$, after which there is no improvement in the GDoF since the GDoF at this point is the same as that with full cooperation. This can be seen also in the following corollary.

\begin{corollary}\label{colg}
The symmetric GDoF for a two-user MIMO IC  with limited receiver cooperation, when $M_1=M_2=N_1=N_2=M$ and $\beta_{21}=\beta_{12}=\beta = M\alpha$ is equal to $M\max(1,\alpha)$ which is the same as that with full cooperation.
\end{corollary}

\begin{proof}
We only need to compare $M\max(1,\alpha)$ with all the bounds of the Corollary \ref{coll_gdof} and see that it is smaller or equal to all of them in Corollary \ref{coll_gdof}, or
\begin{eqnarray}
M\max\{1,\alpha\} &\le& M+\min\{({\alpha}-{1})^+M,\beta\},\nonumber\\
M\max\{1,\alpha\} &\le& M\max\{({1}-{\alpha})^+,{\alpha}\}+\beta,\nonumber\\
M\max\{1,\alpha\} &\le& \frac{1}{2}({1}-{\alpha})^+M+\frac{1}{2}M\max\{{1},{\alpha}\}+\frac{1}{2}\beta,\nonumber\\
M\max\{1,\alpha\} &\le& {M}\max\{{1},{\alpha}\},\nonumber\\
M\max\{1,\alpha\} &\le& \frac{1}{3}M\max\{({1}-{\alpha})^+,{\alpha}\}+\frac{1}{3}({1}-{\alpha})^+M
+\frac{1}{3}M\max\{{1},{\alpha}\}+\frac{2}{3}\beta,\nonumber\\
M\max\{1,\alpha\} &\le& \frac{1}{3}M\max\{({2}-{\alpha})^+,{\alpha}\}+ \frac{1}{3}M\max\{{1},{\alpha}\}+\frac{1}{3}\beta.
\end{eqnarray}
Since all these expressions can be shown to hold, $M\max(1,\alpha)$ is achievable. Further, since $M\max(1,\alpha)$ is also an outer bound, the Corollary \ref{colg} holds.

\end{proof}


\section{Conclusions}
This paper characterizes the approximate capacity region of the two-user MIMO interference channels with limited receiver cooperation within $N_1+N_2$ bits. This approximate capacity region is used to find the DoF region for the two user MIMO interference channels with limited receiver cooperation. We also find the maximum amount of cooperation needed to achieve the outer bound of unlimited receiver cooperation. Further, the GDoF region is found for a two-user MIMO interference channel with equal antennas at all the nodes. With the GDoF region,  we find that the ``W'' curve without cooperation changes gradually to ``V'' curve with full cooperation. The cooperation improves the GDoF till the capacity of the cooperation link is of the order of $\alpha M \log \mathsf{SNR}$ when the GDoF reaches the GDoF with full cooperation.

Finally we note that the GDoF results for general number of transmit and receive antennas remains as an open problem.

\begin{appendices}
\section{Proof of Outer Bound for Theorem \ref{outer_inner_capacity_reciprocal}} \label{apdx_outer}

In this Appendix, we will show that ${\mathbb C}_{RC}\subseteq \mathcal{R}_o$. The set of upper bounds to the capacity region will be derived in two steps. First, the capacity region is outer-bounded by a region defined in terms of the differential entropy of the random variables associated with the signals. These outer-bounds use genie-aided information at the receivers. Second, we outer-bound this region to prove the outer-bound as described in the statement of Theorem \ref{outer_inner_capacity_reciprocal}.

The following result outer-bounds the capacity region of a two-user MIMO IC with limited receiver cooperation.

\begin{lemma}
Let $S_i$ and $\tilde{S_i}$ be defined as $S_i\triangleq \sqrt{{\rho }_{ij}}H_{ij}X_i+Z_j$ and $\tilde{S_i}\triangleq \sqrt{{\rho }_{ij}}H_{ij}X_i+\tilde{Z_j}$, respectively, where $\tilde{Z_i}\sim \mathsf{CN}(0,I_{M_i})$ is independent of everything else. Then, the capacity region of a two-user MIMO IC with limited receiver cooperation is outerbounded by the region formed by $(R_1,R_2)$ satisfying
\begin{eqnarray}
R_1&\le& h(H_{11}X_1+Z_1)-h(Z_1)+\min \{h(H_{12}X_1+Z_2|H_{11}X_1+Z_1)-h(Z_2),C_{21} \},\label{o1}\\
R_2&\le& h(H_{22}X_2+Z_2)-h(Z_2)+\min \{h(H_{21}X_2+Z_1|H_{22}X_2+Z_2)-h(Z_1),C_{12} \}\label{o2},\\
R_1+R_2&\le& h(Y_1|\tilde{S_1})+h(Y_2|\tilde{S_2})-h(\tilde{Z_1})-h(\tilde{Z_1})+C_{21}+C_{12},\label{o3}\\
R_1+R_2&\le& h(H_{11}X_1+Z_1|{S_1})+h(Y_2)-h(Z_1,Z_2)+C_{12},\label{o4}\\
R_1+R_2&\le& h(H_{22}X_2+Z_2|{S_2})+h(Y_1)-h(Z_1,Z_2)+C_{21},\label{o5}\\
R_1+R_2&\le& h(Y_1,Y_2)-h(Z_1,Z_2),\label{o6}\\
2{R_1}+R_2&\le& h(H_{11}X_1+Z_1|{S_1})+h(Y_1)+h(Y_2|{S_2})-h(Z_1,Z_2)-h({Z_1})+C_{21}+C_{12},\label{o7}\\
{R_1}+2{R_2}&\le& h(H_{22}X_2+Z_2|{S_2})+h(Y_2)+h(Y_1|{S_1})-h(Z_1,Z_2)-h({Z_2})+C_{21}+C_{12},\label{o8}\\
2{R_1}+R_2&\le& h(Y_1,Y_2|\tilde{S_2})+h(Y_1)-h(Z_1,Z_2)-h({Z_1})+C_{21},\label{o9}\\
{R_1}+2{R_2}&\le& h(Y_1,Y_2|\tilde{S_1})+h(Y_2)-h(Z_1,Z_2)-h({Z_2})+C_{12}.\label{o10}
\end{eqnarray}\label{lemma_conditional}
\end{lemma}
\begin{proof}
The proof follows the same lines as the proof of Lemma 5.1 in \cite{Tse}, replacing SISO channel gains by MIMO channel matrices and is thus omitted here.
\end{proof}

The rest of the section outer-bounds this region to get the outer bound in Theorem \ref{outer_inner_capacity_reciprocal}. For this, we will introduce some useful Lemmas.

The next result outer-bounds the entropies and the conditional entropies of two random variables by their corresponding Gaussian random variables.

\begin{lemma}[\cite{Gaussian}]
Let $X$ and $Y$ be two random vectors, and let $X^G$ and $Y^G$ be Gaussian vectors with covariance matrices satisfying
\begin{eqnarray}
Cov\left[ \begin{array}{c}
X \\
Y \end{array}
\right]=Cov\left[ \begin{array}{c}
X^G \\
Y^G \end{array}
\right],
\end{eqnarray}

Then, we have
\begin{eqnarray}
h(Y)&\le& h(Y^G),\\
h\left(Y\mathrel{\left|\vphantom{Y X}\right.\kern-\nulldelimiterspace}X\right)&\le& h\left(Y^G\mathrel{\left|\vphantom{Y^G X^G}\right.\kern-\nulldelimiterspace}X^G\right).\label{gauss_cond}
\end{eqnarray}\label{gauss_ob}
\end{lemma}
The next result gives the determinant of a block matrix, which will be used extensively in the sequel.
\begin{lemma}[\cite{det}]
For block matrix $M=\left[ \begin{array}{cc}
A & B \\
C & D \end{array}
\right]$ with submatrices A, B, C, and D, we have:
\begin{eqnarray}
\det  M= \begin{cases}\det  A \det (D-CA^{-1}B), & \text{if A is invertible,}\label{eq1_block}\\
\det D \det (A-BD^{-1}C), & \text{if D is invertible.} \label{eq2_block}
\end{cases}
\end{eqnarray}\label{lem_block}
\end{lemma}
The next result gives a monotonicity result for a function which will be used to upper bound some of the terms in Lemma \ref{lemma_conditional}.

\begin{lemma}[\cite{J1}]
Let $L(K,S)$ be defined as
\begin{eqnarray}
L\left(K, S\right)
\triangleq K-KS{(I_{N}+S^{\dagger }KS)}^{-1}S^{\dagger }K, \label{eqdefnl}
\end{eqnarray}
for some $M\times M$ p.s.d. Hermitian matrix $K$ and some $M\times N$ matrix $S$.
Then if $0 \preceq K_1\preceq K_2$ for some Hermitian matrices $K_1$ and $K_2$, we have
\begin{eqnarray}\label{tpl}
L\left(K_1, S\right)\preceq L\left(K_2, S\right)\label{ll}.
\end{eqnarray}
\label{L}
\end{lemma}

Define $X_1^G$ and $X_2^G$ as having a Gaussian distribution with the covariance matrix
\begin{eqnarray}
Cov\left[ \begin{array}{c}
X_1^G \\
X_2^G \end{array}
\right]=Cov\left[ \begin{array}{c}
X_1 \\
X_2 \end{array}
\right].
\end{eqnarray}
Define $S_i^G\triangleq \sqrt{{\rho }_{ij}}H_{ij}X_i^G+Z_j$, $\tilde{S_i}^G\triangleq \sqrt{{\rho }_{ij}}H_{ij}X_i^G+\tilde{Z_j}$ and $Y_i^G\triangleq\sqrt{{\rho }_{ii}}H_{ii}X_i^G+\sqrt{{\rho }_{ji}}H_{ji}X_j^G+Z_i$.

The rest of the section considers the 10 terms in Lemma \ref{lemma_conditional} and outer-bounds each of them to get the terms in the outer-bound of Theorem \ref{outer_inner_capacity_reciprocal}.

\noindent {\bf \eqref{o1}$\rightarrow$\eqref{ro0eq1}:} We can split the bound in \eqref{o1} into two upper bounds. The first bound is
\begin{eqnarray}
R_1&\le& h(H_{11}X_1+Z_1)-h(Z_1)+h(H_{12}X_1+Z_2|H_{11}X_1+Z_1)-h(Z_2)\nonumber\\
&=& h(H_{12}X_1+Z_2,H_{11}X_1+Z_1)-h(Z_1)-h(Z_2)\nonumber\\
&\stackrel{(a)}{\le}& \log \det\left[ \begin{array}{cc}
I_{N_2}+{\rho }_{12}H_{12}Q_{11}H^{\dagger }_{12} & \sqrt{{\rho }_{12}{\rho }_{11}}H_{12}Q_{11}H^{\dagger }_{11} \\
\sqrt{{\rho }_{12}{\rho }_{11}}H_{11}Q_{11}H^{\dagger }_{12} & I_{N_1}+{\rho }_{11}H_{11}Q_{11}H^{\dagger }_{11} \end{array}
\right]\nonumber\\
&\stackrel{(b)}{=}& {\log  {\det  (I_{N_1}+{\rho}_{11}H_{11}Q_{11}H^{\dagger}_{11})}}+\log \det (I_{N_2}+{\rho }_{12}H_{12}Q_{11}H^{\dagger }_{12}\nonumber\\
&&-{\rho }_{12}{\rho }_{11}H_{12}Q_{11}H^{\dagger }_{11}{(I_{N_1}+{\rho}_{11}H_{11}Q_{11}H^{\dagger}_{11})^{-1}}H_{11}Q_{11}H^{\dagger }_{12})\nonumber\\
&\stackrel{(c)}{\le}& {\log  {\det  (I_{N_1}+{\rho}_{11}H_{11}H^{\dagger}_{11})}}+\log \det (I_{N_2}+{\rho }_{12}H_{12}Q_{11}H^{\dagger }_{12}\nonumber\\
&&-{\rho }_{12}{\rho }_{11}H_{12}Q_{11}H^{\dagger }_{11}{(I_{N_1}+{\rho}_{11}H_{11}Q_{11}H^{\dagger}_{11})^{-1}}H_{11}Q_{11}H^{\dagger }_{12})\nonumber\\
&{=}& {\log  {\det  (I_{N_1}+{\rho}_{11}H_{11}H^{\dagger}_{11})}}+\log \det (I_{N_2}\nonumber\\
&&+{\rho }_{12}H_{12}(Q_{11}-{\rho }_{11}Q_{11}H^{\dagger }_{11}{(I_{N_1}+{\rho}_{11}H_{11}Q_{11}H^{\dagger}_{11})^{-1}}H_{11}Q_{11})H^{\dagger }_{12})\nonumber\\
&\stackrel{(d)}{\le}& {\log  {\det  (I_{N_1}+{\rho}_{11}H_{11}H^{\dagger}_{11})}}+\log \det (I_{N_2}+{\rho }_{12}H_{12}H^{\dagger }_{12}\nonumber\\
&&-{\rho }_{12}{\rho }_{11}H_{12}H^{\dagger }_{11}{(I_{N_1}+{\rho}_{11}H_{11}H^{\dagger}_{11})^{-1}}H_{11}H^{\dagger }_{12}),
\end{eqnarray}
where $(a)$ follows from Lemma \ref{gauss_ob} and from the fact that $h(Z_i)=\log \det \left(2\pi e I_{N_i}\right)$, $(b)$ follows from Lemma \ref{lem_block}, $(c)$ follows from the fact that $\log\det(.)$ is a monotonically increasing function on the cone of positive definite matrices and we have $Q_{ii} \preceq I_{M_i}$ for $i\in \{1,2\}$, and $(d)$ follows from Lemma \ref{L} where $K_1=Q_{11}$, $K_2=I_{M_{1}}$ and $S=\sqrt{{\rho }_{11}}H^{\dagger}_{11}$. It gives the first part of the bound \eqref{ro0eq1}.

The second bound is
\begin{eqnarray}
R_1&\le& h(H_{11}X_1+Z_1)-h(Z_1)+C_{21}\nonumber\\
&\stackrel{(a)}{=}&{\log  {\det  (I_{N_1}+{\rho}_{11}H_{11}Q_{11}H^{\dagger}_{11})}}+C_{21}\nonumber\\
&\stackrel{(b)}{\le}&{\log  {\det  (I_{N_1}+{\rho}_{11}H_{11}H^{\dagger}_{11})}}+C_{21},
\end{eqnarray}
where $(a)$ follows from Lemma \ref{gauss_ob} and from the fact that $h(Z_i)=\log \det  \left(2\pi e I_{N_i}\right)$, $(b)$ follows from the fact that $\log\det(.)$ is a monotonically increasing function on the cone of positive definite matrices and we have $Q_{ii} \preceq I_{M_i}$ for $i\in \{1,2\}$. It gives the second part of the bound \eqref{ro0eq1}.

\noindent {\bf \eqref{o2}$\rightarrow$\eqref{ro0eq2}:} This is obtained similarly to the last bound by exchanging 1 and 2 in the indices.

\noindent {\bf \eqref{o3}$\rightarrow$\eqref{ro0eq3}:} For the bound \eqref{o3} in Lemma \ref{lemma_conditional},
\begin{eqnarray}
R_1 + R_2&\le&
h(Y_1
|\tilde{S_1})+h(Y_2|\tilde{S_2})-h(\tilde{Z_1})-h(\tilde{Z_1})+C_{21}+C_{12}\nonumber\\
&=& h(\sqrt{{\rho }_{11}}H_{11}X_{1}+\sqrt{{\rho }_{21}}H_{21}X_{2}+Z_{1}|\sqrt{{\rho }_{12}}H_{12}X_{1}+\tilde{Z_{2}}
)+h(\sqrt{{\rho }_{12}}H_{12}X_{1}+\nonumber\\
&&\sqrt{{\rho }_{22}}H_{22}X_{2}+Z_{2}|\sqrt{{\rho }_{21}}H_{21}X_{2}+\tilde{Z_{1}})-h(\tilde{Z_1})-h(\tilde{Z_1})+C_{21}+C_{12}\nonumber\\
&\stackrel{(a)}{\le}& h(\sqrt{{\rho }_{11}}H_{11}X_{1}^G+\sqrt{{\rho }_{21}}H_{21}X_{2}^G+Z_{1}|\sqrt{{\rho }_{12}}H_{12}X_{1}^G+\tilde{Z_{2}}
)+h(\sqrt{{\rho }_{12}}H_{12}X_{1}^G+\nonumber\\
&&\sqrt{{\rho }_{22}}H_{22}X_{2}^G+Z_{2}|\sqrt{{\rho }_{21}}H_{21}X_{2}^G+\tilde{Z_{1}})-h(\tilde{Z_1})-h(\tilde{Z_1})+C_{21}+C_{12}\nonumber\\
&=& h(\sqrt{{\rho }_{11}}H_{11}X_{1}^G+\sqrt{{\rho }_{21}}H_{21}X_{2}^G+Z_{1},\sqrt{{\rho }_{12}}H_{12}X_{1}^G+\tilde{Z_{2}}
)-h(\sqrt{{\rho }_{12}}H_{12}X_{1}^G+\tilde{Z_{2}}
)\nonumber\\
&&+h(\sqrt{{\rho }_{12}}H_{12}X_{1}^G+\sqrt{{\rho }_{22}}H_{22}X_{2}^G+Z_{2}\sqrt{{\rho }_{21}}H_{21}X_{2}^G+\tilde{Z_{1}})-h(\sqrt{{\rho }_{21}}H_{21}X_{2}^G+\tilde{Z_{1}})\nonumber\\
&&-h(\tilde{Z_1})-h(\tilde{Z_1})+C_{21}+C_{12}\nonumber\\
&\stackrel{(b)}{=}& \log \det\left[ \begin{array}{cc}
I_{N_1}+{\rho }_{11}H_{11}Q_{11}H^{\dagger }_{11}
+{\rho }_{21}H_{21}Q_{22}H^{\dagger }_{21}
 &
\sqrt{{\rho }_{12}{\rho }_{11}}H_{11}Q_{11}H^{\dagger }_{12}
\\
\sqrt{{\rho }_{12}{\rho }_{11}}H_{12}Q_{11}H^{\dagger }_{11}
&
I_{N_2}+{\rho }_{12}H_{12}Q_{11}H^{\dagger }_{12}
\end{array}
\right]\nonumber\\
&& + \log \det\left[ \begin{array}{cc}
I_{N_2}+{\rho }_{22}H_{22}Q_{22}H^{\dagger }_{22}
+{\rho }_{12}H_{12}Q_{11}H^{\dagger }_{12}
 &
\sqrt{{\rho }_{21}{\rho }_{22}}H_{22}Q_{22}H^{\dagger }_{21}
\\
\sqrt{{\rho }_{21}{\rho }_{22}}H_{21}Q_{22}H^{\dagger }_{22}
&
I_{N_1}+{\rho }_{21}H_{21}Q_{22}H^{\dagger }_{21}
\end{array}
\right]\nonumber\\
&&- \log \det (I_{N_2}+{\rho }_{12}H_{12}Q_{11}H^{\dagger }_{12})
- \log \det (I_{N_1}+{\rho }_{21}H_{21}Q_{22}H^{\dagger }_{21})
+C_{21}+C_{12}\nonumber\\
&\stackrel{(c)}{=}&
\log  \det (I_{N_1}+{\rho}_{11}H_{11}Q_{11}H^{\dagger}_{11}+
{\rho}_{21}H_{21}Q_{22}H^{\dagger}_{21}-
{\rho}_{11}{\rho}_{12}H_{11}Q_{11}H^{\dagger}_{12}\nonumber\\
&&{(I_{N_2}+{\rho}_{12}H_{12}Q_{11}H^{\dagger}_{12})^{-1}}H_{12}Q_{11}H^{\dagger}_{11})+\log  \det (I_{N_2}+{\rho}_{22}H_{22}Q_{22}H^{\dagger}_{22}+\nonumber\\
&&{\rho}_{12}H_{12}Q_{11}H^{\dagger}_{12}-
{\rho}_{22}{\rho}_{21}H_{22}Q_{22}H^{\dagger}_{21}
{(I_{N_1}+{\rho}_{21}H_{21}Q_{22}H^{\dagger}_{21})^{-1}}H_{21}Q_{22}H^{\dagger}_{22})+\nonumber\\
&&C_{12}+C_{21}\nonumber
\end{eqnarray}

\begin{eqnarray}
&\stackrel{(d)}{\le}&  \log  \det (I_{N_1}+{\rho}_{11}H_{11}H^{\dagger}_{11}+
{\rho}_{21}H_{21}H^{\dagger}_{21}-
{\rho}_{11}{\rho}_{12}H_{11}H^{\dagger}_{12}
{(I_{N_2}+{\rho}_{12}H_{12}H^{\dagger}_{12})^{-1}}H_{12}H^{\dagger}_{11})+\nonumber\\
&&\log  \det (I_{N_2}+{\rho}_{22}H_{22}H^{\dagger}_{22}+
{\rho}_{12}H_{12}H^{\dagger}_{12}-
{\rho}_{22}{\rho}_{21}H_{22}H^{\dagger}_{21}
{(I_{N_1}+{\rho}_{21}H_{21}H^{\dagger}_{21})^{-1}}H_{21}H^{\dagger}_{22})+\nonumber\\
&&C_{12}+C_{21},
\end{eqnarray}
where $(a)$ follows from Lemma \ref{gauss_ob}, $(b)$ follows from the fact that $h(Z_i)=\log\det \left(2\pi e I_{N_i}\right)$, $(c)$ follows from Lemma \ref{lem_block}, and $(d)$ follows from the fact that $\log\det(.)$ is a monotonically increasing function on the cone of positive definite matrices and we have $Q_{ii} \preceq I_{M_i}$ for $i\in \{1,2\}$, and Lemma \ref{L} where for the first term $K_1=Q_{11}$, $K_2=I_{M_{1}}$ and $S=\sqrt{{\rho }_{12}}H^{\dagger}_{12}$ and for the second term  where $K_1=Q_{22}$, $K_2=I_{M_{2}}$ and $S=\sqrt{{\rho }_{21}}H^{\dagger}_{21}$. It gives the bound \eqref{ro0eq3}.

\noindent {\bf \eqref{o4}$\rightarrow$\eqref{ro0eq4}:} For the bound \eqref{o4} in Lemma \ref{lemma_conditional},
\begin{eqnarray}
R_1+R_2&\le& h(H_{11}X_1+Z_1|{S_1})+h(Y_2)-h(Z_1,Z_2)+C_{12}\nonumber\\
&=& h(H_{11}X_1+Z_1|H_{12}X_1+Z_2)+h(Y_2)-h(Z_1,Z_2)+C_{12}\nonumber\\
&\le& h(H_{11}X_1^G+Z_1|H_{12}X_1^G+Z_2)+h(Y_2^G)-h(Z_1,Z_2)+C_{12}\nonumber\\
&=& h(H_{11}X_1^G+Z_1,H_{12}X_1^G+Z_2)-h(H_{12}X_1^G+Z_2)+h(Y_2^G)-h(Z_1,Z_2)+C_{12}\nonumber\\
&\stackrel{(a)}{\le}& \log \det\left[ \begin{array}{cc}
I_{N_1}+{\rho }_{11}H_{11}Q_{11}H^{\dagger }_{11}
 &
\sqrt{{\rho }_{12}{\rho }_{11}}H_{11}Q_{11}H^{\dagger }_{12}
\\
\sqrt{{\rho }_{12}{\rho }_{11}}H_{12}Q_{11}H^{\dagger }_{11}
&
I_{N_2}+{\rho }_{12}H_{12}Q_{11}H^{\dagger }_{12}
\end{array}
\right]\nonumber\\
&&+\log \det (I_{N_2}+{\rho }_{12}H_{12}Q_{11}H^{\dagger }_{12}+{\rho }_{22}H_{22}Q_{22}H^{\dagger }_{22})\nonumber\\
&&-\log \det (I_{N_2}+{\rho }_{12}H_{12}Q_{11}H^{\dagger }_{12})+C_{12}\nonumber\\
&\stackrel{(b)}{=}& \log  \det (I_{N_1}+{\rho}_{11}H_{11}Q_{11}H^{\dagger}_{11}-
{\rho}_{11}{\rho}_{12}H_{11}Q_{11}H^{\dagger}_{12}
{(I_{N_2}+{\rho}_{12}H_{12}Q_{11}H^{\dagger}_{12})^{-1}}H_{12}Q_{11}H^{\dagger}_{11})+\nonumber\\
&&+\log \det (I_{N_2}+{\rho }_{12}H_{12}Q_{11}H^{\dagger }_{12}+{\rho }_{22}H_{22}Q_{22}H^{\dagger }_{22})+C_{12}\nonumber\\
&\stackrel{(c)}{\le}& \log  \det (I_{N_1}+{\rho}_{11}H_{11}H^{\dagger}_{11}-
{\rho}_{11}{\rho}_{12}H_{11}H^{\dagger}_{12}
{(I_{N_2}+{\rho}_{12}H_{12}H^{\dagger}_{12})^{-1}}H_{12}H^{\dagger}_{11})+\nonumber\\
&&+\log \det (I_{N_2}+{\rho }_{12}H_{12}Q_{11}H^{\dagger }_{12}+{\rho }_{22}H_{22}Q_{22}H^{\dagger }_{22})+C_{12}\nonumber\\
&\stackrel{(d)}{\le}&\log  \det (I_{N_1}+{\rho}_{11}H_{11}H^{\dagger}_{11}-
{\rho}_{11}{\rho}_{12}H_{11}H^{\dagger}_{12}
{(I_{N_2}+{\rho}_{12}H_{12}H^{\dagger}_{12})^{-1}}H_{12}H^{\dagger}_{11})+\nonumber\\
&&\log  \det (I_{N_2}+{\rho}_{22}H_{22}H^{\dagger}_{22}+
{\rho}_{12}H_{12}H^{\dagger}_{12})+C_{12},
\end{eqnarray}
where $(a)$ follows from the fact that $h(Z_i)=\log \det \left(2\pi e I_{N_i}\right)$, and $(b)$ follows from Lemma \ref{lem_block}, and $(c)$ follows from Lemma \ref{L} where $K_1=Q_{11}$, $K_2=I_{M_{1}}$ and $S=\sqrt{{\rho }_{12}}H^{\dagger}_{12}$, and $(d)$ follows from the fact that $\log\det(.)$ is a monotonically increasing function on the cone of positive definite matrices and we have $Q_{ii} \preceq I_{M_i}$ for $i\in \{1,2\}$. It gives the bound \eqref{ro0eq4}.

\noindent {\bf \eqref{o5}$\rightarrow$\eqref{ro0eq5}:} This is obtained similarly to the last bound by exchanging 1 and 2 in the indices.

\noindent {\bf \eqref{o6}$\rightarrow$\eqref{ro0eq6}:} For the bound \eqref{o6} in Lemma \ref{lemma_conditional}, assume infinite capacity between the receivers, i.e., consider a single receiver. We get
\begin{eqnarray}
R_1+R_2 &\le& h\left(Y_1,Y_2\right)-h\left(Z_1,Z_2\right)\nonumber\\
&\stackrel{(a)}{\le}& \log \det \left(I_{N_1+N_2}+
\left[\begin{array}{c}
\sqrt{\rho_{11}}H_{11} \\
\sqrt{\rho_{12}}H_{12} \end{array}\right]
Q_{11}\left[\sqrt{\rho_{11}}H_{11}^{\dagger}\ \sqrt{\rho_{12}}H_{12}^{\dagger}\right]+\right.\nonumber\\
&&\left.\left[\begin{array}{c}
\sqrt{\rho_{21}}H_{21} \\
\sqrt{\rho_{22}}H_{22} \end{array}\right]
Q_{22}\left[\sqrt{\rho_{21}}H_{21}^{\dagger}\ \sqrt{\rho_{22}}H_{22}^{\dagger}\right]\right)\nonumber\\
&\stackrel{(b)}{\le}& \log \det \left(I_{N_1+N_2}+
\left[\begin{array}{c}
\sqrt{\rho_{11}}H_{11} \\
\sqrt{\rho_{12}}H_{12} \end{array}\right]
\left[\sqrt{\rho_{11}}H_{11}^{\dagger}\ \sqrt{\rho_{12}}H_{12}^{\dagger}\right]+\right.\nonumber\\
&&\left.\left[\begin{array}{c}
\sqrt{\rho_{21}}H_{21} \\
\sqrt{\rho_{22}}H_{22} \end{array}\right]
\left[\sqrt{\rho_{21}}H_{21}^{\dagger}\ \sqrt{\rho_{22}}H_{22}^{\dagger}\right]\right),
\end{eqnarray}
where $(a)$ follows from Lemma \ref{gauss_ob} and from the fact that $h(Z_i)=\log \det  \left(2\pi e I_{N_i}\right)$, and $(b)$ follows from the fact that $\log\det(.)$ is a monotonically increasing function on the cone of positive definite matrices and we have $Q_{ii} \preceq I_{M_i}$ for $i\in \{1,2\}$. It gives the bound \eqref{ro0eq6}.

\noindent {\bf \eqref{o7}$\rightarrow$\eqref{ro0eq7}:} For the bound \eqref{o7} in Lemma \ref{lemma_conditional},
\begin{eqnarray}
2{R_1}+R_2&\le& h(\sqrt{{\rho}_{11}}H_{11}X_1+Z_1|{S_1})+h(Y_1)+h(Y_2|{S_2})-h(Z_1,Z_2)-h({Z_1})+C_{21}+C_{12}\nonumber\\
&\le& h(\sqrt{{\rho}_{11}}H_{11}X_1^G+Z_1|{S_1^G})+h(Y_1^G)+h(Y_2^G|{S_2^G})-h(Z_1,Z_2)-h({Z_1})+C_{21}+C_{12}\nonumber\\
&=&
h(\sqrt{{\rho}_{11}}H_{11}X_1^G+Z_1,{S_1^G})-h({S_1^G})+h(Y_1^G)+h(Y_2^G,{S_2^G})-h({S_2^G})-h(Z_1,Z_2)\nonumber\\
&&-h({Z_1})+C_{21}+C_{12}\nonumber\\
&=& h(\sqrt{{\rho}_{11}}H_{11}X_1^G+Z_1,\sqrt{{\rho}_{12}}H_{12}X_1^G+Z_2)+
h(\sqrt{{\rho}_{12}}H_{12}X_1^G+\sqrt{{\rho}_{22}}H_{22}X_2^G+Z_2,\nonumber\\
&&\sqrt{{\rho}_{21}}H_{21}X_2^G+Z_1)+h(\sqrt{{\rho}_{11}}H_{11}X_1^G+\sqrt{{\rho}_{21}}H_{21}X_2^G+Z_1)
-h(\sqrt{{\rho}_{12}}H_{12}X_1^G+Z_2)\nonumber\\
&&-h(\sqrt{{\rho}_{21}}H_{21}X_2^G+Z_1^G)-h(Z_1,Z_2)-h({Z_1})+C_{21}+C_{12}\nonumber
\end{eqnarray}

\begin{eqnarray}
&\stackrel{(a)}{\le}&
\log \det\left[ \begin{array}{cc}
I_{N_2}+{\rho }_{12}H_{12}Q_{11}H^{\dagger }_{12}+{\rho }_{22}H_{22}Q_{22}H^{\dagger }_{22}
 &
\sqrt{{\rho }_{22}{\rho }_{21}}H_{22}Q_{22}H^{\dagger }_{21}
\\
\sqrt{{\rho }_{22}{\rho }_{21}}H_{21}Q_{22}H^{\dagger }_{22}
&
I_{N_1}+{\rho }_{21}H_{21}Q_{22}H^{\dagger }_{21}
\end{array}
\right]\nonumber\\
&&+\log \det\left[ \begin{array}{cc}
I_{N_1}+{\rho }_{11}H_{11}Q_{11}H^{\dagger }_{11}
 &
\sqrt{{\rho }_{12}{\rho }_{11}}H_{11}Q_{11}H^{\dagger }_{12}
\\
\sqrt{{\rho }_{12}{\rho }_{11}}H_{12}Q_{11}H^{\dagger }_{11}
&
I_{N_2}+{\rho }_{12}H_{12}Q_{11}H^{\dagger }_{12}
\end{array}
\right]\nonumber\\
&&+\log \det (I_{N_1}+{\rho }_{11}H_{11}Q_{11}H^{\dagger }_{11}+{\rho }_{21}H_{21}Q_{22}H^{\dagger }_{21})\nonumber\\
&&-\log \det (I_{N_2}+{\rho }_{12}H_{12}Q_{11}H^{\dagger }_{12})
-\log \det (I_{N_1}+{\rho }_{21}H_{21}Q_{22}H^{\dagger }_{21})
+C_{12}+C_{21}\nonumber\\
&\stackrel{(b)}{\le}& \log  \det (I_{N_1}+{\rho}_{11}H_{11}Q_{11}H^{\dagger}_{11}-
{\rho}_{11}{\rho}_{12}H_{11}Q_{11}H^{\dagger}_{12}
{(I_{N_2}+{\rho}_{12}H_{12}Q_{11}H^{\dagger}_{12})^{-1}}H_{12}Q_{11}H^{\dagger}_{11})+\nonumber\\
&&\log  \det (I_{N_2}+{\rho}_{22}H_{22}Q_{22}H^{\dagger}_{22}+
{\rho}_{12}H_{12}Q_{11}H^{\dagger}_{12}-
{\rho}_{22}{\rho}_{21}H_{22}Q_{22}H^{\dagger}_{21}
{(I_{N_1}+{\rho}_{21}H_{21}Q_{22}H^{\dagger}_{21})^{-1}}\nonumber\\
&&H_{21}Q_{22}H^{\dagger}_{22})+
\log  \det (I_{N_1}+{\rho}_{11}H_{11}Q_{11}H^{\dagger}_{11}+
{\rho}_{21}H_{21}Q_{22}H^{\dagger}_{21})+C_{12}+C_{21}\nonumber\\
&\stackrel{(c)}{\le}& \log  \det (I_{N_1}+{\rho}_{11}H_{11}H^{\dagger}_{11}-
{\rho}_{11}{\rho}_{12}H_{11}H^{\dagger}_{12}
{(I_{N_2}+{\rho}_{12}H_{12}H^{\dagger}_{12})^{-1}}H_{12}H^{\dagger}_{11})+\nonumber\\
&&\log  \det (I_{N_2}+{\rho}_{22}H_{22}H^{\dagger}_{22}+
{\rho}_{12}H_{12}H^{\dagger}_{12}-
{\rho}_{22}{\rho}_{21}H_{22}H^{\dagger}_{21}
{(I_{N_1}+{\rho}_{21}H_{21}H^{\dagger}_{21})^{-1}}H_{21}H^{\dagger}_{22})+\nonumber\\
&&\log  \det (I_{N_1}+{\rho}_{11}H_{11}H^{\dagger}_{11}+
{\rho}_{21}H_{21}H^{\dagger}_{21})+C_{12}+C_{21},
\end{eqnarray}
where $(a)$ follows from Lemma \ref{gauss_ob} and from the fact that $h(Z_i)=\log \det  \left(2\pi e I_{N_i}\right)$, $(b)$ follows from Lemma \ref{lem_block}, and $(c)$ follows from Lemma \ref{L} and the fact that $\log\det(.)$ is a monotonically increasing function on the cone of positive definite matrices and we have $Q_{ii} \preceq I_{M_i}$ for $i\in \{1,2\}$. It gives the bound \eqref{ro0eq7}.

\noindent {\bf \eqref{o8}$\rightarrow$\eqref{ro0eq8}:} This is obtained similarly to the last bound by exchanging 1 and 2 in the indices.

\noindent {\bf \eqref{o9}$\rightarrow$\eqref{ro0eq9}:} For the bound \eqref{o9} in Lemma \ref{lemma_conditional},
\begin{eqnarray}
&&2{R_1}+R_2\nonumber\\
&\le& h(Y_1,Y_2|\tilde{S_2})+h(Y_1)-h(Z_1,Z_2)-h({Z_1})+C_{21}\nonumber\\
&\stackrel{(a)}{\le}& h(Y_1^G,Y_2^G|\tilde{S_2}^G)+h(Y_1^G)-h(Z_1,Z_2)-h({Z_1})+C_{21}\nonumber\\
&=& h(Y_1^G,Y_2^G,H_{21}X_2^G+\hat{Z_1})-h(H_{21}X_2^G+\hat{Z_1})+h(Y_1^G)-h(Z_1,Z_2)
-h({Z_1})+C_{21}\nonumber\\
&\stackrel{(b)}{\le}& h(Y_1,Y_2,H_{21}X_2+\hat{Z_1})-h(H_{21}X_2+\hat{Z_1})-h(Z_1,Z_2)+\nonumber\\
&&\log \det (I_{N_1}+{\rho}_{11}H_{11}H^{\dagger}_{11}+
{\rho}_{21}H_{21}H^{\dagger}_{21})+C_{21}\nonumber
\end{eqnarray}

\begin{eqnarray}
&\stackrel{(c)}{\le}&
\log \det \nonumber\\
&&{\small
\left[\begin{array}{ccc}
I_{N_2}+{\rho}_{22}H_{22}Q_{22}H^{\dagger}_{22}+{\rho}_{12}H_{12}Q_{11}H^{\dagger}_{12} & \sqrt{{\rho}_{21}{\rho}_{22}}H_{22}Q_{22}H^{\dagger}_{21}+
\sqrt{{\rho}_{11}{\rho}_{12}}H_{12}Q_{11}H^{\dagger}_{11} & \sqrt{{\rho}_{22}{\rho}_{21}}H_{22}Q_{22}H^{\dagger}_{21}\\
\sqrt{{\rho}_{21}{\rho}_{22}}H_{21}Q_{22}H^{\dagger}_{22}+
\sqrt{{\rho}_{11}{\rho}_{12}}H_{11}Q_{11}H^{\dagger}_{12} & I_{N_1}+{\rho}_{11}H_{11}Q_{11}H^{\dagger}_{11}+
{\rho}_{21}H_{21}Q_{22}H^{\dagger}_{21} & {\rho}_{21}H_{21}Q_{22}H^{\dagger}_{21}\\
\sqrt{{\rho}_{22}{\rho}_{21}}H_{21}Q_{22}H^{\dagger}_{22} &
{\rho}_{21}H_{21}Q_{22}H^{\dagger}_{21} & I_{N_1}+{\rho}_{21}H_{21}Q_{22}H^{\dagger}_{21}\end{array}
\right]}\nonumber\\
&&-\log \det (I_{N_1}+{\rho}_{21}H_{21}Q_{22}H^{\dagger}_{21})+\log \det (I_{N_1}+{\rho}_{11}H_{11}H^{\dagger}_{11}+
{\rho}_{21}H_{21}H^{\dagger}_{21})
+C_{21}\nonumber\\
&\stackrel{(d)}{=}& \log  \det \left(I_{N_1+N_2}+
\biggl[\begin{array}{c}
\sqrt{\rho_{22}}H_{22} \\
\sqrt{\rho_{21}}H_{21} \end{array}\biggr]
(Q_{22}-Q_{22}H^{\dagger}_{21}{(I_{N_1}+{\rho}_{21}H_{21}Q_{22}H^{\dagger}_{21})^{-1}}H_{21}Q_{22})
[\sqrt{\rho_{22}}H_{22}^{\dagger}\ \sqrt{\rho_{21}}H_{21}^{\dagger}]\right.\nonumber\\
&&\left.+\biggl[\begin{array}{c}
\sqrt{\rho_{12}}H_{12} \\
\sqrt{\rho_{11}}H_{11} \end{array}\biggr]
Q_{11}[\sqrt{\rho_{12}}H_{12}^{\dagger}\ \sqrt{\rho_{11}}H_{11}^{\dagger}]\right)+\log  \det (I_{N_1}+{\rho}_{11}H_{11}H^{\dagger}_{11}+
{\rho}_{21}H_{21}H^{\dagger}_{12})+C_{21}\nonumber\\
&\stackrel{(e)}{\le}& \log  \det \left(I_{N_1+N_2}+
\biggl[\begin{array}{c}
\sqrt{\rho_{22}}H_{22} \\
\sqrt{\rho_{21}}H_{21} \end{array}\biggr]
(I_{M_2}-H^{\dagger}_{21}{(I_{N_1}+{\rho}_{21}H_{21}H^{\dagger}_{21})^{-1}}H_{21})
[\sqrt{\rho_{22}}H_{22}^{\dagger}\ \sqrt{\rho_{21}}H_{21}^{\dagger}]\right.\nonumber\\
&&\left.+\biggl[\begin{array}{c}
\sqrt{\rho_{12}}H_{12} \\
\sqrt{\rho_{11}}H_{11} \end{array}\biggr]
[\sqrt{\rho_{12}}H_{12}^{\dagger}\ \sqrt{\rho_{11}}H_{11}^{\dagger}]\right)+\log  \det (I_{N_1}+{\rho}_{11}H_{11}H^{\dagger}_{11}+
{\rho}_{21}H_{21}H^{\dagger}_{12})+C_{21},
\end{eqnarray}
where $(a)$ and $(b)$ follow from Lemma \ref{gauss_ob} and from the fact that $h(Z_i)=\log \det \left(2\pi e I_{N_i}\right)$ and the fact that $\log\det(.)$ is a monotonically increasing function on the cone of positive definite matrices and we have $Q_{ii} \preceq I_{M_i}$ for $i\in \{1,2\}$, $(c)$ follows from Lemma \ref{gauss_ob}, $(d)$ follows from Lemma \ref{lem_block}, and $(e)$ follows from Lemma \ref{L} and the fact that $\log\det(.)$ is a monotonically increasing function on the cone of positive definite matrices and we have $Q_{ii} \preceq I_{M_i}$ for $i\in \{1,2\}$. It gives the bound \eqref{ro0eq9}.

\noindent {\bf \eqref{o10}$\rightarrow$\eqref{ro0eq10}:} This is obtained similarly to the last bound by exchanging 1 and 2 in the indices.

\section{Proof of Achievability for Theorem \ref{outer_inner_capacity_reciprocal}} \label{apdx_inner}

In this section, we prove the achievability for Theorem \ref{outer_inner_capacity_reciprocal}. Denote the RHS of the 10 terms in \eqref{ro0eq1}-\eqref{ro0eq10} as $I_1$ to $I_{10}$, respectively.
We will show a constant gap achiavability result for the two-user MIMO Gaussian IC with limited receiver cooperation in the following Lemma.

\begin{lemma}\label{achiv}
The capacity region for the two-user MIMO IC with receiver cooperation contains the region formed by $(R_1,R_2)$ such that
\begin{eqnarray}
R_1 &\le& I_1-N_1-N_2,\nonumber\\
R_2 &\le& I_2-N_1-N_2,\nonumber\\
R_1+R_2 &\le& \min\{I_3,I_4,I_5,I_6\}-N_1-N_2-\max(N_1,N_2),\nonumber\\
2R_1+R_2 &\le& \min\{I_7,I_9\}-2N_1-2N_2,\nonumber\\
R_1+2R_2 &\le& \min\{I_8,I_{10}\}-2N_1-3N_2.
\end{eqnarray}
\end{lemma}
The rest of this section proves this Lemma. This region is within $N_1+N_2$ bits of the outer bound giver by $\mathcal{R}_o$ and thus proves the achievability for Theorem 1. In the following, we will consider the rate regions for $STG_{2\rightarrow 1\rightarrow 2}$ and then take the convex hull of $STG_{2\rightarrow 1\rightarrow 2}$ and $STG_{1\rightarrow 2\rightarrow 1}$ to get this result.

\begin{lemma} \label{inner}
If we consider $STG_{2\rightarrow 1\rightarrow 2}$, the capacity region of the two-user MIMO Gaussian IC with limited receiver cooperation includes the set of $(R_1,R_2)$ such that
\begin{eqnarray}
R_1&\le& I(X_1;Y_1|X_{2c}),\label{rr1x}\\
R_1&\le& I(X_1;Y_1|X_{1c},X_{2c})+I(X_{1c},X_2;Y_2|X_{2c})+C_{12},\label{rr1}\\
R_2&\le& I(X_2;Y_2|X_{1c})+C_{12},\label{rr2x}\\
R_2&\le& I(X_{2c};Y_1|X_1)+I(X_2;Y_2|X_{1c},X_{2c}),\label{rr2}\\
R_1+R_2&\le& I(X_{2c},X_1;Y_1)+I(X_2;Y_2|X_{1c},X_{2c})+{(C_{21}-\xi)}^{+},\label{r1r2.1}\\
R_1+R_2&\le& I(X_{2c},X_1;Y_1,\hat{Y_2})+I(X_2;Y_2|X_{1c},X_{2c}),\label{r1r2.2}\\
R_1+R_2&\le& I(X_{2c},X_1;Y_1|X_{1c})+I(X_{1c},X_2;Y_2|X_{2c})+C_{12}+{(C_{21}-\xi)}^{+},\label{r1r2.3}\\
R_1+R_2&\le& I(X_{2c},X_1;Y_1,\hat{Y_2}|X_{1c})+I(X_{1c},X_2;Y_2|X_{2c})+C_{12},\label{r1r2.4}\\
R_1+R_2&\le& I(X_1;Y_1|X_{1c},X_{2c})+I(X_{1c},X_2;Y_2)+C_{12},\label{r1r2.5}\\
R_1+R_2&\le& I(X_1;Y_1|X_{1c},X_{2c})+I(X_{2c};Y_1|X_1)+I(X_{1c},X_2;Y_2|X_{2c})+C_{12},\label{r1r2.6}\\
2R_1+R_2&\le& I(X_1,X_{2c};Y_1)+I(X_1;Y_1|X_{1c},X_{2c})+
I(X_{1c},X_2;Y_2|X_{2c})+C_{12}+{(C_{21}-\xi)}^{+},\label{2r1r2.1}\\
2R_1+R_2&\le& I(X_1,X_{2c};Y_1,\hat {Y_2})+I(X_1;Y_1|X_{1c},X_{2c})
+I(X_{1c},X_2;Y_2|X_{2c})+C_{12},\label{2r1r2.2}\\
R_1+2R_2&\le& I(X_1,X_{2c};Y_1|X_{1c})+I(X_{1c},X_{2};Y_2)
+I(X_2;Y_2|X_{1c},X_{2c})+C_{12}+{(C_{21}-\xi)}^{+},\label{r12r2.1}
\end{eqnarray}

\begin{eqnarray}
R_1+2R_2&\le& I(X_1,X_{2c};Y_1|X_{1c})+I(X_{2c};Y_1|X_{1})+I(X_{1c},X_{2};Y_2|X_{2c})
+I(X_2;Y_2|X_{1c},X_{2c})\nonumber\\
&&+C_{12}+{(C_{21}-\xi)}^{+},\label{r12r2.2}\\
R_1+2R_2&\le& I(X_1,X_{2c};Y_1,\hat {Y_2}|X_{1c})+I(X_{1c},X_2;Y_2)
+I(X_2;Y_2|X_{1c},X_{2c})+C_{12},\label{r12r2.3}\\
R_1+2R_2&\le& I(X_1,X_{2c};Y_1,\hat {Y_2}|X_{1c})+
I(X_{2c};Y_1|X_1)+
I(X_{1c},X_2;Y_2|X_{2c})\nonumber\\
&&+I(X_2;Y_2|X_{1c},X_{2c})+C_{12}\label{r12r2.4}.
\end{eqnarray}
where $\hat{Y}_i$ is defined in \eqref{hat}.
\end{lemma}
\begin{proof}
The proof follows similarly to that in subsection V.C. of \cite{Tse}, replacing scalars in the SISO channel by vectors for the MIMO channel.
\end{proof}

The rest of the section inner bounds the convex hull of union of this region and the one achieved from $STG_{1\rightarrow 2\rightarrow 1}$ to get the inner bound in Theorem \ref{outer_inner_capacity_reciprocal}.

The achievability scheme is a 2-round protocol as described in Section III and the transmission scheme is based on \eqref{1t}, \eqref{2t} and \eqref{3t}.

We will first evaluate some entropies that will be used in inner bounds of the achievable rate region.
\begin{eqnarray}\label{achyi}
&&h\left(Y_i\right)={\log  {\det (I_{N_i}+{\rho }_{ii}H_{ii}H^{\dagger }_{ii}+{\rho }_{ji}H_{ji}H^{\dagger }_{ji})\ }\ }+N_i\log (2\pi e),
\end{eqnarray}
\begin{eqnarray}\label{achyixi}
h\left(Y_i|X_i\right)={\log  {\det \left(I_{N_i}+{\rho }_{ji}H_{ji}H^{\dagger }_{ji}\right)\ }\ }+N_i\log (2\pi e).
\end{eqnarray}
In addition, we have
\begin{eqnarray}\label{achyiu}
&&h(Y_i|X_{ic},X_{jc})\ge h(Y_i|X_{ic},X_{jc},X_j)=\log \det (I_{N_i}+{\rho }_{ii}H_{ii}Q_{ip}H^{\dagger }_{ii})+N_i\log (2\pi e)\nonumber\\
&=&\log  \det  (I_{N_i}+{\rho }_{ii}H_{ii}H^{\dagger }_{ii}-
\sqrt{{\rho }_{ii}{\rho }_{ij}}H_{ii}H^{\dagger }_{ij}
(I_{N_j}+{\rho }_{ij}H_{ij}H^{\dagger }_{ij})^{-1}
\sqrt{{\rho }_{ii}{\rho }_{ij}}H_{ij}H^{\dagger }_{ii}+N_i\log (2\pi e).
\end{eqnarray}
Moreover, we have
\begin{eqnarray}
h(Y_i|X_{jc},X_i)
&\le& {\log  {\det (I_{N_i}+{\rho }_{ji}H_{ji}Q_{jp}H^{\dagger }_{ji})\ }\ }+N_i\log (2\pi e)\nonumber\\
&\stackrel{(a)}{\le}& \log {\det  \left(2I_{N_i}\right)\ }+N_i\log (2\pi e)\nonumber\\
&=& N_i+N_i\log (2\pi e),\label{U_B}
\end{eqnarray}
where $(a)$ follows from Lemma 11 of \cite{J1} by substituting $\sqrt{{\rho }_{ji}}H_{ji}^{\dagger}$ in $S$. This shows that $h(Y_i|X_{jc},X_i)$ is upper-bounded by $N_i$.

The rest of the section evaluates some terms in Lemma \ref{inner}. We will not evaluate the bounds \eqref{rr1x} and \eqref{2r1r2.2} for now and show that the rest of the bounds contain a region within $N_1+N_2$ bits of the outer bounds.

\noindent {\bf \eqref{rr1}:} For this bound in Lemma \ref{inner}, we have
\begin{eqnarray}
&&I(X_1;Y_1|X_{1c},X_{2c})+I(X_{1c},X_2;Y_2|X_{2c})+C_{12}\nonumber\\
&=&I(X_1;Y_1|X_{1c},X_{2c})+I(X_{1c},X_{2p};\sqrt{{\rho }_{12}}H_{12}X_{1}+\sqrt{{\rho }_{22}}H_{22}X_{2p}+Z_{2})+C_{12}\nonumber\\
&=&I(X_1;Y_1|X_{1c},X_{2c})+h(\sqrt{{\rho }_{12}}H_{12}X_{1}+\sqrt{{\rho }_{22}}H_{22}X_{2p}+Z_{2})-h(\sqrt{{\rho }_{12}}H_{12}X_{1p}+Z_{2})+C_{12}\nonumber\\
&\ge&I(X_1;Y_1|X_{1c},X_{2c})+h(\sqrt{{\rho }_{12}}H_{12}X_{1}+Z_{2})-
h(\sqrt{{\rho }_{12}}H_{12}X_{1p}+Z_{2})\nonumber\\
&=&I(X_1;Y_1|X_{1c},X_{2c})+I(X_{1c};\sqrt{{\rho }_{12}}H_{12}X_{1}+Z_{2})\nonumber\\
&=&I(X_1;Y_1|X_{2c},X_{1c})+I(X_{1c};Y_2|X_2)\nonumber\\
&=& h(Y_2|X_2)-h(Y_2|X_2,X_{1c})+h(Y_1|X_{2c},X_{1c})-h(Y_1|X_{2c},X_{1c},X_1)\nonumber\\
&\stackrel{(a)}{\ge}& h(Y_2|X_2)+h(Y_1|X_{2c},X_{1c})-N_1-N_2-(N_1+N_2)\log (2\pi e)\nonumber\\
&\stackrel{(b)}{=}& \log \det (I_{N_2}+{\rho }_{12}H_{12}H^{\dagger }_{12})+\log \det(I_{N_1}+{\rho }_{11}H_{11}Q_{1p}H^{\dagger }_{11}+{\rho }_{21}H_{21}Q_{2p}H^{\dagger }_{21})-N_1-N_2\nonumber\\
&\stackrel{(c)}{\ge}& \log \det (I_{N_2}+{\rho }_{12}H_{12}H^{\dagger }_{12})+\log \det(I_{N_1}+{\rho }_{11}H_{11}Q_{1p}H^{\dagger }_{11})-N_1-N_2\nonumber\\
&\stackrel{(d)}{=}& {\log  {\det  (I_{N_1}+{\rho}_{11}H_{11}H^{\dagger}_{11})}}+\log \det (I_{N_2}+{\rho }_{12}H_{12}H^{\dagger }_{12}-{\rho }_{12}{\rho }_{11}H_{12}H^{\dagger }_{11}\nonumber\\
&&{(I_{N_1}+{\rho}_{11}H_{11}H^{\dagger}_{11})^{-1}}H_{11}H^{\dagger }_{12})-N_1-N_2,
\end{eqnarray}
where $(a)$ follows from \eqref{U_B}, $(b)$ follows from the assumed Gaussian distributions, $(c)$ follows from the fact that $\log\det(.)$ is a monotonically increasing function on the cone of positive definite matrices, and $(d)$ follows from the fact that using Lemma \ref{lem_block},
\begin{eqnarray}
&&\log \det (I_{N_2}+{\rho }_{12}H_{12}H^{\dagger }_{12})+\log \det(I_{N_1}+{\rho }_{11}H_{11}Q_{1p}H^{\dagger }_{11})\nonumber\\
&=&
\log \det \left[ \begin{array}{cc}
I_{N_2}+\rho_{12}H_{12}H_{12}^{\dagger} & \sqrt{\rho_{12}\rho_{12}}H_{12}H_{11}^{\dagger} \\
\sqrt{\rho_{11}\rho_{12}}H_{11}H_{12}^{\dagger} & I_{N_1}+\rho_{11}H_{11}H_{11}^{\dagger} \end{array}
\right]\nonumber\\
&=&
\log \det (I_{N_2}+{\rho }_{12}H_{12}H^{\dagger }_{12}-{\rho }_{12}{\rho }_{11}H_{12}H^{\dagger }_{11}{(I_{N_1}+{\rho}_{11}H_{11}H^{\dagger}_{11})^{-1}}H_{11}H^{\dagger }_{12})+\nonumber\\
&&{\log  {\det  (I_{N_1}+{\rho}_{11}H_{11}H^{\dagger}_{11})}}.
\end{eqnarray}

Thus, we see that this $R_1$ bound is within $N_1+N_2$ bits to the outer bound in \eqref{ro0eq1}.

\noindent {\bf \eqref{rr2x}:} For this term in Lemma \ref{inner}, we have
\begin{eqnarray}
&&I(X_2;Y_2|X_{1c})+C_{12}\nonumber\\
&=& h(Y_2|X_{1c})- h(Y_2|X_{1c},X_2)+C_{12}\nonumber\\
&\stackrel{(a)}{\ge}& h(Y_2|X_{1c})+C_{12}- N_2-N_2\log (2\pi e)\nonumber\\
&=& \log \det(I_{N_2}+{\rho }_{22}H_{22}H^{\dagger }_{22}+{\rho }_{12}H_{12}Q_{1p}H^{\dagger }_{12})- N_2+C_{12}\nonumber\\
&\stackrel{(b)}{\ge}& \log \det(I_{N_2}+{\rho }_{22}H_{22}H^{\dagger }_{22})+C_{12}- N_2,
\end{eqnarray}
where $(a)$ follows from \eqref{U_B} and $(b)$ follows from the fact that $\log\det(.)$ is a monotonically increasing function on the cone of positive definite matrices.

Thus, we see that this $R_2$ bound is within $N_2$ bits of the outer bound in \eqref{ro0eq2}.

\noindent {\bf \eqref{rr2}:} For this term in Lemma \ref{inner}, we have
\begin{eqnarray}
&&I(X_{2c};Y_1|X_1)+I(X_2;Y_2|X_{1c},X_{2c})\nonumber\\
&=& h(Y_1|X_1)-h(Y_1|X_1,X_{2c})+h(Y_2|X_{1c},X_{2c})-h(Y_2|X_{1c},X_{2c},X_2)\nonumber\\
&\stackrel{(a)}{\ge}& h(Y_1|X_1)+h(Y_2|X_{1c},X_{2c})-N_1-N_2-(N_1+N_2)\log (2\pi e)\nonumber\\
&\stackrel{(b)}{=}& \log \det (I_{N_1}+{\rho }_{21}H_{21}H^{\dagger }_{21})+\log \det(I_{N_2}+{\rho }_{22}H_{22}Q_{2p}H^{\dagger }_{22}+{\rho }_{12}H_{12}Q_{1p}H^{\dagger }_{12})-N_1-N_2\nonumber\\
&\stackrel{(c)}{\ge}& \log \det (I_{N_1}+{\rho }_{21}H_{21}H^{\dagger }_{21})+\log \det(I_{N_2}+{\rho }_{22}H_{22}Q_{2p}H^{\dagger }_{22})-N_1-N_2\nonumber\\
&\stackrel{(d)}{=}& {\log  {\det  (I_{N_2}+{\rho}_{22}H_{22}H^{\dagger}_{22})}}+\log \det (I_{N_1}+{\rho }_{21}H_{21}H^{\dagger }_{21}-{\rho }_{21}{\rho }_{22}H_{21}H^{\dagger }_{22}\nonumber\\
&&{(I_{N_2}+{\rho}_{22}H_{22}H^{\dagger}_{22})^{-1}}H_{22}H^{\dagger }_{21})-N_1-N_2,
\end{eqnarray}
where $(a)$ follows from \eqref{U_B}, $(b)$ follows from the assumed Gaussian distributions, and $(c)$ follows from the fact that $\log\det(.)$ is a monotonically increasing function on the cone of positive definite matrices and $(d)$ follows from Lemma \ref{lem_block}. Using Lemma \ref{lem_block} it is easy to see that
\begin{eqnarray}
&&\log \det (I_{N_1}+{\rho }_{21}H_{21}H^{\dagger }_{21})+\log \det(I_{N_2}+{\rho }_{22}H_{22}Q_{2p}H^{\dagger }_{22})\nonumber\\
&=& \log  \det \left[ \begin{array}{cc}
I_{N_1}+\rho_{21}H_{21}H_{21}^{\dagger} & \sqrt{\rho_{21}\rho_{21}}H_{21}H_{22}^{\dagger} \\
\sqrt{\rho_{22}\rho_{21}}H_{22}H_{21}^{\dagger} & I_{N_2}+\rho_{22}H_{22}H_{22}^{\dagger} \end{array}
\right]\nonumber\\
&=&\log \det (I_{N_1}+{\rho }_{21}H_{21}H^{\dagger }_{21}-{\rho }_{21}{\rho }_{22}H_{21}H^{\dagger }_{22}{(I_{N_2}+{\rho}_{22}H_{22}H^{\dagger}_{22})^{-1}}H_{22}H^{\dagger }_{21})+\nonumber\\
&&{\log  {\det  (I_{N_2}+{\rho}_{22}H_{22}H^{\dagger}_{22})}}.
\end{eqnarray}
Thus, we see that this $R_2$ bound is within $N_1+N_2$ bits of the outer bound in \eqref{ro0eq2}.

\noindent {\bf \eqref{r1r2.1}:} For this bound in Lemma \ref{inner}, we have
\begin{eqnarray}
&&I(X_{2c},X_1;Y_1)+I(X_2;Y_2|X_{1c},X_{2c})+{(C_{21}-\xi)}^{+}\nonumber\\
&=& h(Y_1)-h(Y_1|X_{2c},X_1)+h(Y_2|X_{1c},X_{2c})-h(Y_2|X_{1c},X_2)
+{(C_{21}-\xi)}^{+}\nonumber\\
&\stackrel{(a)}{\ge}& h(Y_1)+h(Y_2|X_{1c},X_{2c})
+C_{21}-N_1-2{N_2}-(N_1+N_2)\log (2\pi e)\nonumber\\
&=& \log \det (I_{N_1}+{\rho }_{21}H_{21}H^{\dagger }_{21}+{\rho }_{11}H_{11}H^{\dagger }_{11})+\log \det(I_{N_2}+{\rho }_{22}H_{22}Q_{2p}H^{\dagger }_{22}+{\rho }_{12}H_{12}Q_{1p}H^{\dagger }_{12})\nonumber\\
&&+C_{21}-N_1-2{N_2}\nonumber\\
&\stackrel{(b)}{\ge}& \log \det (I_{N_1}+{\rho }_{21}H_{21}H^{\dagger }_{21}+{\rho }_{11}H_{11}H^{\dagger }_{11})+\log \det(I_{N_2}+{\rho }_{22}H_{22}Q_{2p}H^{\dagger }_{22})\nonumber\\
&&+C_{21}-N_1-2{N_2}\nonumber\\
&=& \log  \det (I_{N_2}+{\rho}_{22}H_{22}H^{\dagger}_{22}-
{\rho}_{22}{\rho}_{21}H_{22}H^{\dagger}_{21}
{(I_{N_1}+{\rho}_{21}H_{21}H^{\dagger}_{21})^{-1}}H_{21}H^{\dagger}_{22})+\nonumber\\
&&\log  \det (I_{N_1}+{\rho}_{11}H_{11}H^{\dagger}_{11}+
{\rho}_{21}H_{21}H^{\dagger}_{21})+C_{21}-N_1-2{N_2},
\end{eqnarray}
where $(a)$ follows from \eqref{U_B} and \eqref{xi}, and $(b)$ follows from the fact that $\log\det(.)$ is a monotonically increasing function on the cone of positive definite matrices.

Thus, we see that this $R_1+R_2$ bound is within $N_1+2N_2$ bits of the outer bound in \eqref{ro0eq5}.

\noindent {\bf \eqref{r1r2.2}:} For this bound in Lemma \ref{inner}, we have
\begin{eqnarray}
&&I(X_{2c},X_1;Y_1,\hat{Y_2})+I(X_2;Y_2|X_{1c},X_{2c})\nonumber\\
&=& h(Y_1,\hat{Y_2})-h(Y_1,\hat{Y_2}|X_{2c},X_1)
+h(Y_2|X_{1c},X_{2c})-h(Y_2|X_{1c},X_2)\nonumber\\
&\stackrel{(a)}{\ge}& h(Y_1,\hat{Y_2})-h(Y_1,\hat{Y_2}|X_{2c},X_1)
+h(Y_2|X_{1c},X_{2c})-N_2-N_2\log (2\pi e) \nonumber\\
&=& h(Y_1,\hat{Y_2})-h(Y_1,\hat{Y_2}|X_{2c},X_1)
+\log \det(I_{N_2}+{\rho }_{22}H_{22}Q_{2p}H^{\dagger }_{22}+{\rho }_{12}H_{12}Q_{1p}H^{\dagger }_{12})-N_2 \nonumber\\
&\stackrel{(b)}{\ge}& h(Y_1,\hat{Y_2})-h(Y_1,\hat{Y_2}|X_{2c},X_1)
+\log \det(I_{N_2}+{\rho }_{22}H_{22}Q_{2p}H^{\dagger }_{22})-N_2 \nonumber\\
&=& h(Y_1,\hat{Y_2})-h(\sqrt{{\rho }_{21}}H_{21}X_{2p}+Z_{1},\sqrt{{\rho }_{22}}H_{22}X_{2p}+Z_{2}+\hat{Z_{2}})
+\log \det(I_{N_2}+{\rho }_{22}H_{22}Q_{2p}H^{\dagger }_{22})\nonumber\\
&&-N_2 \nonumber\\
&\stackrel{(c)}{=}& h(Y_1,\hat{Y_2})- \log \det (\Delta+I_{N_2}+H_{22}Q_{2p}H_{22}^{\dagger})- \log \det (I_{N_1}+H_{12}Q_{2p}H_{12}^{\dagger}-\nonumber\\
&&H_{12}Q_{2p}H_{22}^{\dagger}
{(\Delta+I_{N_2}+H_{22}Q_{2p}H_{22}^{\dagger})^{-1}}H_{22}Q_{2p}H_{12}^{\dagger})
+\log \det(I_{N_2}+{\rho }_{22}H_{22}Q_{2p}H^{\dagger }_{22})\nonumber\\
&&-N_2-(N_1+N_2)\log (2\pi e) \nonumber
\end{eqnarray}

\begin{eqnarray}
&\stackrel{(d)}{=}& h(Y_1,\hat{Y_2})- \log \det (I_{N_2}+H_{22}Q_{2p}H_{22}^{\dagger})- \log \det (I_{N_1}+H_{12}Q_{2p}H_{12}^{\dagger}\nonumber\\
&&-H_{12}Q_{2p}H_{22}^{\dagger}
{(\Delta+I_{N_2}+H_{22}Q_{2p}H_{22}^{\dagger})^{-1}}H_{22}Q_{2p}H_{12}^{\dagger})
+\log \det(I_{N_2}+{\rho }_{22}H_{22}Q_{2p}H^{\dagger }_{22})\nonumber\\
&&-2N_2-(N_1+N_2)\log (2\pi e) \nonumber\\
&{=}& h(Y_1,\hat{Y_2})
- \log \det (I_{N_1}+H_{12}Q_{2p}H_{12}^{\dagger}-H_{12}Q_{2p}H_{22}^{\dagger}
{(\Delta+I_{N_2}+H_{22}Q_{2p}H_{22}^{\dagger})^{-1}}H_{22}Q_{2p}H_{12}^{\dagger})\nonumber\\
&&-2N_2-(N_1+N_2)\log (2\pi e) \nonumber\\
&\stackrel{(e)}{\ge}& h(Y_1,\hat{Y_2})
- \log \det (I_{N_1}+H_{12}Q_{2p}H_{12}^{\dagger})-2N_2-(N_1+N_2)\log (2\pi e) \nonumber\\
&\stackrel{(f)}{\ge}& h(Y_1,\hat{Y_2})-N_1-2N_2-(N_1+N_2)\log (2\pi e) \nonumber\\
&=&\log\det\left[ \begin{array}{cc}
I_{N_1}+{\rho }_{11}H_{11}H^{\dagger }_{11}+{\rho }_{21}H_{21}H^{\dagger }_{21} & \sqrt{{\rho }_{11}{\rho }_{12}}H_{11}H^{\dagger }_{12}+\sqrt{{\rho }_{21}{\rho }_{22}}H_{21}H^{\dagger }_{22} \\
\sqrt{{\rho }_{11}{\rho }_{12}}H_{12}H^{\dagger }_{11}+\sqrt{{\rho }_{21}{\rho }_{22}}H_{22}H^{\dagger }_{21} & \Delta+I_{N_2}+{\rho }_{22}H_{22}H^{\dagger }_{22}+{\rho }_{12}H_{12}H^{\dagger }_{12} \end{array}
\right]\nonumber\\
&&-N_1-2N_2-(N_1+N_2)\log (2\pi e) \nonumber\\
&=& \log \det (\Delta+I_{N_2}+{\rho }_{22}H_{22}H^{\dagger }_{22}+{\rho }_{12}H_{12}H^{\dagger }_{12})+ \log \det (I_{N_1}+{\rho }_{11}H_{11}H^{\dagger }_{11}+{\rho }_{21}H_{21}H^{\dagger }_{21}\nonumber\\
&&-\sqrt{{\rho }_{11}{\rho }_{12}}H_{11}H^{\dagger }_{12}+\sqrt{{\rho }_{21}{\rho }_{22}}H_{21}H^{\dagger }_{22}{(\Delta+I_{N_2}+{\rho }_{22}H_{22}H^{\dagger }_{22}+{\rho }_{12}H_{12}H^{\dagger }_{12})^{-1}}\nonumber\\
&&\sqrt{{\rho }_{11}{\rho }_{12}}H_{12}H^{\dagger }_{11}+\sqrt{{\rho }_{21}{\rho }_{22}}H_{22}H^{\dagger }_{21})-N_1-2N_2 \nonumber\\
&\ge& \log \det (I_{N_2}+{\rho }_{22}H_{22}H^{\dagger }_{22}+{\rho }_{12}H_{12}H^{\dagger }_{12})+ \log \det (I_{N_1}+{\rho }_{11}H_{11}H^{\dagger }_{11}+{\rho }_{21}H_{21}H^{\dagger }_{21}\nonumber\\
&&-\sqrt{{\rho }_{11}{\rho }_{12}}H_{11}H^{\dagger }_{12}+\sqrt{{\rho }_{21}{\rho }_{22}}H_{21}H^{\dagger }_{22}{(I_{N_2}+{\rho }_{22}H_{22}H^{\dagger }_{22}+{\rho }_{12}H_{12}H^{\dagger }_{12})^{-1}}\sqrt{{\rho }_{11}{\rho }_{12}}H_{12}H^{\dagger }_{11}\nonumber\\
&&+\sqrt{{\rho }_{21}{\rho }_{22}}H_{22}H^{\dagger }_{21})-N_1-2N_2
\end{eqnarray}
\begin{eqnarray}
&=&\log\det\left[ \begin{array}{cc}
I_{N_1}+{\rho }_{11}H_{11}H^{\dagger }_{11}+{\rho }_{21}H_{21}H^{\dagger }_{21} & \sqrt{{\rho }_{11}{\rho }_{12}}H_{11}H^{\dagger }_{12}+\sqrt{{\rho }_{21}{\rho }_{22}}H_{21}H^{\dagger }_{22} \\
\sqrt{{\rho }_{11}{\rho }_{12}}H_{12}H^{\dagger }_{11}+\sqrt{{\rho }_{21}{\rho }_{22}}H_{22}H^{\dagger }_{21} & I_{N_2}+{\rho }_{22}H_{22}H^{\dagger }_{22}+{\rho }_{12}H_{12}H^{\dagger }_{12} \end{array}
\right]\nonumber\\
&&-N_1-2N_2 \nonumber\\
&=& \log\det\left(I_{N_1+N_2}+
\left[\begin{array}{c}
\sqrt{\rho_{11}}H_{11} \\
\sqrt{\rho_{12}}H_{12} \end{array}\right]
[\sqrt{\rho_{11}}H_{11}^{\dagger}\ \sqrt{\rho_{12}}H_{12}^{\dagger}]+
\left[\begin{array}{c}
\sqrt{\rho_{21}}H_{21} \\
\sqrt{\rho_{22}}H_{22} \end{array}\right]
[\sqrt{\rho_{21}}H_{21}^{\dagger}\ \sqrt{\rho_{22}}H_{22}^{\dagger}]\right)\nonumber\\
&&-N_1-2N_2 \nonumber\\
&=& h(Y_1,Y_2)-N_1-2N_2 \nonumber\\
&{=}&\log\det\left(I_{N_1+N_2}+
\left[\begin{array}{c}
\sqrt{\rho_{11}}H_{11} \\
\sqrt{\rho_{12}}H_{12} \end{array}\right]
[\sqrt{\rho_{11}}H_{11}^{\dagger}\ \sqrt{\rho_{12}}H_{12}^{\dagger}]+
\left[\begin{array}{c}
\sqrt{\rho_{21}}H_{21} \\
\sqrt{\rho_{22}}H_{22} \end{array}\right]
[\sqrt{\rho_{21}}H_{21}^{\dagger}\ \sqrt{\rho_{22}}H_{22}^{\dagger}]\right)\nonumber\\
&&-N_1-2N_2,
\end{eqnarray}
where $(a)$ and $(f)$ follow from \eqref{U_B}, and $(b)$ and $(e)$ follow from the fact that $\log\det(.)$ is a monotonically increasing function on the cone of positive definite matrices, $(c)$ follows from the fact that
\begin{eqnarray}\label{piip}
&&h(\sqrt{{\rho }_{21}}H_{21}X_{2p}+Z_{1},\sqrt{{\rho }_{22}}H_{22}X_{2p}+Z_{2}+\hat{Z_{2}})\nonumber\\
&=&\log \det\left[ \begin{array}{cc}
I_{N_1}+{\rho }_{21}H_{21}Q_{2p}H^{\dagger }_{21} & \sqrt{{\rho }_{21}{\rho }_{22}}H_{21}Q_{2p}H^{\dagger }_{22} \\
\sqrt{{\rho }_{21}{\rho }_{22}}H_{22}Q_{2p}H^{\dagger }_{21} & \Delta+I_{N_2}+{\rho }_{22}H_{22}Q_{2p}H^{\dagger }_{22} \end{array}
\right]+(N_1+N_2)\log (2\pi e)\nonumber\\
&=& \log \det (\Delta+I_{N_2}+H_{22}Q_{2p}H_{22}^{\dagger})+ \log \det (I_{N_1}+H_{12}Q_{2p}H_{12}^{\dagger}-\nonumber\\
&&H_{12}Q_{2p}H_{22}^{\dagger}
{(\Delta+I_{N_2}+H_{22}Q_{2p}H_{22}^{\dagger})^{-1}}H_{22}Q_{2p}H_{12}^{\dagger})+(N_1+N_2)\log (2\pi e),
\end{eqnarray}

and $(d)$ follows from the fact that $\Delta=I_{N_2}+H_{22}Q_{2p}H_{22}^{\dagger}$ and hence:
\begin{eqnarray}
&&\log \det (\Delta+I_{N_2}+H_{22}Q_{2p}H_{22}^{\dagger})\nonumber\\
&=& \log \det 2(I_{N_2}+H_{22}Q_{2p}H_{22}^{\dagger})\nonumber\\
&=& \log \det (I_{N_2}+H_{22}Q_{2p}H_{22}^{\dagger})+N_2.
\end{eqnarray}

Thus, we see that this $R_1+R_2$ bound is within $N_1+2N_2$ bits of the outer bound in \eqref{ro0eq6}.

\noindent {\bf \eqref{r1r2.3}:} For this bound in Lemma \ref{inner}, we have

\begin{eqnarray}
&&I(X_{2c},X_1;Y_1|X_{1c})+I(X_{1c},X_2;Y_2|X_{2c})+C_{12}+{(C_{21}-\xi)}^{+}\nonumber\\
&=& h(Y_1|X_{1c})-h(Y_1|X_{1c},X_{2c},X_1)+h(Y_2|X_{2c})-h(Y_2|X_{2c},X_{1c},X_2)
+C_{12}+{(C_{21}-\xi)}^{+}\nonumber\\
&\stackrel{(a)}{\ge}& h(\sqrt{{\rho }_{11}}H_{11}X_{1p}+\sqrt{{\rho }_{21}}H_{21}X_{2}+Z_{1})
+h(\sqrt{{\rho }_{12}}H_{12}X_{1}+\sqrt{{\rho }_{22}}H_{22}X_{2p}+Z_{2})\nonumber\\
&&+C_{12}+C_{21}-N_1-2N_2-(N_1+N_2)\log (2\pi e)\nonumber\\
&=& \log \det(I_{N_1}+{\rho }_{11}H_{11}Q_{1p}H_{11}^{\dagger}+{\rho }_{21}H_{21}H_{21}^{\dagger})
+\log \det(I_{N_2}+{\rho }_{22}H_{22}Q_{2p}H_{22}^{\dagger}+{\rho }_{12}H_{12}H_{12}^{\dagger})\nonumber\\
&&+C_{12}+C_{21}-N_1-2N_2\nonumber\\
&=&\log  \det (I_{N_1}+{\rho}_{11}H_{11}H^{\dagger}_{11}+
{\rho}_{21}H_{21}H^{\dagger}_{21}-
{\rho}_{11}{\rho}_{12}H_{11}H^{\dagger}_{12}
{(I_{N_2}+{\rho}_{12}H_{12}H^{\dagger}_{12})^{-1}}H_{12}H^{\dagger}_{11})+\nonumber\\
&&\log  \det (I_{N_2}+{\rho}_{22}H_{22}H^{\dagger}_{22}+
{\rho}_{12}H_{12}H^{\dagger}_{12}-
{\rho}_{22}{\rho}_{21}H_{22}H^{\dagger}_{21}
{(I_{N_1}+{\rho}_{21}H_{21}H^{\dagger}_{21})^{-1}}H_{21}H^{\dagger}_{22})+\nonumber\\
&&C_{12}+C_{21}-N_1-2N_2.
\end{eqnarray}
where $(a)$ follows from \eqref{U_B} and \eqref{xi}.

Thus, we see that this $R_1+R_2$ bound is within $N_1+2N_2$ bits of the outer bound in \eqref{ro0eq3}.

\noindent {\bf \eqref{r1r2.4}:} For this bound in Lemma \ref{inner}, we have

\begin{eqnarray}
&&I(X_{2c},X_1;Y_1,\hat{Y_2}|X_{1c})+I(X_{1c},X_2;Y_2|X_{2c})+C_{12}\nonumber\\
&=& I(X_{2c};Y_1,\hat{Y_2}|X_{1c})+I(X_1;Y_1,\hat{Y_2}|X_{1c},X_{2c})
+I(X_{1c},X_2;Y_2|X_{2c})+C_{12}\nonumber\\
&\ge& I(X_{2c};\hat{Y_2}|X_{1c})+I(X_1;Y_1|X_{1c},X_{2c})
+I(X_{1c},X_2;Y_2|X_{2c})+C_{12}\nonumber\\
&\stackrel{(a)}{\ge}& I(X_{2c};{Y_2}|X_{1c})-N_2+I(X_1;Y_1|X_{1c},X_{2c})
+I(X_{1c},X_2;Y_2|X_{2c})+C_{12}\nonumber\\
&\stackrel{(b)}{\ge}& I(X_1;Y_1|X_{1c},X_{2c})
+I(X_{1c},X_2;Y_2)+C_{12}-N_2\nonumber\\
&=& h(Y_1|X_{1c},X_{2c})-h(Y_1|X_1,X_{1c},X_{2c})
+h(Y_2)-h(Y_2|X_{1c},X_2)+C_{12}-N_2\nonumber\\
&\stackrel{(c)}{\ge}& h(Y_1|X_{1c},X_{2c})+h(Y_2)
+C_{12}-N_1-2N_2-(N_1+N_2)\log (2\pi e)\nonumber\\
&=& \log  \det (I_{N_1}+{\rho}_{11}H_{11}Q_{1p}H^{\dagger}_{11}+{\rho}_{21}H_{21}Q_{2p}H^{\dagger}_{21})+\nonumber\\
&&\log  \det (I_{N_2}+{\rho}_{22}H_{22}H^{\dagger}_{22}+
{\rho}_{12}H_{12}H^{\dagger}_{12})+C_{12}-N_1-2N_2\nonumber\\
&\stackrel{(d)}{\ge}& \log  \det (I_{N_1}+{\rho}_{11}H_{11}Q_{1p}H^{\dagger}_{11})+\nonumber\\
&&\log  \det (I_{N_2}+{\rho}_{22}H_{22}H^{\dagger}_{22}+
{\rho}_{12}H_{12}H^{\dagger}_{12})+C_{12}-N_1-2N_2\nonumber\\
&=& \log  \det (I_{N_1}+{\rho}_{11}H_{11}H^{\dagger}_{11}-
{\rho}_{11}{\rho}_{12}H_{11}H^{\dagger}_{12}
{(I_{N_2}+{\rho}_{12}H_{12}H^{\dagger}_{12})^{-1}}H_{12}H^{\dagger}_{11})+\nonumber\\
&&\log  \det (I_{N_2}+{\rho}_{22}H_{22}H^{\dagger}_{22}+
{\rho}_{12}H_{12}H^{\dagger}_{12})+C_{12}-N_1-2N_2,
\end{eqnarray}
where $(a)$ follows from
\begin{eqnarray}
I(X_{2c};\hat{Y_2}|X_{1c}) \ge I(X_{2c};{Y_2}|X_{1c})-N_2,
\end{eqnarray}
which is true since
\begin{eqnarray}
&&I(X_{2c};\hat{Y_2}|X_{1c})-I(X_{2c};{Y_2}|X_{1c})+N_2\nonumber\\
&=& h(\hat{Y_2}|X_{1c})-h(\hat{Y_2}|X_{1c},X_{2c})-h({Y_2}|X_{1c})+h({Y_2}|X_{1c},X_{2c})+N_2\nonumber\\
&=& \log \det (\Delta+I_{N_2}+{\rho}_{22}H_{22}H^{\dagger}_{22}+
{\rho}_{12}Q_{1p}H_{12}H^{\dagger}_{12})\nonumber\\
&&-\log \det (\Delta+I_{N_2}+{\rho}_{22}H_{22}Q_{2p}H^{\dagger}_{22}+
{\rho}_{12}H_{12}Q_{1p}H^{\dagger}_{12})\nonumber\\
&&-\log \det (I_{N_2}+{\rho}_{22}H_{22}H^{\dagger}_{22}+
{\rho}_{12}H_{12}Q_{1p}H^{\dagger}_{12})\nonumber\\
&&+\log \det (I_{N_2}+{\rho}_{22}H_{22}Q_{2p}H^{\dagger}_{22}+
{\rho}_{12}H_{12}Q_{1p}H^{\dagger}_{12})+N_2\nonumber
\end{eqnarray}

\begin{eqnarray}
&=& \log \det (\Delta+I_{N_2}+{\rho}_{22}H_{22}H^{\dagger}_{22}+
{\rho}_{12}Q_{1p}H_{12}H^{\dagger}_{12})\nonumber\\
&&-\log \det (I_{N_2}+{\rho}_{22}H_{22}H^{\dagger}_{22}+
{\rho}_{12}H_{12}Q_{1p}H^{\dagger}_{12})\nonumber\\
&&-\log \det (2\Delta+{\rho}_{12}H_{12}Q_{1p}H^{\dagger}_{12})\nonumber\\
&&+\log \det (\Delta+{\rho}_{12}H_{12}Q_{1p}H^{\dagger}_{12})+N_2\nonumber\\
&\ge& 0,
\end{eqnarray}
$(b)$ follows from the fact that
\begin{eqnarray}
&&I(X_{2c};{Y_2}|X_{1c})+I(X_{1c},X_2;Y_2|X_{2c})\nonumber\\
&=& I(X_{2c};{Y_2},X_{1c})+I(X_{1c},X_2;Y_2|X_{2c})\nonumber\\
&\ge& I(X_{2c};{Y_2})+I(X_{1c},X_2;Y_2|X_{2c})\nonumber\\
&=& I(X_{1c},X_2,X_{2c};{Y_2})+I(X_{1c},X_2;Y_2),
\end{eqnarray}
$(c)$ follows from \eqref{U_B} and $(d)$ follows from the fact that $\log\det(.)$ is a monotonically increasing function on the cone of positive definite matrices.

Thus, we see that this $R_1+R_2$ bound is within $N_1+2N_2$ bits of the outer bound in \eqref{ro0eq4}.

\noindent {\bf \eqref{r1r2.5}:} For this bound in Lemma \ref{inner}, similar to the last term we have
\begin{eqnarray}
&&I(X_1;Y_1|X_{1c},X_{2c})+I(X_{1c},X_2;Y_2)+C_{12}\nonumber\\
&\ge& \log  \det (I_{N_1}+{\rho}_{11}H_{11}H^{\dagger}_{11}-
{\rho}_{11}{\rho}_{12}H_{11}H^{\dagger}_{12}
{(I_{N_2}+{\rho}_{12}H_{12}H^{\dagger}_{12})^{-1}}H_{12}H^{\dagger}_{11})+\nonumber\\
&&\log  \det (I_{N_2}+{\rho}_{22}H_{22}H^{\dagger}_{22}+
{\rho}_{12}H_{12}H^{\dagger}_{12})+C_{12}-N_1-N_2,
\end{eqnarray}
which results from the proof of the last bound.

Thus, we see that this $R_1+R_2$ bound is within $N_1+N_2$ bits of the outer bound in \eqref{ro0eq4}.

\noindent {\bf \eqref{r1r2.6}:} For this bound in Lemma \ref{inner} we have
\begin{eqnarray}
&&I(X_1;Y_1|X_{1c},X_{2c})+I(X_{2c};Y_1|X_1)+I(X_{1c},X_2;Y_2|X_{2c})+C_{12}\nonumber\\
&=& h(Y_1|X_{1c},X_{2c})-h(Y_1|X_{1c},X_{2c},X_1)+h(Y_1|X_1)-h(Y_1|X_1 ,X_{2c})+h(Y_2|X_{2c})\nonumber\\
&&-h(Y_2|X_{1c},X_2,X_{2c})+C_{12}\nonumber\\
&\stackrel{(a)}{\ge}&
h(Y_1|X_{1c},X_{2c})+h(Y_1|X_1)+h(Y_2|X_{2c})+C_{12}-2N_1-N_2-(2N_1+N_2)\log (2\pi e)\nonumber\\
&=&
h(Y_1|X_{1c},X_{2c})+h(\sqrt{{\rho }_{21}}H_{21}X_{2}+Z_{1})+h(Y_2,X_{2c})-h(X_{2c})+C_{12}-2N_1-N_2-(2N_1+N_2)\log (2\pi e)\nonumber\\
&=&\log\det\left[\begin{array}{cc}
{\rho }_{22}H_{22}H^{\dagger }_{22}+{\rho }_{12}H_{12}H^{\dagger }_{12}
& \sqrt{{\rho }_{22}}H_{22}Q_{2c} \\
\sqrt{{\rho }_{22}}Q_{2c}H^{\dagger }_{22} &
Q_{2c} \end{array}\right]
-\log \det \left(Q_{2c}\right)
+h(Y_1|X_{1c},X_{2c})+\nonumber\\
&&h(\sqrt{{\rho }_{21}}H_{21}X_{2}+Z_{1})+C_{12}-2N_1-N_2-(2N_1)\log (2\pi e)\nonumber\\
&=&
\log\det\left(I_{N_{2}}+{\rho }_{22}H_{22}H^{\dagger }_{22}+{\rho }_{12}H_{12}H^{\dagger }_{12}\right)+\nonumber\\
&&\log\det\left(Q_{2c}-
{\rho }_{22}Q_{2c}H^{\dagger }_{22}
(I_{N_{2}}+{\rho }_{22}H_{22}H^{\dagger }_{22}+{\rho }_{12}H_{12}H^{\dagger }_{12})^{-1}
H_{22}Q_{2c}
\right)
-\log \det \left(Q_{2c}\right)+\nonumber\\
&&h(Y_1|X_{1c},X_{2c})+h(\sqrt{{\rho }_{21}}H_{21}X_{2}+Z_{1})+C_{12}-2N_1-N_2-(2N_1)\log (2\pi e)\nonumber\\
&=&
\log\det\left(I_{N_{2}}+{\rho }_{22}H_{22}H^{\dagger }_{22}+{\rho }_{12}H_{12}H^{\dagger }_{12}\right)+
\log \det (I_{N_1}+{\rho }_{11}H_{11}Q_{1p}H^{\dagger }_{11}+{\rho }_{21}H_{21}Q_{2p}H^{\dagger }_{21})+\nonumber\\
&&\log\det\left(Q_{2c}-
{\rho }_{22}Q_{2c}H^{\dagger }_{22}
(I_{N_{2}}+{\rho }_{22}H_{22}H^{\dagger }_{22}+{\rho }_{12}H_{12}H^{\dagger }_{12})^{-1}
H_{22}Q_{2c}
\right)
-\log \det \left(Q_{2c}\right)+\nonumber\\
&&+\log\det\left(I_{N_{1}}+{\rho }_{21}H_{21}H^{\dagger }_{21}\right)
+C_{12}-2N_1-N_2\nonumber\\
&\ge&
\log\det\left(I_{N_{2}}+{\rho }_{22}H_{22}H^{\dagger }_{22}+{\rho }_{12}H_{12}H^{\dagger }_{12}\right)+
\log \det (I_{N_1}+{\rho }_{11}H_{11}Q_{1p}H^{\dagger }_{11})+\nonumber\\
&&\log\det\left(Q_{2c}-
{\rho }_{22}Q_{2c}H^{\dagger }_{22}
(I_{N_{2}}+{\rho }_{22}H_{22}H^{\dagger }_{22})^{-1}
H_{22}Q_{2c}
\right)
-\log \det \left(Q_{2c}\right)+\nonumber\\
&&+\log\det\left(I_{N_{1}}+{\rho }_{21}H_{21}H^{\dagger }_{21}\right)
+C_{12}-2N_1-N_2\nonumber\\
&\ge&
\log\det\left(I_{N_{2}}+{\rho }_{22}H_{22}H^{\dagger }_{22}+{\rho }_{12}H_{12}H^{\dagger }_{12}\right)+
\log \det (I_{N_1}+{\rho }_{11}H_{11}Q_{1p}H^{\dagger }_{11})+\nonumber\\
&&\log\det\left(Q_{2c}-
Q^{2}_{2c}\right)
-\log \det \left(Q_{2c}\right)+\log\det\left(I_{N_{1}}+{\rho }_{21}H_{21}H^{\dagger }_{21}\right)
+C_{12}-2N_1-N_2\nonumber\\
&=&
\log\det\left(I_{N_{2}}+{\rho }_{22}H_{22}H^{\dagger }_{22}+{\rho }_{12}H_{12}H^{\dagger }_{12}\right)+
\log \det (I_{N_1}+{\rho }_{11}H_{11}Q_{1p}H^{\dagger }_{11})+\nonumber\\
&&+\log \det \left(Q_{2p}\right)+\log\det\left(I_{N_{1}}+{\rho }_{21}H_{21}H^{\dagger }_{21}\right)
+C_{12}-2N_1-N_2\nonumber\\
&=&
\log\det\left(I_{N_{2}}+{\rho }_{22}H_{22}H^{\dagger }_{22}+{\rho }_{12}H_{12}H^{\dagger }_{12}\right)+
\log \det (I_{N_1}+{\rho }_{11}H_{11}Q_{1p}H^{\dagger }_{11})+\nonumber\\
&&+\log \det \left(Q_{2p}\right)+\log\det\left(I_{M_{2}}+{\rho }_{21}H^{\dagger }_{21}H_{21}\right)
+C_{12}-2N_1-N_2\nonumber\\
&=&
\log\det\left(I_{N_{2}}+{\rho }_{22}H_{22}H^{\dagger }_{22}+{\rho }_{12}H_{12}H^{\dagger }_{12}\right)+
\log \det (I_{N_1}+{\rho }_{11}H_{11}Q_{1p}H^{\dagger }_{11})+\nonumber\\
&&+\log \det \left(I_{M_{2}}-{\rho }_{21}H^{\dagger }_{21}
(I_{N_{1}}+{\rho }_{21}H_{21}H^{\dagger }_{21})^{-1}
H_{21}\right)+\log\det\left(I_{M_{2}}+{\rho }_{21}H^{\dagger }_{21}H_{21}\right)
+C_{12}-2N_1-N_2\nonumber
\end{eqnarray}
\begin{eqnarray}\label{vvv}
&=&
\log\det\left(I_{N_{2}}+{\rho }_{22}H_{22}H^{\dagger }_{22}+{\rho }_{12}H_{12}H^{\dagger }_{12}\right)+
\log \det (I_{N_1}+{\rho }_{11}H_{11}Q_{1p}H^{\dagger }_{11})+\nonumber\\
&&+\log \det \left(I_{M_{2}}-{\rho }_{21}H^{\dagger }_{21}
(I_{N_{1}}+{\rho }_{21}H_{21}H^{\dagger }_{21})^{-1}
H_{21}+{\rho }_{21}H^{\dagger }_{21}H_{21}-\right.\nonumber\\
&&\left.{\rho }_{21}H^{\dagger }_{21}H_{21}
{\rho }_{21}H^{\dagger }_{21}
(I_{N_{1}}+{\rho }_{21}H_{21}H^{\dagger }_{21})^{-1}
H_{21}\right)+
C_{12}-2N_1-N_2\nonumber\\
&=&
\log\det\left(I_{N_{2}}+{\rho }_{22}H_{22}H^{\dagger }_{22}+{\rho }_{12}H_{12}H^{\dagger }_{12}\right)+
\log \det (I_{N_1}+{\rho }_{11}H_{11}Q_{1p}H^{\dagger }_{11})+\nonumber\\
&&+\log \det \left(I_{M_{2}}
+{\rho }_{21}H^{\dagger }_{21}
\left(
I_{M_{2}}
-H_{21}
{\rho }_{21}H^{\dagger }_{21}
(I_{N_{1}}+{\rho }_{21}H_{21}H^{\dagger }_{21})^{-1}
-(I_{N_{1}}+{\rho }_{21}H_{21}H^{\dagger }_{21})^{-1}
\right)H_{21}
\right)\nonumber\\
&&+C_{12}-2N_1-N_2\nonumber\\
&\stackrel{(c)}{=}& \log \det (I_{N_1}+{\rho }_{11}H_{11}Q_{1p}H^{\dagger }_{11})
+\log \det (I_{N_2}+{\rho }_{22}H_{22}H^{\dagger }_{22}+{\rho }_{12}H_{12}H^{\dagger }_{12})+C_{12}-2N_1-N_2\nonumber\\
&=& \log  \det (I_{N_1}+{\rho}_{11}H_{11}H^{\dagger}_{11}-
{\rho}_{11}{\rho}_{12}H_{11}H^{\dagger}_{12}
{(I_{N_2}+{\rho}_{12}H_{12}H^{\dagger}_{12})^{-1}}H_{12}H^{\dagger}_{11})+\nonumber\\
&&\log  \det (I_{N_2}+{\rho}_{22}H_{22}H^{\dagger}_{22}+
{\rho}_{12}H_{12}H^{\dagger}_{12})+C_{12}-2N_1-N_2,
\end{eqnarray}
where $(a)$ follows from \eqref{U_B}, $(b)$ follows from Lemma \ref{lem_block} and $(c)$ follows from the fact that $\log\det(.)$ is a monotonically increasing function on the cone of positive definite matrices.

Thus, we see that this $R_1+R_2$ bound is within $2N_1+N_2$ bits of the outer bound in \eqref{ro0eq4}.

\noindent {\bf \eqref{2r1r2.1}:} For this bound in Lemma \ref{inner} we have
\begin{eqnarray}\label{longg}
&&I(X_1,X_{2c};Y_1)+I(X_1;Y_1|X_{1c},X_{2c})+
I(X_{1c},X_2;Y_2|X_{2c})+C_{12}+{(C_{21}-\xi)}^{+}\nonumber\\
&=& h(Y_1)-h(Y_1|X_1,X_{2c})+h(Y_1|X_{1c},X_{2c})-h(Y_1|X_{1c},X_{2c},X_1)+
h(Y_2|X_{2c})\nonumber\\
&&-h(Y_2|X_{1c},X_2,X_{2c})+C_{12}+{(C_{21}-\xi)}^{+}\nonumber\\
&\stackrel{(a)}{\ge}& h(Y_1)+h(Y_1|X_{1c},X_{2c})+
h(Y_2|X_{2c})+C_{12}+C_{21}-2{N_1}-2{N_2}-(2N_1+N_2)\log (2\pi e)\nonumber\\
&=& \log  \det (I_{N_1}+{\rho}_{11}H_{11}H^{\dagger}_{11}+
{\rho}_{21}H_{21}H^{\dagger}_{21})
+\log  \det (I_{N_1}+{\rho}_{11}H_{11}Q_{1p}H^{\dagger}_{11}+
{\rho}_{21}H_{21}Q_{2p}H^{\dagger}_{21})\nonumber\\
&&+\log  \det (I_{N_2}+{\rho}_{22}H_{22}Q_{2p}H^{\dagger}_{22}+
{\rho}_{12}H_{12}H^{\dagger}_{12})
+C_{12}+C_{21}-2{N_1}-2{N_2}\nonumber\\
&\stackrel{(b)}{\ge}& \log  \det (I_{N_1}+{\rho}_{11}H_{11}H^{\dagger}_{11}+
{\rho}_{21}H_{21}H^{\dagger}_{21})
+\log  \det (I_{N_1}+{\rho}_{11}H_{11}Q_{1p}H^{\dagger}_{11})\nonumber\\
&&+\log  \det (I_{N_2}+{\rho}_{22}H_{22}Q_{2p}H^{\dagger}_{22}+
{\rho}_{12}H_{12}H^{\dagger}_{12})
+C_{12}+C_{21}-2{N_1}-2{N_2}\nonumber\\
&=& \log  \det (I_{N_1}+{\rho}_{11}H_{11}H^{\dagger}_{11}-
{\rho}_{11}{\rho}_{12}H_{11}H^{\dagger}_{12}
{(I_{N_2}+{\rho}_{12}H_{12}H^{\dagger}_{12})^{-1}}H_{12}H^{\dagger}_{11})+\nonumber\\
&&\log  \det (I_{N_2}+{\rho}_{22}H_{22}H^{\dagger}_{22}+
{\rho}_{12}H_{12}H^{\dagger}_{12}-
{\rho}_{22}{\rho}_{21}H_{22}H^{\dagger}_{21}
{(I_{N_1}+{\rho}_{21}H_{21}H^{\dagger}_{21})^{-1}}H_{21}H^{\dagger}_{22})+\nonumber\\
&&\log  \det (I_{N_1}+{\rho}_{11}H_{11}H^{\dagger}_{11}+
{\rho}_{21}H_{21}H^{\dagger}_{21})+C_{12}+C_{21}-2{N_1}-2{N_2},
\end{eqnarray}
where $(a)$ follows from \eqref{U_B} and \eqref{xi}, and $(b)$ follows from the fact that $\log\det(.)$ is a monotonically increasing function on the cone of positive definite matrices.

Thus, we see that this $2R_1+R_2$ bound is within $2N_1+2N_2$ bits of the outer bound in \eqref{ro0eq7}.

\noindent {\bf \eqref{r12r2.1}:} For this bound in Lemma \ref{inner} we have
\begin{eqnarray}
&&I(X_1,X_{2c};Y_1|X_{1c})+I(X_{1c},X_{2};Y_2)
+I(X_2;Y_2|X_{1c},X_{2c})+C_{12}+{(C_{21}-\xi)}^{+}\nonumber\\
&=& h(Y_1|X_{1c})-h(Y_1|X_{1c},X_1,X_{2c})+h(Y_2)-h(Y_2|X_{1c},X_{2})
+h(Y_2|X_{1c},X_{2c})\nonumber\\
&&-h(Y_2|X_{1c},X_{2c},X_2)+C_{12}+{(C_{21}-\xi)}^{+}\nonumber\\
&\stackrel{(a)}{\ge}& h(Y_1|X_{1c})+h(Y_2)+h(Y_2|X_{1c},X_{2c})+C_{12}+C_{21}-2{N_1}-2{N_2}-(N_1+2N_2)\log (2\pi e)\nonumber\\
&=& \log \det (I_{N_1}+{\rho }_{11}H_{11}Q_{1p}H^{\dagger }_{11}+{\rho }_{21}H_{21}H^{\dagger }_{21})
+\log \det (I_{N_2}+{\rho }_{22}H_{22}H^{\dagger }_{22}+{\rho }_{12}H_{12}H^{\dagger }_{12})\nonumber\\
&&+\log \det (I_{N_2}+{\rho }_{22}H_{22}Q_{2p}H^{\dagger }_{22}+{\rho }_{12}H_{12}Q_{1p}H^{\dagger }_{12})+C_{12}+C_{21}-2{N_1}-2{N_2}\nonumber\\
&\stackrel{(b)}{\ge}&
\log \det (I_{N_1}+{\rho }_{11}H_{11}Q_{1p}H^{\dagger }_{11}+{\rho }_{21}H_{21}H^{\dagger }_{21})
+\log \det (I_{N_2}+{\rho }_{22}H_{22}H^{\dagger }_{22}+{\rho }_{12}H_{12}H^{\dagger }_{12})\nonumber\\
&&+\log \det (I_{N_2}+{\rho }_{22}H_{22}Q_{2p}H^{\dagger }_{22})+C_{12}+C_{21}-2{N_1}-2{N_2}\nonumber\\
&=& \log  \det (I_{N_2}+{\rho}_{22}H_{22}H^{\dagger}_{22}-
{\rho}_{22}{\rho}_{21}H_{22}H^{\dagger}_{21}
{(I_{N_1}+{\rho}_{21}H_{21}H^{\dagger}_{21})^{-1}}H_{21}H^{\dagger}_{22})+\nonumber\\
&&\log  \det (I_{N_1}+{\rho}_{11}H_{11}H^{\dagger}_{11}+
{\rho}_{21}H_{21}H^{\dagger}_{21}-
{\rho}_{11}{\rho}_{12}H_{11}H^{\dagger}_{12}
{(I_{N_2}+{\rho}_{12}H_{12}H^{\dagger}_{12})^{-1}}H_{12}H^{\dagger}_{11})+\nonumber\\
&&\log  \det (I_{N_2}+{\rho}_{22}H_{22}H^{\dagger}_{22}+
{\rho}_{12}H_{12}H^{\dagger}_{12})+C_{21}+C_{12}-2{N_1}-2{N_2},
\end{eqnarray}
where $(a)$ follows from \eqref{U_B} and \eqref{xi}, and $(b)$ follows from the fact that $\log\det(.)$ is a monotonically increasing function on the cone of positive definite matrices.

Thus, we see that this $R_1+2R_2$ bound is within $2N_1+2N_2$ bits of the outer bound in \eqref{ro0eq8}.

\noindent {\bf \eqref{r12r2.2}:} For this bound in Lemma \ref{inner} we have
\begin{eqnarray}
&& I(X_1,X_{2c};Y_1|X_{1c})+I(X_{2c};Y_1|X_{1})+I(X_{1c},X_{2};Y_2|X_{2c})
+I(X_2;Y_2|X_{1c},X_{2c})\nonumber\\
&&+C_{12}+{(C_{21}-\xi)}^{+}\nonumber\\
&=& h(Y_1|X_{1c})-h(Y_1|X_1,X_{2c},X_{1c})+h(Y_1|X_{1})-h(Y_1|X_{1},X_{2c})+h(Y_2|X_{2c})\nonumber\\
&&-h(Y_2|X_{1c},X_{2},X_{2c})
+h(Y_2|X_{1c},X_{2c})-h(Y_2|X_{1c},X_{2c},X_2)+C_{12}+{(C_{21}-\xi)}^{+}\nonumber\\
&\stackrel{(a)}{\ge}& h(Y_1|X_{1c})+h(Y_1|X_{1})+h(Y_2|X_{2c})
+h(Y_2|X_{1c},X_{2c})+C_{12}+C_{21}-2N_1-3N_2-2(N_1+N_2)\log (2\pi e)\nonumber\\
&=& \log  \det (I_{N_2}+{\rho}_{22}H_{22}Q_{2p}H^{\dagger}_{22}+{\rho}_{12}H_{12}Q_{1p}H^{\dagger}_{12})+\log  \det (I_{N_1}+{\rho}_{11}H_{11}Q_{1p}H^{\dagger}_{11}+
{\rho}_{21}H_{21}H^{\dagger}_{21})+\nonumber\\
&& \log \det (I_{N_1}+{\rho }_{21}H_{21}H^{\dagger }_{21})+\log \det (I_{N_2}+{\rho }_{22}H_{22}Q_{2p}H^{\dagger }_{22}+{\rho }_{12}H_{12}H^{\dagger }_{12})+\nonumber\\
&&C_{21}+C_{12}-2{N_1}-3{N_2}\nonumber
\end{eqnarray}

\begin{eqnarray}
&\stackrel{(b)}{\ge}& \log  \det (I_{N_2}+{\rho}_{22}H_{22}Q_{2p}H^{\dagger}_{22})+\log  \det (I_{N_1}+{\rho}_{11}H_{11}Q_{1p}H^{\dagger}_{11}+
{\rho}_{21}H_{21}H^{\dagger}_{21})+\nonumber\\
&& \log \det (I_{N_1}+{\rho }_{21}H_{21}H^{\dagger }_{21})+\log \det (I_{N_2}+{\rho }_{22}H_{22}Q_{2p}H^{\dagger }_{22}+{\rho }_{12}H_{12}H^{\dagger }_{12})+\nonumber\\
&&C_{21}+C_{12}-2{N_1}-3{N_2}\nonumber\\
&\stackrel{(c)}{\ge}& \log  \det (I_{N_2}+{\rho}_{22}H_{22}Q_{2p}H^{\dagger}_{22})+\log  \det (I_{N_1}+{\rho}_{11}H_{11}Q_{1p}H^{\dagger}_{11}+
{\rho}_{21}H_{21}H^{\dagger}_{21})+\nonumber\\
&&\log  \det (I_{N_2}+{\rho}_{22}H_{22}H^{\dagger}_{22}+
{\rho}_{12}H_{12}H^{\dagger}_{12})+C_{21}+C_{12}-2{N_1}-3{N_2}\nonumber\\
&=& \log  \det (I_{N_2}+{\rho}_{22}H_{22}H^{\dagger}_{22}-
{\rho}_{22}{\rho}_{21}H_{22}H^{\dagger}_{21}
{(I_{N_1}+{\rho}_{21}H_{21}H^{\dagger}_{21})^{-1}}H_{21}H^{\dagger}_{22})+\nonumber\\
&&\log  \det (I_{N_1}+{\rho}_{11}H_{11}H^{\dagger}_{11}+
{\rho}_{21}H_{21}H^{\dagger}_{21}-
{\rho}_{11}{\rho}_{12}H_{11}H^{\dagger}_{12}
{(I_{N_2}+{\rho}_{12}H_{12}H^{\dagger}_{12})^{-1}}H_{12}H^{\dagger}_{11})+\nonumber\\
&&\log  \det (I_{N_2}+{\rho}_{22}H_{22}H^{\dagger}_{22}+
{\rho}_{12}H_{12}H^{\dagger}_{12})+C_{21}+C_{12}-2{N_1}-3{N_2},
\end{eqnarray}
where $(a)$ follows from \eqref{U_B} and \eqref{xi}, and $(b)$ follows from the fact that $\log\det(.)$ is a monotonically increasing function on the cone of positive definite matrices and $(c)$ follows from
\begin{eqnarray}
&& \log \det (I_{N_1}+{\rho }_{21}H_{21}H^{\dagger }_{21})\nonumber\\
&=&h(Y_1|X_1)\nonumber\\
&\stackrel{(d)}{\ge}& h(Y_2)-h(Y_2|X_{2c})\nonumber\\
&=& \log  \det (I_{N_2}+{\rho}_{22}H_{22}H^{\dagger}_{22}+{\rho}_{12}H_{12}H^{\dagger}_{12})-\nonumber\\
&&\log \det (I_{N_2}+{\rho }_{22}H_{22}Q_{2p}H^{\dagger }_{22}+{\rho}_{12}H_{12}H^{\dagger }_{12}),
\end{eqnarray}
where $(d)$ follows from \eqref{vvv}.

Thus, we see that this $R_1+2R_2$ bound is within $2N_1+3N_2$ bits of the outer bound in \eqref{ro0eq8}.

\noindent {\bf \eqref{r12r2.3}:} For this bound in Lemma \ref{inner} we have
\begin{eqnarray}
&& I(X_1,X_{2c};Y_1,\hat {Y_2}|X_{1c})+I(X_{1c},X_2;Y_2)
+I(X_2;Y_2|X_{1c},X_{2c})+C_{12}\nonumber\\
&=& h(Y_1,\hat {Y_2}|X_{1c})-h(Y_1,\hat {Y_2}|X_1,X_{2c},X_{1c})+h(Y_2)-h(Y_2| X_{1c},X_2)
+h(Y_2|X_{1c},X_{2c})\nonumber\\
&&-h(Y_2|X_{1c},X_{2c},X_2)+C_{12}\nonumber\\
&\stackrel{(a)}{\ge}& h(Y_1,\hat {Y_2}|X_{1c})-h(Y_1,\hat {Y_2}|X_1,X_{2c},X_{1c})
+h(Y_2)
+h(Y_2|X_{1c},X_{2c})+C_{12}-2{N_2}-2N_2\log (2\pi e)\nonumber\\
&=& h(Y_1,\hat {Y_2}|X_{1c})-h(\sqrt{{\rho }_{21}}H_{21}X_{2p}+Z_{1},\sqrt{{\rho }_{22}}H_{22}X_{2p}+Z_{2}+\hat{Z_{2}})
+h(Y_2|X_{1c},X_{2c})+h(Y_2)\nonumber\\
&&+C_{12}-2{N_2}-2N_2\log (2\pi e)\nonumber
\end{eqnarray}

\begin{eqnarray}
&\stackrel{(b)}{=}&
h(Y_1,\hat {Y_2}|X_{1c})-
\log \det (\Delta+I_{N_2}+\rho_{22}H_{22}Q_{2p}H_{22}^{\dagger})- \log \det (I_{N_1}+\rho_{12}H_{12}Q_{2p}H_{12}^{\dagger}-\nonumber\\
&&\rho_{12}\rho_{22}H_{12}Q_{2p}H_{22}^{\dagger}
{(\Delta+I_{N_2}+\rho_{22}H_{22}Q_{2p}H_{22}^{\dagger})^{-1}}H_{22}Q_{2p}H_{12}^{\dagger})
+h(Y_2|X_{1c},X_{2c})+h(Y_2)\nonumber\\
&&+C_{12}-2{N_2}-(N_1+3N_2)\log (2\pi e)\nonumber\\
&\stackrel{(c)}{\ge}& h(Y_1,\hat {Y_2}|X_{1c})-
\log \det (\Delta+I_{N_2}+\rho_{22}H_{22}Q_{2p}H_{22}^{\dagger})- \log \det (I_{N_1}+\rho_{12}H_{12}Q_{2p}H_{12}^{\dagger})\nonumber\\
&&+h(Y_2|X_{1c},X_{2c})+h(Y_2)+C_{12}-2{N_2}-(N_1+3N_2)\log (2\pi e)\nonumber\\
&=& h(Y_1,\hat {Y_2}|X_{1c})-
\log \det (\Delta+I_{N_2}+\rho_{22}H_{22}Q_{2p}H_{22}^{\dagger})- \log \det (I_{N_1}+\rho_{12}H_{12}Q_{2p}H_{12}^{\dagger})\nonumber\\
&&+\log \det (I_{N_2}+\rho_{22}H_{22}Q_{2p}H_{22}^{\dagger}+\rho_{12}H_{12}Q_{1p}H_{12}^{\dagger})+h(Y_2)+C_{12}-2{N_2}-(N_1+2N_2)\log (2\pi e)\nonumber\\
&\stackrel{(d)}{=}& h(Y_1,\hat {Y_2}|X_{1c})-
\log \det 2(I_{N_2}+\rho_{22}H_{22}Q_{2p}H_{22}^{\dagger})- \log \det (I_{N_1}+\rho_{12}H_{12}Q_{2p}H_{12}^{\dagger})\nonumber\\
&&+\log \det (I_{N_2}+\rho_{22}H_{22}Q_{2p}H_{22}^{\dagger}+\rho_{12}H_{12}Q_{1p}H_{12}^{\dagger})
+h(Y_2)+C_{12}-2{N_2}-(N_1+2N_2)\log (2\pi e)\nonumber\\
&=& h(Y_1,\hat {Y_2}|X_{1c})-
\log \det (I_{N_2}+\rho_{22}H_{22}Q_{2p}H_{22}^{\dagger})- \log \det (I_{N_1}+\rho_{12}H_{12}Q_{2p}H_{12}^{\dagger})\nonumber\\
&&+\log \det (I_{N_2}+\rho_{22}H_{22}Q_{2p}H_{22}^{\dagger}+\rho_{12}H_{12}Q_{1p}H_{12}^{\dagger})
+h(Y_2)+C_{12}-3{N_2}-(N_1+2N_2)\log (2\pi e)\nonumber\\
&\ge& h(Y_1,\hat {Y_2}|X_{1c})- \log \det (I_{N_1}+\rho_{12}H_{12}Q_{2p}H_{12}^{\dagger})
+h(Y_2)+C_{12}-3{N_2}-(N_1+2N_2)\log (2\pi e)\nonumber\\
&\ge& h(Y_1,\hat {Y_2}|X_{1c})+h(Y_2)+C_{12}-N_1-3{N_2}-(N_1+2N_2)\log (2\pi e)\nonumber\\
&\stackrel{(e)}{\ge}& \log  \det \left(I_{N_1+N_2}+
\left[\begin{array}{c}
\sqrt{\rho_{11}}H_{11} \\
\sqrt{\rho_{12}}H_{12} \end{array}\right]
(I_{M_1}-H^{\dagger}_{12}{(I_{N_2}+{\rho}_{12}H_{12}H^{\dagger}_{12})^{-1}}H_{12})
[\sqrt{\rho_{11}}H_{11}^{\dagger}\ \sqrt{\rho_{12}}H_{12}^{\dagger}]\right.\nonumber\\
&&\left.+\left[\begin{array}{c}
\sqrt{\rho_{21}}H_{21} \\
\sqrt{\rho_{22}}H_{22} \end{array}\right]
[\sqrt{\rho_{21}}H_{21}^{\dagger}\ \sqrt{\rho_{22}}H_{22}^{\dagger}]\right)
+h(Y_2)+C_{12}-{N_1}-3{N_2}-N_2\log (2\pi e)\nonumber\\
&=& \log  \det \left(I_{N_1+N_2}+
\left[\begin{array}{c}
\sqrt{\rho_{11}}H_{11} \\
\sqrt{\rho_{12}}H_{12} \end{array}\right]
(I_{M_1}-H^{\dagger}_{12}{(I_{N_2}+{\rho}_{12}H_{12}H^{\dagger}_{12})^{-1}}H_{12})
[\sqrt{\rho_{11}}H_{11}^{\dagger}\ \sqrt{\rho_{12}}H_{12}^{\dagger}]\right.\nonumber\\
&&\left.+\left[\begin{array}{c}
\sqrt{\rho_{21}}H_{21} \\
\sqrt{\rho_{22}}H_{22} \end{array}\right]
[\sqrt{\rho_{21}}H_{21}^{\dagger}\ \sqrt{\rho_{22}}H_{22}^{\dagger}]\right)+\log  \det (I_{N_2}+{\rho}_{22}H_{22}H^{\dagger}_{22}+
{\rho}_{12}H_{12}H^{\dagger}_{12})+C_{12}\nonumber\\
&&-{N_1}-3{N_2},
\end{eqnarray}
where $(a)$ follows from \eqref{U_B}, $(b)$ is achieved similar to \eqref{piip}, and $(c)$ follows from the fact that $\log\det(.)$ is a monotonically increasing function on the cone of positive definite matrices, $(d)$ follows from \eqref{delt}, and $(e)$ is due to
\begin{eqnarray}
&&h(Y_1,\hat {Y_2}|X_{1c})\nonumber\\
&=& h(Y_1,\hat {Y_2},X_{1c})- h(X_{1c})\nonumber\\
&=& \log \det
\left[\begin{array}{ccc}
I_{N_1}+{\rho}_{11}H_{11}H^{\dagger}_{11}+{\rho}_{21}H_{21}H^{\dagger}_{21} & \sqrt{{\rho}_{12}{\rho}_{11}}H_{11}H^{\dagger}_{12}+
\sqrt{{\rho}_{22}{\rho}_{21}}H_{21}H^{\dagger}_{22} & \sqrt{{\rho}_{11}}H_{11}Q_{1c}\\
\sqrt{{\rho}_{12}{\rho}_{11}}H_{12}H^{\dagger}_{11}+
\sqrt{{\rho}_{22}{\rho}_{21}}H_{22}H^{\dagger}_{21} & \Delta+I_{N_2}+{\rho}_{22}H_{22}H^{\dagger}_{22}+
{\rho}_{12}H_{12}H^{\dagger}_{12} & \sqrt{{\rho}_{12}}H_{12}Q_{1c}\\
\sqrt{{\rho}_{11}}Q_{1c}H^{\dagger}_{11} &
\sqrt{{\rho}_{12}}Q_{1c}H^{\dagger}_{12} & Q_{1c}\end{array}
\right]\nonumber\\
&&- h(X_{1c})+(M_1+N_1+N_2)\log (2\pi e)\nonumber\\
&\stackrel{(f)}{=}& \log \det \left(
\left[\begin{array}{cc}
I_{N_1}+{\rho}_{11}H_{11}H^{\dagger}_{11}+{\rho}_{21}H_{21}H^{\dagger}_{21} & \sqrt{{\rho}_{12}{\rho}_{11}}H_{11}H^{\dagger}_{12}+\sqrt{{\rho}_{22}{\rho}_{21}}H_{21}H^{\dagger}_{22}\\
\sqrt{{\rho}_{12}{\rho}_{11}}H_{12}H^{\dagger}_{11}+\sqrt{{\rho}_{22}{\rho}_{21}}H_{22}H^{\dagger}_{21} & \Delta+I_{N_2}+{\rho}_{22}H_{22}H^{\dagger}_{22}+{\rho}_{12}H_{12}H^{\dagger}_{12}\end{array}
\right]-\right.\nonumber\\
&&\left.
\left[ \begin{array}{c}
H_{11} \\
H_{12} \end{array}
\right]
Q_{1c}({Q_{1c}}^{-1})Q_{1c}
[H^{\dagger}_{11} H^{\dagger}_{12}]
\right)+h(X_{1c})- h(X_{1c})+(N_1+N_2)\log (2\pi e)\nonumber\\
&=& \log \det \left(
\left[\begin{array}{cc}
I_{N_1}+{\rho}_{11}H_{11}H^{\dagger}_{11}+{\rho}_{21}H_{21}H^{\dagger}_{21} & \sqrt{{\rho}_{12}{\rho}_{11}}H_{11}H^{\dagger}_{12}+\sqrt{{\rho}_{22}{\rho}_{21}}H_{21}H^{\dagger}_{22}\\
\sqrt{{\rho}_{12}{\rho}_{11}}H_{12}H^{\dagger}_{11}+\sqrt{{\rho}_{22}{\rho}_{21}}H_{22}H^{\dagger}_{21} & \Delta+I_{N_2}+{\rho}_{22}H_{22}H^{\dagger}_{22}+{\rho}_{12}H_{12}H^{\dagger}_{12}\end{array}
\right]-\right.\nonumber\\
&&\left.
\left[ \begin{array}{c}
H_{11} \\
H_{12} \end{array}
\right]
(I_{M_1}-Q_{1p})
[H^{\dagger}_{11} H^{\dagger}_{12}]
\right)+(N_1+N_2)\log (2\pi e)\nonumber\\
&\stackrel{(g)}{\ge}& \log \det \left(
\left[\begin{array}{cc}
I_{N_1}+{\rho}_{11}H_{11}H^{\dagger}_{11}+{\rho}_{21}H_{21}H^{\dagger}_{21} & \sqrt{{\rho}_{12}{\rho}_{11}}H_{11}H^{\dagger}_{12}+\sqrt{{\rho}_{22}{\rho}_{21}}H_{21}H^{\dagger}_{22}\\
\sqrt{{\rho}_{12}{\rho}_{11}}H_{12}H^{\dagger}_{11}+\sqrt{{\rho}_{22}{\rho}_{21}}H_{22}H^{\dagger}_{21} & I_{N_2}+{\rho}_{22}H_{22}H^{\dagger}_{22}+{\rho}_{12}H_{12}H^{\dagger}_{12}\end{array}
\right]-\right.\nonumber\\
&&
\left.\left[ \begin{array}{c}
H_{11} \\
H_{12} \end{array}
\right]
(I_{M_1}-Q_{1p})
[H^{\dagger}_{11} H^{\dagger}_{12}]
\right)+(N_1+N_2)\log (2\pi e)\nonumber\\
&{=}& \log  \det \left(I_{N_1+N_2}+
\left[\begin{array}{c}
\sqrt{\rho_{11}}H_{11} \\
\sqrt{\rho_{12}}H_{12} \end{array}\right]
(I_{M_1}-H^{\dagger}_{12}{(I_{N_2}+{\rho}_{12}H_{12}H^{\dagger}_{12})^{-1}}H_{12})
[\sqrt{\rho_{11}}H_{11}^{\dagger}\ \sqrt{\rho_{12}}H_{12}^{\dagger}]\right.\nonumber\\
&&\left.+\left[\begin{array}{c}
\sqrt{\rho_{21}}H_{21} \\
\sqrt{\rho_{22}}H_{22} \end{array}\right]
[\sqrt{\rho_{21}}H_{21}^{\dagger}\ \sqrt{\rho_{22}}H_{22}^{\dagger}]\right)+(N_1+N_2)\log (2\pi e),
\end{eqnarray}
where $(f)$ is due to Lemma \ref{lem_block} and $(g)$ results from Lemma \ref{lem_block} and the fact that $\log\det(.)$ is a monotonically increasing function on the cone of positive definite matrices and also the fact that $\Delta$ is a positive definite matrix.

Thus, we see that this $R_1+2R_2$ bound is within $N_1+3N_2$ bits of the outer bound in \eqref{ro0eq10}.

\noindent {\bf \eqref{r12r2.4}:} For this bound in Lemma \ref{inner} we have
\begin{eqnarray}
&& I(X_1,X_{2c};Y_1,\hat {Y_2}|X_{1c})+
I(X_{2c};Y_1|X_1)+
I(X_{1c},X_2;Y_2|X_{2c})\nonumber\\
&&+I(X_2;Y_2|X_{1c},X_{2c})+C_{12}\nonumber\\
&=& h(Y_1,\hat {Y_2}|X_{1c})-h(Y_1,\hat {Y_2}|X_{1c},X_1,X_{2c})+
h(Y_1|X_1)-h(Y_1|X_1,X_{2c})+
h(Y_2|X_{2c})\nonumber\\
&&-h(Y_2|X_{2c},X_{1c},X_2)+h(Y_2|X_{1c},X_{2c})
-h(Y_2|X_{1c},X_{2c},X_2)+C_{12}\nonumber\\
&\stackrel{(a)}{\ge}& h(Y_1,\hat {Y_2}|X_{1c})-h(Y_1,\hat {Y_2}|X_{1c},X_1,X_{2c})+
h(Y_1|X_1)+h(Y_2|X_{2c})+h(Y_2|X_{1c},X_{2c})\nonumber\\
&&+C_{12}-{N_1}-2{N_2}-(N_1+2{N_2})\log (2\pi e)\nonumber\\
&\stackrel{(b)}{\ge}& h(Y_1,\hat {Y_2}|X_{1c})+h(Y_1|X_1)+h(Y_2|X_{2c})+C_{12}-2{N_1}-3{N_2}-2(N_1+N_2)\log (2\pi e)\nonumber\\
&{=}& h(Y_1,\hat {Y_2}|X_{1c})+
\log \det (I_{N_1}+{\rho }_{21}H_{21}H^{\dagger }_{21})+\log \det (I_{N_2}+{\rho }_{22}H_{22}Q_{2p}H^{\dagger }_{22}+{\rho }_{12}H_{12}H^{\dagger }_{12})
\nonumber\\
&&+C_{12}-2{N_1}-3{N_2}-(N_1+N_2)\log (2\pi e)\nonumber\\
&\stackrel{(c)}{\ge}& h(Y_1,\hat {Y_2}|X_{1c})+
\log  \det (I_{N_2}+{\rho}_{22}H_{22}H^{\dagger}_{22}+
{\rho}_{12}H_{12}H^{\dagger}_{12})
+C_{12}-2{N_1}-3{N_2}-(N_1+N_2)\log (2\pi e)\nonumber\\
&\stackrel{(d)}{\ge}& \log  \det \left(I_{N_1+N_2}+
\left[\begin{array}{c}
\sqrt{\rho_{11}}H_{11} \\
\sqrt{\rho_{12}}H_{12} \end{array}\right]
(I_{M_1}-H^{\dagger}_{12}{(I_{N_2}+{\rho}_{12}H_{12}H^{\dagger}_{12})^{-1}}H_{12})
[\sqrt{\rho_{11}}H_{11}^{\dagger}\ \sqrt{\rho_{12}}H_{12}^{\dagger}]\right.\nonumber\\
&&\left.+\left[\begin{array}{c}
\sqrt{\rho_{21}}H_{21} \\
\sqrt{\rho_{22}}H_{22} \end{array}\right]
[\sqrt{\rho_{21}}H_{21}^{\dagger}\ \sqrt{\rho_{22}}H_{22}^{\dagger}]\right)+\log  \det (I_{N_2}+{\rho}_{22}H_{22}H^{\dagger}_{22}+
{\rho}_{12}H_{12}H^{\dagger}_{12})+C_{12}\nonumber\\
&&-2{N_1}-3{N_2},
\end{eqnarray}
where $(a)$ follows from \eqref{U_B}, $(c)$ follows from \eqref{vvv}, and $(b)$ and $(d)$ can be seen similar to the proof of the last bound.

Thus, we see that this $R_1+2R_2$ bound is within $2N_1+3N_2$ bits of the outer bound in \eqref{ro0eq10}.

We define the region $\mathcal{R}^p$ including all the achievability bounds in \eqref{rr1x}-\eqref{r12r2.4} except for \eqref{rr1x} and \eqref{2r1r2.2}. Up to now, we have analyzed all the bounds of $\mathcal{R}^p$. We proved in $\mathcal{R}^p$ that:
\begin{eqnarray}
R_1 &\le& I_1-N_1-N_2,\nonumber\\
R_2 &\le& I_2-N_1-N_2,\nonumber\\
R_1+R_2 &\le& \min\{I_3,I_4,I_5,I_6\}-N_1-N_2-\max(N_1,N_2),\nonumber\\
2R_1+R_2 &\le& \min\{I_7,I_9\}-2N_1-2N_2,\nonumber\\
R_1+2R_2 &\le& \min\{I_8,I_{10}\}-2N_1-3N_2.
\end{eqnarray}
Thus, $\mathcal{R}^p$ contains the region which is within $N_1+N_2$ bits to the outer bound $\mathcal{R}_0$.

Now, add constraints \eqref{2r1r2.2} and \eqref{rr1x} to $\mathcal{R}^p$. \cite{Tse} proved that whenever \eqref{2r1r2.2} is active, at least one of the $R_1+R_2$ bounds is active, which can be extended to the MIMO case because Claim 5.6 in \cite{Tse} is true in general independent of the number of antennas.  We will now  present similar reasoning for bound \eqref{rr1x} to show that whenever bound \eqref{rr1x} is active, at least one of the $R_1+R_2$ bounds is active.

The value of $R_1+R_2$ at the intersection of \eqref{rr1x} and \eqref{r12r2.1} is greater than the average value of $R_1+R_2$ in \eqref{r1r2.1} and \eqref{r1r2.5}:
\begin{eqnarray}
&&{\rm RHS\ of\ }\eqref{rr1x}+{\rm RHS\ of\ }\eqref{r12r2.1}\nonumber\\
&=&I(X_1;Y_1|X_{2c})+I(X_1,X_{2c};Y_1|X_{1c})+I(X_{1c},X_{2};Y_2)
+I(X_2;Y_2|X_{1c},X_{2c})+C_{12}+{(C_{21}-\xi)}^{+}\nonumber\\
&\stackrel{(a)}{\ge}& I(X_{2c},X_1;Y_1)+I(X_2;Y_2|X_{1c},X_{2c})+{(C_{21}-\xi)}^{+}+
I(X_1;Y_1|X_{1c},X_{2c})+I(X_{1c},X_2;Y_2)+C_{12}\nonumber\\
&=&{\rm RHS\ of\ }\eqref{r1r2.1}+{\rm RHS\ of\ }\eqref{r1r2.5},
\end{eqnarray}
where $(a)$ follows from the fact that
\begin{eqnarray}\label{eeee}
&&I(X_1;Y_1|X_{2c})+I(X_1,X_{2c};Y_1|X_{1c})-I(X_{2c},X_1;Y_1)-I(X_1;Y_1|X_{1c},X_{2c})\nonumber\\
&=& h(Y_1|X_{2c})+h(Y_1|X_{1c})-h(Y_1)-h(Y_1|X_{1c},X_{2c})\nonumber\\
&\stackrel{(b)}{\ge}& 0,
\end{eqnarray}
where $(b)$ results from the following fact that if $A$, $B$, $C$ and $D$ are invertible positive semi-definite $M\times M$ matrices then
\begin{eqnarray}
\det(A+B).\det(A+C)&\ge&\det(A+B+C).\det(A),
\end{eqnarray}
because it is equivalent to
\begin{eqnarray}
\det(A+B).\det(A^{-1}).\det(A+C)&\ge&\det(A+B+C),
\end{eqnarray}
or
\begin{eqnarray}
\det(A+B+C+BA^{-1}C)&\ge&\det(A+B+C),
\end{eqnarray}
which is trivial.

It shows that when both the bounds \eqref{rr1x} and \eqref{r12r2.1} are active, at least one of the bounds \eqref{r1r2.1} or \eqref{r1r2.5} will be active also.

The value of $R_1+R_2$ at the intersection of \eqref{rr1x} and \eqref{r12r2.2} is greater than the average value of $R_1+R_2$ in \eqref{r1r2.1} and \eqref{r1r2.6}:
\begin{eqnarray}
&&{\rm RHS\ of\ }\eqref{rr1x}+{\rm RHS\ of\ }\eqref{r12r2.2}\nonumber\\
&=&I(X_1;Y_1|X_{2c})+I(X_1,X_{2c};Y_1|X_{1c})+I(X_{2c};Y_1|X_{1})+I(X_{1c},X_{2};Y_2|X_{2c})
+I(X_2;Y_2|X_{1c},X_{2c})\nonumber\\
&&+C_{12}+{(C_{21}-\xi)}^{+}\nonumber\\
&\stackrel{(a)}{\ge}& I(X_{2c},X_1;Y_1)+I(X_2;Y_2|X_{1c},X_{2c})+{(C_{21}-\xi)}^{+}+
I(X_1;Y_1|X_{1c},X_{2c})+I(X_{2c};Y_1|X_1)+\nonumber\\
&&I(X_{1c},X_2;Y_2|X_{2c})+C_{12}\nonumber\\
&=&{\rm RHS\ of\ }\eqref{r1r2.1}+{\rm RHS\ of\ }\eqref{r1r2.6},
\end{eqnarray}
where $(a)$ follows from \eqref{eeee}.

It shows that when both the bounds \eqref{rr1x} and \eqref{r12r2.2} are active, at least one of the bounds \eqref{r1r2.1} or \eqref{r1r2.6} will be active also.

The value of $R_1+R_2$ at the intersection of \eqref{rr1x} and \eqref{r12r2.3} is greater than the average value of $R_1+R_2$ in \eqref{r1r2.2} and \eqref{r1r2.5}:
\begin{eqnarray}
&&{\rm RHS\ of\ }\eqref{rr1x}+{\rm RHS\ of\ }\eqref{r12r2.3}\nonumber\\
&=&I(X_1;Y_1|X_{2c})+I(X_1,X_{2c};Y_1,\hat {Y_2}|X_{1c})+I(X_{1c},X_2;Y_2)
+I(X_2;Y_2|X_{1c},X_{2c})+C_{12}\nonumber\\
&\stackrel{(a)}{\ge}& I(X_1,X_{2c};Y_1,\hat {Y_2}|X_{1c})+I(X_2;Y_2|X_{1c},X_{2c})+
I(X_1;Y_1|X_{1c},X_{2c})+I(X_{1c},X_2;Y_2)+C_{12}\nonumber\\
&=&{\rm RHS\ of\ }\eqref{r1r2.2}+{\rm RHS\ of\ }\eqref{r1r2.5},
\end{eqnarray}
where $(a)$ follows from the fact that
\begin{eqnarray}\label{eee}
&&I(X_1;Y_1|X_{2c})-I(X_1;Y_1|X_{1c},X_{2c})\nonumber\\
&=& I(Y_1|X_{2c})-I(Y_1|X_{1c},X_{2c})\nonumber\\
&=& \log \det(I_{N_1}+{\rho }_{11}H_{11}H^{\dagger }_{11}+{\rho }_{21}H_{21}Q_{2p}H^{\dagger }_{21})-\log \det(I_{N_1}+{\rho }_{11}H_{11}Q_{1p}H^{\dagger }_{11}+{\rho }_{21}H_{21}Q_{2p}H^{\dagger }_{21})\nonumber\\
&{\ge}& 0.
\end{eqnarray}

It shows that when both the bounds \eqref{rr1x} and \eqref{r12r2.3} are active, at least one of the bounds \eqref{r1r2.2} or \eqref{r1r2.5} will be active also.

The value of $R_1+R_2$ at the intersection of \eqref{rr1x} and \eqref{r12r2.4} is greater than the average value of $R_1+R_2$ in \eqref{r1r2.2} and \eqref{r1r2.6}:
\begin{eqnarray}
&&{\rm RHS\ of\ }\eqref{rr1x}+{\rm RHS\ of\ }\eqref{r12r2.4}\nonumber\\
&=&I(X_1;Y_1|X_{2c})+I(X_1,X_{2c};Y_1,\hat {Y_2}|X_{1c})+
I(X_{2c};Y_1|X_1)+
I(X_{1c},X_2;Y_2|X_{2c})\nonumber\\
&&+I(X_2;Y_2|X_{1c},X_{2c})+C_{12}\nonumber\\
&\stackrel{(a)}{\ge}& I(X_1,X_{2c};Y_1,\hat {Y_2}|X_{1c})+I(X_2;Y_2|X_{1c},X_{2c})+
I(X_1;Y_1|X_{1c},X_{2c})+I(X_{2c};Y_1|X_1)+\nonumber\\
&&I(X_{1c},X_2;Y_2|X_{2c})+C_{12}\nonumber\\
&=&{\rm RHS\ of\ }\eqref{r1r2.2}+{\rm RHS\ of\ }\eqref{r1r2.6},
\end{eqnarray}
where $(a)$ follows from \eqref{eee}.

It shows that when both the bounds \eqref{rr1x} and \eqref{r12r2.4} are active, at least one of the bounds \eqref{r1r2.2} or \eqref{r1r2.6} will be active also.

So, when \eqref{rr1x} is active, we can see that at least one of the $R_1+R_2$ bounds in \eqref{r1r2.1}-\eqref{r1r2.6} is active in $\mathcal{R}_{2\rightarrow 1\rightarrow 2}$. Hence, with a strategy similar to the one in Claim 5.6 of \cite{Tse} for \eqref{2r1r2.2} we can see that the bound \eqref{rr1x} does not show up in $conv\{\mathcal{R}_{2 \rightarrow 1 \rightarrow 2} \cup \mathcal{R}_{1 \rightarrow 2 \rightarrow 1}\}$.

Therefore, the $R_1$ bound \eqref{rr1x} and the $2R_1+R_2$ bound \eqref{2r1r2.2} do not show up in $\mathcal{R}=conv\{\mathcal{R}_{2 \rightarrow 1 \rightarrow 2} \cup \mathcal{R}_{1 \rightarrow 2 \rightarrow 1}\}$ and $\mathcal{R}$ is within $N_1+N_2$ bits per user to the outer bounds in Theorem 1.

\section{Proof of Theorem \ref{thm_dof}}
\label{Appendix1}
In this section, we will find the limit of $\mathcal{R}_o/\log \mathsf{SNR}$ as $\mathsf{SNR}\to \infty$ to get the result stated in Theorem \ref{thm_dof} when ${{C }_{ij}}\sim{\mathsf{SNR}}^{\beta_{ij}}$ and ${{\rho}_{ij}}\sim{\mathsf{SNR}}$ where ${\beta_{12}},{\beta_{21}}\in {\mathbb R}^+$.

This follows from Theorem \ref{outer_inner_capacity_reciprocal} since the capacity region is inner and outer- bounded by  $\mathcal{R}_o$ with constant gaps which would vanish for the DoF.
Before going over each of the above terms and finding their high SNR limit, we first give some lemmas that will be used in the proof.

\begin{lemma}[\cite{Varanasi}]
Let $H_{1}\in \mathbb{C}^{N\times M_1}$, $H_{2}\in \mathbb{C}^{N\times M_2}$,..., and $H_{k}\in \mathbb{C}^{N\times M_k}$ be $k$ full rank and independent channel matrices. Then, the following holds
\begin{eqnarray}
&&{\log  {\det  (I_{N}+{\rho }H_{1}H^{\dagger }_{1}+{\rho }H_{2}H^{\dagger }_{2}+...+{\rho }H_{k}H^{\dagger }_{k})\ }\ } \nonumber\\
&=&{\log  {\det  (I_{N}+{\rho }[H_{1}\ ...\ H_{k}][H_{1}\ ...\ H_{k}]^{\dagger })\ }\ } \nonumber\\
&=& \min\{N,M_1+M_2+...+M_k\}{\log \rho }+o({\log \rho }).
\end{eqnarray}\label{dgg}
\end{lemma}

\begin{lemma}[\cite{J1}]
\label{dlemma-dof}
Let $H_{ii}\in \mathbb{C}^{N_i\times M_i}$ and $H_{ij}\in \mathbb{C}^{N_i\times M_j}$ be two channel matrices with each entry independently chosen from $\mathsf{CN}(0,1)$. Then, the following holds with probability $1$ (over the randomness of channel matrices).
\begin{eqnarray}
&&\log \det (I_{N_i}+{\rho }H_{ii}H^{\dagger }_{ii}-{\rho }H_{ii}H^{\dagger }_{ij}{(I_{N_j}+{\rho }H_{ij}H^{\dagger }_{ij})}^{-1}
{\rho }H_{ij}H^{\dagger }_{ii})\nonumber\\
&&=\min\{N_i,(M_i-N_j)^+\}{\log \rho }+o({\log \rho }).
\end{eqnarray}
\end{lemma}

\begin{lemma}
\label{dlemma-dof2}
Let $H_{ii}\in \mathbb{C}^{N_i\times M_i}$ and $H_{ij}\in \mathbb{C}^{N_j\times M_i}$ be two channel matrices with each entry independently chosen from $\mathsf{CN}(0,1)$. Then, the following holds with probability $1$ (over the randomness of channel matrices).
\begin{eqnarray}
&&\log \det (I_{N_j}+{\rho }H_{ij}H^{\dagger }_{ij}-{\rho }H_{ij}H^{\dagger }_{ii}{(I_{N_i}+{\rho }H_{ii}H^{\dagger }_{ii})}^{-1}
{\rho }H_{ii}H^{\dagger }_{ij})\nonumber\\
&&=\min\{N_j,(M_i-N_i)^+\}{\log \rho }+o({\log \rho }).
\end{eqnarray}
\end{lemma}
\begin{proof}
The proof is similar to that of Lemma \ref{dlemma-dof}.
\end{proof}

Now we find the high SNR limits of the bounds in \eqref{ro0eq1}-\eqref{ro0eq10} leading to Theorem 2.

\noindent {\bf \eqref{ro0eq1}$\rightarrow$\eqref{dofeq1}:} Consider bound \eqref{ro0eq1} in $\mathcal{R}_o$, we have
\begin{eqnarray}
&&{\log  {\det  \left(I_{N_1}+{\rho}_{11}H_{11}H^{\dagger}_{11}\right)}}+\min \{\log \det \left(I_{N_2}+{\rho }_{12}H_{12}H^{\dagger }_{12}-{\rho }_{12}{\rho }_{11}H_{12}H^{\dagger }_{11}\right. \nonumber\\
&&\left.{\left(I_{N_1}+{\rho}_{11}H_{11}H^{\dagger}_{11}\right)^{-1}}H_{11}H^{\dagger }_{12}\right),C_{21} \}\nonumber\\
&=&{\log  {\det  (I_{N_1}+{\rho}H_{11}H^{\dagger}_{11})}}+\min \{\log \det (I_{N_2}+{\rho }^{\alpha}H_{12}H^{\dagger }_{12}-{\rho^2}H_{12}H^{\dagger }_{11} \nonumber\\
&&{(I_{N_1}+{\rho}H_{11}H^{\dagger}_{11})^{-1}}H_{11}H^{\dagger }_{12}),C_{21} \}\nonumber\\
&\stackrel{(a)}{=}&({\min\{M_1,N_1\}}+\min\{
\min\{N_2,(M_1-N_1)^+\},\beta_{21}\}
){\log \mathsf{SNR}\ }+o({\log  \mathsf{SNR}\ }),
\end{eqnarray}
where $(a)$ follows from Lemma \ref{dgg} and Lemma \ref{dlemma-dof2}. Now, dividing both sides by $\log \mathsf{SNR}$, we obtain \eqref{dofeq1}.

\noindent {\bf \eqref{ro0eq2}$\rightarrow$\eqref{dofeq2}:} This is obtained similarly to the last bound by exchanging $1$ and $2$ in the indices.

\noindent {\bf \eqref{ro0eq3}$\rightarrow$\eqref{dofeq3}:} Consider bound \eqref{ro0eq3} in $\mathcal{R}_o$, we have
\begin{eqnarray}
&&\log  \det \left(I_{N_1}+{\rho}_{11}H_{11}H^{\dagger}_{11}+
{\rho}_{21}H_{21}H^{\dagger}_{21}-
{\rho}_{11}{\rho}_{12}H_{11}H^{\dagger}_{12}
{\left(I_{N_2}+{\rho}_{12}H_{12}H^{\dagger}_{12}\right)^{-1}}H_{12}H^{\dagger}_{11}\right)+\nonumber\\
&&\log  \det \left(I_{N_2}+{\rho}_{22}H_{22}H^{\dagger}_{22}+
{\rho}_{12}H_{12}H^{\dagger}_{12}-
{\rho}_{22}{\rho}_{21}H_{22}H^{\dagger}_{21}
{\left(I_{N_1}+{\rho}_{21}H_{21}H^{\dagger}_{21}\right)^{-1}}H_{21}H^{\dagger}_{22}\right)+\nonumber\\
&&C_{12}+C_{21}\nonumber\\
&=&\log  \det (I_{N_1}+{\rho}H_{11}H^{\dagger}_{11}+
{\rho}H_{21}H^{\dagger}_{21}-
{\rho}{\rho}H_{11}H^{\dagger}_{12}
{(I_{N_2}+{\rho}H_{12}H^{\dagger}_{12})^{-1}}H_{12}H^{\dagger}_{11})+\nonumber\\
&&\log  \det (I_{N_2}+{\rho}H_{22}H^{\dagger}_{22}+
{\rho}H_{12}H^{\dagger}_{12}-
{\rho}{\rho}H_{22}H^{\dagger}_{21}
{(I_{N_1}+{\rho}H_{21}H^{\dagger}_{21})^{-1}}H_{21}H^{\dagger}_{22})+\nonumber\\
&&C_{12}+C_{21}\nonumber\\
&\stackrel{(a)}{=}&(
\min\{N_1,(M_1-N_2)^++M_2\}
+\min\{N_2,(M_2-N_1)^++M_1\}
+\nonumber\\
&&\beta_{12}
+\beta_{21}
){\log \mathsf{SNR}\ }+o({\log  \mathsf{SNR}\ })),
\end{eqnarray}
where $(a)$ follows from Lemma \ref{dgg} and Lemma \ref{dlemma-dof}. Now, dividing both sides by $\log \mathsf{SNR}$, we obtain \eqref{dofeq3}.

\noindent {\bf \eqref{ro0eq4}$\rightarrow$\eqref{dofeq4}:} Consider bound \eqref{ro0eq4} in $\mathcal{R}_o$, we have
\begin{eqnarray}
&&\log  \det \left(I_{N_1}+{\rho}_{11}H_{11}H^{\dagger}_{11}-
{\rho}_{11}{\rho}_{12}H_{11}H^{\dagger}_{12}
{\left(I_{N_2}+{\rho}_{12}H_{12}H^{\dagger}_{12}\right)^{-1}}H_{12}H^{\dagger}_{11}\right)+\nonumber\\
&&\log  \det \left(I_{N_2}+{\rho}_{22}H_{22}H^{\dagger}_{22}+
{\rho}_{12}H_{12}H^{\dagger}_{12}\right)+C_{12}\nonumber\\
&=&\log  \det (I_{N_1}+{\rho}H_{11}H^{\dagger}_{11}-
{\rho}{\rho}H_{11}H^{\dagger}_{12}
{(I_{N_2}+{\rho}^{\alpha}H_{12}H^{\dagger}_{12})^{-1}}H_{12}H^{\dagger}_{11})+\nonumber\\
&&\log  \det (I_{N_2}+{\rho}H_{22}H^{\dagger}_{22}+
{\rho}H_{12}H^{\dagger}_{12})+C_{12}\nonumber\\
&\stackrel{(a)}{=}&(\min\{N_1,(M_1-N_2)^+\}+
\min\{N_2,M_1+M_2\}
+\beta_{12}
){\log \mathsf{SNR}\ }+o({\log  \mathsf{SNR}\ }),
\end{eqnarray}
where $(a)$ follows from Lemma \ref{dgg} and Lemma \ref{dlemma-dof2}. Now, dividing both sides by $\log \mathsf{SNR}$, we obtain \eqref{dofeq4}.

\noindent {\bf \eqref{ro0eq5}$\rightarrow$\eqref{dofeq5}:} This is obtained similarly to the previous bound by exchanging $1$ and $2$ in the indices.

\noindent {\bf \eqref{ro0eq6}$\rightarrow$\eqref{dofeqx}:} Consider bound \eqref{ro0eq6} in $\mathcal{R}_o$, using Lemma \ref{dgg} we have
\begin{eqnarray}
&& \log  \det \left(I_{N_1+N_2}+
\left[\begin{array}{c}
\sqrt{\rho_{11}}H_{11} \\
\sqrt{\rho_{12}}H_{12} \end{array}\right]
[\sqrt{\rho_{11}}H_{11}^{\dagger}\ \sqrt{\rho_{12}}H_{12}^{\dagger}]+
\left[\begin{array}{c}
\sqrt{\rho_{21}}H_{21} \\
\sqrt{\rho_{22}}H_{22} \end{array}\right]
[\sqrt{\rho_{21}}H_{21}^{\dagger}\ \sqrt{\rho_{22}}H_{22}^{\dagger}]\right)\nonumber\\
&=&\log  \det \left(I_{N_1+N_2}+
\rho\left[\begin{array}{c}
H_{11} \\
H_{12} \end{array}\right]
[H_{11}^{\dagger}\ H_{12}^{\dagger}]+
\rho\left[\begin{array}{c}
H_{21} \\
H_{22} \end{array}\right]
[H_{21}^{\dagger}\ H_{22}^{\dagger}]\right)\nonumber\\
&=&\min\{N_1+N_2,M_1+M_2\}{\log \mathsf{SNR}\ }+o({\log \mathsf{SNR}\ }).
\end{eqnarray}

\noindent {\bf \eqref{ro0eq7}$\rightarrow$\eqref{dofeq6}:} Consider bound bound \eqref{ro0eq7} in $\mathcal{R}_o$, we have
\begin{eqnarray}
&&\log  \det \left(I_{N_1}+{\rho}_{11}H_{11}H^{\dagger}_{11}-
{\rho}_{11}{\rho}_{12}H_{11}H^{\dagger}_{12}
{\left(I_{N_2}+{\rho}_{12}H_{12}H^{\dagger}_{12}\right)^{-1}}H_{12}H^{\dagger}_{11}\right)+\nonumber\\
&&\log  \det \left(I_{N_2}+{\rho}_{22}H_{22}H^{\dagger}_{22}+
{\rho}_{12}H_{12}H^{\dagger}_{12}-
{\rho}_{22}{\rho}_{21}H_{22}H^{\dagger}_{21}
{\left(I_{N_1}+{\rho}_{21}H_{21}H^{\dagger}_{21}\right)^{-1}}H_{21}H^{\dagger}_{22}\right)+\nonumber\\
&&\log  \det \left(I_{N_1}+{\rho}_{11}H_{11}H^{\dagger}_{11}+
{\rho}_{21}H_{21}H^{\dagger}_{21}\right)+C_{12}+C_{21}\nonumber\\
&=&\log \det  \log  \det (I_{N_1}+{\rho}H_{11}H^{\dagger}_{11}-
{\rho}H_{11}H^{\dagger}_{12}
{(I_{N_2}+{\rho}H_{12}H^{\dagger}_{12})^{-1}}{\rho}H_{12}H^{\dagger}_{11})+\nonumber\\
&&\log  \det (I_{N_2}+{\rho}H_{22}H^{\dagger}_{22}+
{\rho}H_{12}H^{\dagger}_{12}-
{\rho}H_{22}H^{\dagger}_{21}
{(I_{N_1}+{\rho}H_{21}H^{\dagger}_{21})^{-1}}{\rho}H_{21}H^{\dagger}_{22})+\nonumber\\
&&\log  \det (I_{N_1}+{\rho}H_{11}H^{\dagger}_{11}+
{\rho}H_{21}H^{\dagger}_{21})+\beta_{12}+\beta_{21} \nonumber\\
&\stackrel{(a)}{=}& \min\{N_2,(M_2-N_1)^++M_1\}+\min\{N_1,(M_1-N_2)^+\}+\nonumber\\
&&\min\{N_1,M_1+M_2\}+\beta_{12}+\beta_{21},
\end{eqnarray}
where $(a)$ is obtained from Lemma \ref{dgg} and Lemma \ref{dlemma-dof}. Now, dividing both sides by $\log \mathsf{SNR}$, we obtain \eqref{dofeq6}.

\noindent {\bf \eqref{ro0eq8}$\rightarrow$\eqref{dofeq7}:} This is obtained similarly to the previous bound by exchanging $1$ and $2$ in the indices.

\noindent {\bf \eqref{ro0eq9}$\rightarrow$\eqref{dofeqy}:} Consider bound \eqref{ro0eq9} in $\mathcal{R}_o$, we have

\begin{eqnarray}
&&\log  \det \left(I_{N_1+N_2}+
\left[\begin{array}{c}
\sqrt{\rho_{22}}H_{22} \\
\sqrt{\rho_{21}}H_{21} \end{array}\right]
(I_{M_2}-{\rho_{12}}H^{\dagger}_{21}{(I_{N_1}+{\rho}_{21}H_{21}H^{\dagger}_{21})^{-1}}H_{21})
[\sqrt{\rho_{22}}H_{22}^{\dagger}\ \sqrt{\rho_{21}}H_{21}^{\dagger}]\right.\nonumber\\
&&\left.+\left[\begin{array}{c}
\sqrt{\rho_{12}}H_{12} \\
\sqrt{\rho_{11}}H_{11} \end{array}\right]
[\sqrt{\rho_{12}}H_{12}^{\dagger}\ \sqrt{\rho_{11}}H_{11}^{\dagger}]\right)+\log  \det \left(I_{N_1}+{\rho}_{11}H_{11}H^{\dagger}_{11}+
{\rho}_{21}H_{21}H^{\dagger}_{12}\right)+C_{21}\nonumber\\
&=&\log  \det \left(I_{N_1+N_2}+
\rho\left[\begin{array}{c}
H_{22} \\
H_{21} \end{array}\right]
(I_{M_2}-{\rho}H^{\dagger}_{21}{(I_{N_1}+{\rho}H_{21}H^{\dagger}_{21})^{-1}}H_{21})
[H_{22}^{\dagger}\ H_{21}^{\dagger}]\right.\nonumber\\
&&\left.+\rho\left[\begin{array}{c}
H_{12} \\
H_{11} \end{array}\right]
[H_{12}^{\dagger}\ H_{11}^{\dagger}]\right)+\log  \det \left(I_{N_1}+{\rho}H_{11}H^{\dagger}_{11}+
{\rho}H_{21}H^{\dagger}_{12}\right)+C_{21}\nonumber\\
&=&(\min\{N_1+N_2,M_1\}+\min\{N_1,M_1+M_2\}+
\beta_{21}){\log \mathsf{SNR}\ }\nonumber\\
&&+o({\log \mathsf{SNR}\ }).
\end{eqnarray}

\noindent {\bf \eqref{ro0eq10}$\rightarrow$\eqref{dofeqz}:} This is obtained similarly to the previous bound by exchanging $1$ and $2$ in the indices.

Combining the above results we obtain Theorem \ref{thm_dof} results.

\section{Proof of Theorem \ref{thm_gdof}}
\label{Appendix3}

In this section, we will find the limit of $\mathcal{R}_o/\log \mathsf{SNR}$ as $\mathsf{SNR}\to \infty$ to get the result stated in Theorem \ref{thm_gdof} when ${{C }_{ij}}\sim{\mathsf{SNR}}^{\beta_{ij}}$, ${{\rho}_{ij}}\sim{\mathsf{SNR}}$ for $i=j$ and ${{\rho}_{ij}}\sim{\mathsf{SNR}^\alpha}$ for $i\neq j$ where ${\beta_{12}},{\beta_{21}}\in {\mathbb R}^+$.

This follows from Theorem \ref{outer_inner_capacity_reciprocal} since the capacity region is inner and outer- bounded by  $\mathcal{R}_o$ with constant gaps which would vanish for the DoF.
Before going over each of the above terms and finding their high SNR limit. We first give some lemmas that will be used for the proof.

\begin{lemma}[\cite{Varanasi}]
Let $H_{1}\in \mathbb{C}^{M\times M}$, $H_{2}\in \mathbb{C}^{M\times M}$, ..., and $H_{k}\in \mathbb{C}^{M\times M}$ be k full rank channel matrices. Then, the following holds
\begin{eqnarray}
&&{\log  {\det  (I_{M}+{\rho }^{\alpha_{1}}H_{1}H^{\dagger }_{1}+{\rho }^{\alpha_{2}}H_{2}H^{\dagger }_{2}+...+{\rho }^{\alpha_{k}}H_{k}H^{\dagger }_{k})\ }\ } \nonumber\\
&=&\max\{\alpha_1,\alpha_2,...,\alpha_k\}M
{\log \rho }+o({\log \rho }).
\end{eqnarray}\label{gg}
\end{lemma}

\begin{lemma}[\cite{J1}]
\label{lemma-dof}
Let $H_{ii}\in \mathbb{C}^{M\times M}$ and $H_{ij}\in \mathbb{C}^{M\times M}$ be two channel matrices with each entry independently chosen from $\mathsf{CN}(0,1)$. Then, the following holds with probability $1$ (over the randomness of channel matrices).
\begin{eqnarray}
&&\log \det (I_M+{\rho }H_{ii}H^{\dagger }_{ii}-\sqrt{{\rho }{\rho }^{\alpha}}H_{ii}H^{\dagger }_{ij}{(I_M+{\rho }^{\alpha}H_{ij}H^{\dagger }_{ij})}^{-1}
\sqrt{{\rho }{\rho }^{\alpha}}H_{ij}H^{\dagger }_{ii})\nonumber\\
&&=(1-{\alpha })^+M{\log \rho }+o({\log \rho }).
\end{eqnarray}
\end{lemma}

\begin{lemma}
\label{lemma-dof2}
Let $H_{ii}\in \mathbb{C}^{M\times M}$ and $H_{ij}\in \mathbb{C}^{M\times M}$ be two channel matrices with each entry independently chosen from $\mathsf{CN}(0,1)$. Then, the following holds with probability $1$ (over the randomness of channel matrices).
\begin{eqnarray}
&&\log \det (I_{N_j}+{\rho }^{\alpha}H_{ij}H^{\dagger }_{ij}-\sqrt{{\rho }{\rho }^{\alpha}}H_{ij}H^{\dagger }_{ii}{(I_{N_i}+{\rho }H_{ii}H^{\dagger }_{ii})}^{-1}
\sqrt{{\rho }{\rho }^{\alpha}}H_{ii}H^{\dagger }_{ij})\nonumber\\
&&=({\alpha }-1)^+M{\log \rho }+o({\log \rho }).
\end{eqnarray}
\end{lemma}
\begin{proof}
The proof is similar to that of given in \cite{J1}.
\end{proof}

\begin{lemma} \label{sym-lemma}
Let $H\in \mathbb{C}^{M\times M}$ be a full rank channel matrix. Then, the following holds
\begin{eqnarray}
&& I_{M}-{\rho}H^{\dagger}{(I_{M}+
{\rho}HH^{\dagger})^{-1}}H\nonumber\\
&=& (I_{M}+{\rho}H^{\dagger }H)^{-1}
\end{eqnarray}
\end{lemma}
\begin{proof}
Let $B \triangleq I_{M}+{\rho}H^{\dagger }H$. Thus,
\begin{eqnarray}
I_{M}-{\rho}H^{\dagger}{(B^{\dagger})^{-1}}H=(B)^{-1}
\end{eqnarray}
Since $B$ is invertible, it is enough to show that
\begin{eqnarray}
B-{\rho}H^{\dagger}{(B^{\dagger})^{-1}}HB=I_M,
\end{eqnarray}
which is equivalent to showing
\begin{eqnarray}
{\rho}H^{\dagger}{(B^{\dagger})^{-1}}HB={\rho}H^{\dagger }H
\end{eqnarray}
So it is enough to prove ${(B^{\dagger})^{-1}}HB=H $. Or,
$HB= B^{\dagger}H$, which holds since $B = I_{M}+{\rho}H^{\dagger }H$.
\end{proof}

Now we find the high SNR limits of the bounds in \eqref{ro0eq1}-\eqref{ro0eq10} leading to Theorem 3.

\noindent {\bf \eqref{ro0eq1}$\rightarrow$\eqref{gdofeq1}:} Consider bound \eqref{ro0eq1} in $\mathcal{R}_o$, we have
\begin{eqnarray}
&&{\log  {\det  (I_{M}+{\rho}H_{11}H^{\dagger}_{11})}}+\min \{\log \det (I_{M}+{\rho }^{\alpha}H_{12}H^{\dagger }_{12}-{\rho }^{\alpha+1}H_{12}H^{\dagger }_{11} \nonumber\\
&&{(I_{M}+{\rho}H_{11}H^{\dagger}_{11})^{-1}}H_{11}H^{\dagger }_{12}),C_{21} \}\nonumber\\
&\stackrel{(a)}{=}&(M+
\min\{({\alpha}-1)^+M,\beta\}
){\log \mathsf{SNR}\ }+o({\log  \mathsf{SNR}\ })),
\end{eqnarray}
where $(a)$ follows from Lemma \ref{gg} and Lemma \ref{lemma-dof2}. Now, dividing both sides by $\log \mathsf{SNR}$, we obtain \eqref{gdofeq1}.

\noindent {\bf \eqref{ro0eq2}$\rightarrow$\eqref{gdofeq2}:} This is obtained similarly to the last bound by exchanging $1$ and $2$ in the indices.

\noindent {\bf \eqref{ro0eq3}$\rightarrow$\eqref{gdofeq3}:} Consider bound \eqref{ro0eq3} in $\mathcal{R}_o$, we have

\begin{eqnarray}
&&\log  \det (I_{M}+{\rho}H_{11}H^{\dagger}_{11}+
{\rho}^{\alpha}H_{21}H^{\dagger}_{21}-
{\rho}^{\alpha+1}H_{11}H^{\dagger}_{12}
{(I_{N_2}+{\rho}^{\alpha}H_{12}H^{\dagger}_{12})^{-1}}H_{12}H^{\dagger}_{11})+\nonumber\\
&&\log  \det (I_{M}+{\rho}H_{22}H^{\dagger}_{22}+
{\rho}^{\alpha}H_{12}H^{\dagger}_{12}-
{\rho}^{\alpha+1}H_{22}H^{\dagger}_{21}
{(I_{N_1}+{\rho}^{\alpha}H_{21}H^{\dagger}_{21})^{-1}}H_{21}H^{\dagger}_{22})+\nonumber\\
&&C_{12}+C_{21}\nonumber\\
&\stackrel{(a)}{=}&(
2M\max\{({1}-{\alpha})^+,{\alpha}\}+2\beta
){\log \mathsf{SNR}\ }+o({\log  \mathsf{SNR}\ })),
\end{eqnarray}
where $(a)$ follows from Lemma \ref{gg} and Lemma \ref{lemma-dof}. Now, dividing both sides by $\log \mathsf{SNR}$, we obtain \eqref{gdofeq3}.

\noindent {\bf \eqref{ro0eq4}$\rightarrow$\eqref{gdofeq4}:} Consider bound \eqref{ro0eq4} in $\mathcal{R}_o$, we have
\begin{eqnarray}
&&\log  \det (I_{M}+{\rho}H_{11}H^{\dagger}_{11}-
{\rho}^{\alpha+1}H_{11}H^{\dagger}_{12}
{(I_{M}+{\rho}^{\alpha}H_{12}H^{\dagger}_{12})^{-1}}H_{12}H^{\dagger}_{11})+\nonumber\\
&&\log  \det (I_{M}+{\rho}H_{22}H^{\dagger}_{22}+
{\rho}^{\alpha}H_{12}H^{\dagger}_{12})+C_{12}\nonumber\\
&\stackrel{(a)}{=}&(
({1}-{\alpha})^+M + M\max\{{1},{\alpha}\}+\beta
){\log \mathsf{SNR}\ }+o({\log  \mathsf{SNR}\ })),
\end{eqnarray}
where $(a)$ follows from Lemma \ref{gg} and Lemma \ref{lemma-dof2}. Now, dividing both sides by $\log \mathsf{SNR}$, we obtain \eqref{gdofeq4}.

\noindent {\bf \eqref{ro0eq5}$\rightarrow$\eqref{gdofeq4}:} This is obtained similarly to the last bound by exchanging $1$ and $2$ in the indices which gives the same bound as the last one.

\noindent {\bf \eqref{ro0eq6}$\rightarrow$\eqref{gdofeq6}:} Consider bound \eqref{ro0eq6} in $\mathcal{R}_o$, we have
\begin{eqnarray}
&&\log  \det \left(I_{2M}+
\left[\begin{array}{c}
\sqrt{\rho}H_{11} \\
\sqrt{\rho^{\alpha}}H_{12} \end{array}\right]
\left[\sqrt{\rho}H_{11}^{\dagger}\ \sqrt{\rho^{\alpha}}H_{12}^{\dagger}\right]+
\left[\begin{array}{c}
\sqrt{\rho^{\alpha}}H_{21} \\
\sqrt{\rho}H_{22} \end{array}\right]
\left[\sqrt{\rho^{\alpha}}H_{21}^{\dagger}\ \sqrt{\rho}H_{22}^{\dagger}\right]\right)\nonumber\\
&=& \log \det \left({\left[ \begin{array}{cc}
I_{M}+{\rho }H_{11}H^{\dagger }_{11}+{\rho }^{\alpha}H_{21}H^{\dagger }_{21} & {{\rho }^{\frac{\alpha+1}{2}}}H_{11}H^{\dagger}_{12}+{{\rho }^{\frac{\alpha+1}{2}}}H_{21}H^{\dagger}_{22} \\
{{\rho }^{\frac{\alpha+1}{2}}}H_{12}H^{\dagger}_{11}+{{\rho }^{\frac{\alpha+1}{2}}}H_{22}H^{\dagger}_{21} & I_{M}+{\rho }H_{22}H^{\dagger }_{22}+{\rho }^{\alpha}H_{12}H^{\dagger }_{12} \end{array}
\right]}\right)\nonumber\\
&\stackrel{(a)}{=}&
\log \det  (I_{M}+{\rho }H_{11}H^{\dagger }_{11}+{\rho }^{\alpha}H_{21}H^{\dagger }_{21})+
\log \det (I_{M}+{\rho }H_{22}H^{\dagger }_{22}+{\rho }^{\alpha}H_{12}H^{\dagger }_{12}-\nonumber\\
&&({{\rho }^{\frac{\alpha+1}{2}}}H_{12}H^{\dagger}_{11}+{{\rho }^{\frac{\alpha+1}{2}}}H_{22}H^{\dagger}_{21}){(I_{M}+{\rho }H_{11}H^{\dagger }_{11}+{\rho }^{\alpha}H_{21}H^{\dagger }_{21})^{-1}}\nonumber\\
&&({{\rho }^{\frac{\alpha+1}{2}}}H_{11}H^{\dagger}_{12}+{{\rho }^{\frac{\alpha+1}{2}}}H_{21}H^{\dagger}_{22}))\nonumber\\
&=&\log \det  (I_{M}+{\rho }H_{11}H^{\dagger }_{11}+{\rho }^{\alpha}H_{21}H^{\dagger }_{21})+
\log \det (I_{M}+{\rho }H_{22}H^{\dagger }_{22}+{\rho }^{\alpha}H_{12}H^{\dagger }_{12}-\nonumber\\
&&{{\rho }^{\alpha+1}}(H_{12}H^{\dagger}_{11}+H_{22}H^{\dagger}_{21}){(I_{M}+{\rho }H_{11}H^{\dagger }_{11}+{\rho }^{\alpha}H_{21}H^{\dagger }_{21})^{-1}}(H_{11}H^{\dagger}_{12}+H_{21}H^{\dagger}_{22}))\nonumber\\
&=& (2M\max\{1,\alpha\}){\log \mathsf{SNR}\ }+o({\log \mathsf{SNR}\ }),
\end{eqnarray}
where $(a)$ is obtained from Lemma \ref{lem_block}. Now, dividing both sides by $\log \mathsf{SNR}$, we obtain \eqref{gdofeq6}.

\noindent {\bf \eqref{ro0eq7}$\rightarrow$\eqref{gdofeq7}:} Consider bound \eqref{ro0eq7} in $\mathcal{R}_o$, we have
\begin{eqnarray}
&&\log \det  \log  \det (I_{M}+{\rho}H_{11}H^{\dagger}_{11}-
{\rho}^{\alpha+1}H_{11}H^{\dagger}_{12}
{(I_{N_2}+{\rho}^{\alpha}H_{12}H^{\dagger}_{12})^{-1}}H_{12}H^{\dagger}_{11})+\nonumber\\
&&\log  \det (I_{M}+{\rho}H_{22}H^{\dagger}_{22}+
{\rho}^{\alpha}H_{12}H^{\dagger}_{12}-
{\rho}^{\alpha+1}H_{22}H^{\dagger}_{21}
{(I_{N_1}+{\rho}^{\alpha}H_{21}H^{\dagger}_{21})^{-1}}H_{21}H^{\dagger}_{22})+\nonumber\\
&&\log  \det (I_{M}+{\rho}H_{11}H^{\dagger}_{11}+
{\rho}^{\alpha}H_{21}H^{\dagger}_{21})+C_{12}+C_{21} \nonumber\\
&\stackrel{(a)}{=}& (
M\max\{({1}-{\alpha})^+,{\alpha}\}+({1}-{\alpha})^+M
+M\max\{{1},{\alpha}\}+2\beta){\log \mathsf{SNR}\ }+o({\log \mathsf{SNR}\ }),
\end{eqnarray}
where $(a)$ is obtained from Lemma \ref{gg}, Lemma \ref{gg} and Lemma \ref{lemma-dof}. Now, dividing both sides by $\log \mathsf{SNR}$, we obtain \eqref{gdofeq7}.

\noindent {\bf \eqref{ro0eq8}$\rightarrow$\eqref{gdofeq8}:} This is obtained similarly to the last bound by exchanging $1$ and $2$ in the indices.

\noindent {\bf \eqref{ro0eq9}$\rightarrow$\eqref{gdofeq9}:} Consider bound \eqref{ro0eq9} in $\mathcal{R}_o$, we have

\begin{eqnarray}
&&\log  \det \left(I_{2M}+
\left[\begin{array}{c}
\sqrt{\rho}H_{22} \\
\sqrt{\rho^{\alpha}}H_{21} \end{array}\right]
\left(I_{M}-{\rho}^{\alpha}H^{\dagger}_{21}{\left(I_{N_1}+{\rho}^{\alpha}H_{21}H^{\dagger}_{21}\right)^{-1}}H_{21}\right)
\left[\sqrt{\rho}H_{22}^{\dagger}\ \sqrt{\rho^{\alpha}}H_{21}^{\dagger}\right]\right.\nonumber\\
&&\left.+\left[\begin{array}{c}
\sqrt{\rho^{\alpha}}H_{12} \\
\sqrt{\rho}H_{11} \end{array}\right]
\left[\sqrt{\rho^{\alpha}}H_{12}^{\dagger}\ \sqrt{\rho}H_{11}^{\dagger}\right]\right)+\log  \det \left(I_{M}+{\rho}H_{11}H^{\dagger}_{11}+
{\rho}^{\alpha}H_{21}H^{\dagger}_{12}\right)+C_{21}\nonumber\\
&\stackrel{(a)}{=}&
\log  \det \left(I_{2M}+
\left[\begin{array}{c}
\sqrt{\rho}H_{22} \\
\sqrt{\rho^{\alpha}}H_{21} \end{array}\right]
{\left(I_{M}+{\rho}^{\alpha}H^{\dagger }_{21}H_{21}\right)^{-1}}
\left[\sqrt{\rho}H_{22}^{\dagger}\ \sqrt{\rho^{\alpha}}H_{21}^{\dagger}\right]\right.\nonumber\\
&&\left.+\left[\begin{array}{c}
\sqrt{\rho^{\alpha}}H_{12} \\
\sqrt{\rho}H_{11} \end{array}\right]
\left[\sqrt{\rho^{\alpha}}H_{12}^{\dagger}\ \sqrt{\rho}H_{11}^{\dagger}\right]\right)+\log  \det \left(I_{M}+{\rho}H_{11}H^{\dagger}_{11}+
{\rho}^{\alpha}H_{21}H^{\dagger}_{12}\right)+C_{21}\nonumber\\
&\stackrel{(b)}{=}&
\left(M\max\{{1},{\alpha}\}+\beta\right){\log \mathsf{SNR}\ }+o\left({\log \mathsf{SNR}}\right) + \log \det \nonumber\\
&&\left[ \begin{array}{cc}
I_{M}+{\rho }H_{22}{(I_{M}+{\rho}^{\alpha}H^{\dagger }_{21}H_{21})^{-1}}H^{\dagger }_{22}+{\rho }^{\alpha}H_{12}H^{\dagger }_{12} & {{\rho }^{\frac{\alpha+1}{2}}}(H_{22}{(I_{M}+{\rho}^{\alpha}H^{\dagger }_{21}H_{21})^{-1}}H^{\dagger}_{21}+H_{12}H^{\dagger}_{11}) \\
{{\rho }^{\frac{\alpha+1}{2}}}(H_{21}{(I_{M}+{\rho}^{\alpha}H^{\dagger }_{21}H_{21})^{-1}}H^{\dagger}_{22}+H_{11}H^{\dagger}_{12}) & I_{M}+{\rho }H_{11}H^{\dagger }_{11}+{\rho }^{\alpha}H_{21}{(I_{M}+{\rho}^{\alpha}H^{\dagger }_{21}H_{21})^{-1}}H^{\dagger }_{21} \end{array}
\right]\nonumber\\
&\stackrel{(c)}{=}&
(M\max\{{1},{\alpha}\}+\beta){\log \mathsf{SNR}\ }+o({\log \mathsf{SNR}}) +\log \det  (I_{M}+{\rho }H_{11}H^{\dagger }_{11}+\nonumber\\
&&{\rho }^{\alpha}H_{21}{(I_{M}+{\rho}^{\alpha}H^{\dagger }_{21}H_{21})^{-1}}H^{\dagger }_{21})+
\log \det (I_{M}+{\rho }H_{22}{(I_{M}+{\rho}^{\alpha}H^{\dagger }_{21}H_{21})^{-1}}H^{\dagger }_{22}+{\rho }^{\alpha}H_{12}H^{\dagger }_{12}-\nonumber\\
&&({{\rho }^{\frac{\alpha+1}{2}}}(H_{22}{(I_{M}+{\rho}^{\alpha}H^{\dagger }_{21}H_{21})^{-1}}H^{\dagger}_{21}+H_{12}H^{\dagger}_{11}))
{(I_{M}+{\rho }H_{11}H^{\dagger }_{11}+{\rho }^{\alpha}H_{21}{(I_{M}+{\rho}^{\alpha}H^{\dagger }_{21}H_{21})^{-1}}H^{\dagger }_{21})^{-1}}\nonumber\\
&&({{\rho }^{\frac{\alpha+1}{2}}}(H_{21}{(I_{M}+{\rho}^{\alpha}H^{\dagger }_{21}H_{21})^{-1}}H^{\dagger}_{22}+H_{11}H^{\dagger}_{12})))\nonumber\\
&{=}&
(M + M\max\{{1},{\alpha}\}+\beta){\log \mathsf{SNR}}+ \log \det (I_{M}+{\rho }H_{22}{(I_{M}+{\rho}^{\alpha}H^{\dagger }_{21}H_{21})^{-1}}H^{\dagger }_{22}+{\rho }^{\alpha}H_{12}H^{\dagger }_{12}-\nonumber\\
&&{{\rho }^{\alpha+1}}(H_{22}{(I_{M}+{\rho}^{\alpha}H^{\dagger }_{21}H_{21})^{-1}}H^{\dagger}_{21}+H_{12}H^{\dagger}_{11})
{(I_{M}+{\rho }H_{11}H^{\dagger }_{11}+{\rho }^{\alpha}H_{21}{(I_{M}+{\rho}^{\alpha}H^{\dagger }_{21}H_{21})^{-1}}H^{\dagger }_{21})^{-1}}\nonumber\\
&&(H_{21}{(I_{M}+{\rho}^{\alpha}H^{\dagger }_{21}H_{21})^{-1}}H^{\dagger}_{22}+H_{11}H^{\dagger}_{12}))+o({\log \mathsf{SNR}}) \nonumber
\end{eqnarray}

\begin{eqnarray}
&{=}&
(M + M\max\{{1},{\alpha}\}+\beta){\log \mathsf{SNR}\ }+ \log \det (I_{M}+{\rho }H_{22}{(I_{M}+{\rho}^{\alpha}H^{\dagger }_{21}H_{21})^{-1}}H^{\dagger }_{22}+{\rho }^{\alpha}H_{12}H^{\dagger }_{12}\nonumber\\
&&-{\rho }^{\alpha+1}H_{12}H^{\dagger}_{11}{(I_{M}+{\rho }H_{11}H^{\dagger }_{11}+{\rho }^{\alpha}H_{21}{(I_{M}+{\rho}^{\alpha}H^{\dagger }_{21}H_{21})^{-1}}H^{\dagger }_{21})^{-1}}H_{11}H^{\dagger}_{12}\nonumber\\
&&-{\rho }^{\alpha+1}H_{22}{(I_{M}+{\rho}^{\alpha}H^{\dagger }_{21}H_{21})^{-1}}H^{\dagger}_{21}{(I_{M}+{\rho }H_{11}H^{\dagger }_{11}+{\rho }^{\alpha}H_{21}{(I_{M}+{\rho}^{\alpha}H^{\dagger }_{21}H_{21})^{-1}}H^{\dagger }_{21})^{-1}}H_{11}H^{\dagger}_{12}\nonumber\\
&&-{\rho }^{\alpha+1}H_{12}H^{\dagger}_{11}{(I_{M}+{\rho }H_{11}H^{\dagger }_{11}+{\rho }^{\alpha}H_{21}{(I_{M}+{\rho}^{\alpha}H^{\dagger }_{21}H_{21})^{-1}}H^{\dagger }_{21})^{-1}}H_{21}{(I_{M}+{\rho}^{\alpha}H^{\dagger }_{21}H_{21})^{-1}}H^{\dagger}_{22}\nonumber\\
&&-{\rho }^{\alpha+1}H_{22}{(I_{M}+{\rho}^{\alpha}H^{\dagger }_{21}H_{21})^{-1}}H^{\dagger}_{21}{(I_{M}+{\rho }H_{11}H^{\dagger }_{11}+{\rho }^{\alpha}H_{21}{(I_{M}+{\rho}^{\alpha}H^{\dagger }_{21}H_{21})^{-1}}H^{\dagger }_{21})^{-1}}H_{21}
\nonumber\\
&&{(I_{M}+{\rho}^{\alpha}H^{\dagger }_{21}H_{21})^{-1}}H^{\dagger}_{22})+o({\log \mathsf{SNR}})\nonumber\\
&\stackrel{(d)}{=}&
(M + M\max\{{1},{\alpha}\}+\beta){\log \mathsf{SNR}}+ \log \det (I_{M}+{\rho }H_{22}{(I_{M}+{\rho}^{\alpha}H^{\dagger }_{21}H_{21})^{-1}}H^{\dagger }_{22}+{\rho }^{\alpha}H_{12}H^{\dagger }_{12}\nonumber\\
&&-{\rho }^{\alpha+1}H_{12}H^{\dagger}_{11}{(I_{M}+{\rho }H_{11}H^{\dagger }_{11}+{\rho }^{\alpha}H_{21}{(I_{M}+{\rho}^{\alpha}H^{\dagger }_{21}H_{21})^{-1}}H^{\dagger }_{21})^{-1}}H_{11}H^{\dagger}_{12})+o({\log \mathsf{SNR}})\nonumber\\
&{=}&
(M + M\max\{{1},{\alpha}\}+\beta){\log \mathsf{SNR}\ }+ \log \det (I_{M}+{\rho }H_{22}{(I_{M}+{\rho}^{\alpha}H^{\dagger }_{21}H_{21})^{-1}}H^{\dagger }_{22}+{\rho }^{\alpha}H_{12}H^{\dagger }_{12}\nonumber\\
&&-{\rho }^{\alpha+1}H_{12}H^{\dagger}_{11}{(I_{M}+{\rho }H_{11}H^{\dagger }_{11}+
({\rho }^{-\alpha}{H^{\dagger}_{21}}^{-1}{(I_{M}+{\rho}^{\alpha}H^{\dagger }_{21}H_{21})}H^{-1}_{21}
)^{-1})^{-1}}H_{11}H^{\dagger}_{12})+o({\log \mathsf{SNR}}) \nonumber\\
&{=}&
(M + M\max\{{1},{\alpha}\}+\beta){\log \mathsf{SNR}\ }+ \log \det (I_{M}+{\rho }H_{22}{(I_{M}+{\rho}^{\alpha}H^{\dagger }_{21}H_{21})^{-1}}H^{\dagger }_{22}+{\rho }^{\alpha}H_{12}H^{\dagger }_{12}\nonumber\\
&&-{\rho }^{\alpha+1}H_{12}H^{\dagger}_{11}{(I_{M}+{\rho }H_{11}H^{\dagger }_{11}+
({\rho }^{-\alpha}{H^{\dagger}_{21}}^{-1}H^{-1}_{21}+I_{M}
)^{-1})^{-1}}H_{11}H^{\dagger}_{12})+o({\log \mathsf{SNR}}) \nonumber\\
&{=}&
(M + M\max\{{1},{\alpha}\}+\beta){\log \mathsf{SNR}\ }+ \log \det (I_{M}+{\rho }H_{22}{(I_{M}+{\rho}^{\alpha}H^{\dagger }_{21}H_{21})^{-1}}H^{\dagger }_{22}+\nonumber\\
&&{\rho }^{\alpha}H_{12}
(I_M-{\rho }H^{\dagger}_{11}{(I_{M}+{\rho }H_{11}H^{\dagger }_{11}+
({\rho }^{-\alpha}{H^{\dagger}_{21}}^{-1}H^{-1}_{21}+I_{M}
)^{-1})^{-1}}H_{11})
H^{\dagger}_{12})+o({\log \mathsf{SNR}}) \nonumber\\
&{=}&
(M + M\max\{{1},{\alpha}\}+\beta){\log \mathsf{SNR}\ }+ \log \det (I_{M}+{\rho }H_{22}{(I_{M}+{\rho}^{\alpha}H^{\dagger }_{21}H_{21})^{-1}}H^{\dagger }_{22}+\nonumber\\
&&{\rho }^{\alpha}H_{12}
(I_M-({\rho }^{-1}H^{-1}_{11}{(I_{M}+{\rho }H_{11}H^{\dagger }_{11}+
({\rho }^{-\alpha}{H^{\dagger}_{21}}^{-1}H^{-1}_{21}+I_{M}
)^{-1})}{H_{11}^{\dagger}}^{-1})^{-1})
H^{\dagger}_{12})+o({\log \mathsf{SNR}}) \nonumber\\
&\stackrel{(e)}{=}&
(M + M\max\{{1},{\alpha}\}+\beta){\log \mathsf{SNR}\ }+ \log \det (I_{M}+{\rho }H_{22}{(I_{M}+{\rho}^{\alpha}H^{\dagger }_{21}H_{21})^{-1}}H^{\dagger }_{22}+\nonumber\\
&&{\rho }^{\alpha}H_{12}
({\rho }^{-1}H^{-1}_{11}{(I_{M}+{\rho }H_{11}H^{\dagger }_{11}+
({\rho }^{-\alpha}{H^{\dagger}_{21}}^{-1}H^{-1}_{21}+I_{M}
)^{-1})}{H_{11}^{\dagger}}^{-1}-I_M)\nonumber\\
&&
({\rho }^{-1}H^{-1}_{11}{(I_{M}+{\rho }H_{11}H^{\dagger }_{11}+
({\rho }^{-\alpha}{H^{\dagger}_{21}}^{-1}H^{-1}_{21}+I_{M}
)^{-1})}{H_{11}^{\dagger}}^{-1})^{-1}
H^{\dagger}_{12})+o({\log \mathsf{SNR}}) \nonumber\\
&{=}&
(M + M\max\{{1},{\alpha}\}+\beta){\log \mathsf{SNR}\ }+ \log \det (I_{M}+{\rho }H_{22}{(I_{M}+{\rho}^{\alpha}H^{\dagger }_{21}H_{21})^{-1}}H^{\dagger }_{22}+\nonumber\\
&&{\rho }^{\alpha}H_{12}
({\rho }^{-1}H^{-1}_{11}{(I_{M}+{\rho }H_{11}H^{\dagger }_{11}+
({\rho }^{-\alpha}{H^{\dagger}_{21}}^{-1}H^{-1}_{21}+I_{M}
)^{-1})}{H_{11}^{\dagger}}^{-1}-{{\rho}^{-1}}H^{-1}_{11}(\rho H_{11}H^{\dagger}_{11})
{H^{\dagger}_{11}}^{-1})\nonumber\\
&&
({\rho }^{-1}H^{-1}_{11}{(I_{M}+{\rho }H_{11}H^{\dagger }_{11}+
({\rho }^{-\alpha}{H^{\dagger}_{21}}^{-1}H^{-1}_{21}+I_{M}
)^{-1})}{H_{11}^{\dagger}}^{-1})^{-1}
H^{\dagger}_{12})+o({\log \mathsf{SNR}}) \nonumber\\
&{=}&
(M + M\max\{{1},{\alpha}\}+\beta){\log \mathsf{SNR}\ }+ \log \det (I_{M}+{\rho }H_{22}{(I_{M}+{\rho}^{\alpha}H^{\dagger }_{21}H_{21})^{-1}}H^{\dagger }_{22}+\nonumber\\
&&{\rho }^{\alpha}H_{12}
({\rho }^{-1}H^{-1}_{11}{(I_{M}+
({\rho }^{-\alpha}{H^{\dagger}_{21}}^{-1}H^{-1}_{21}+I_{M}
)^{-1})}{H_{11}^{\dagger}}^{-1})\nonumber\\
&&
({\rho }^{-1}H^{-1}_{11}{(I_{M}+{\rho }H_{11}H^{\dagger }_{11}+
({\rho }^{-\alpha}{H^{\dagger}_{21}}^{-1}H^{-1}_{21}+I_{M}
)^{-1})}{H_{11}^{\dagger}}^{-1})^{-1}
H^{\dagger}_{12})+o({\log \mathsf{SNR}}) \nonumber
\end{eqnarray}

\begin{eqnarray}
&{=}&
(M + M\max\{{1},{\alpha}\}+\beta){\log \mathsf{SNR}\ }+
M\max\{({1}-{\alpha})^+,{\alpha-1}\}{\log \mathsf{SNR}\ }+o({\log \mathsf{SNR}}) \nonumber\\
&=& (M\max\{({2}-{\alpha})^+,{\alpha}\} + M\max\{{1},{\alpha}\}+\beta){\log \mathsf{SNR}\ }+o({\log \mathsf{SNR}}),
\end{eqnarray}
where $(a)$ is obtained from Lemma \ref{sym-lemma}, $(b)$ is obtained from Lemma \ref{gg}, $(c)$ is obtained from Lemma \ref{lem_block}, $(d)$ is because the three eliminated sentences have a constant upper bounds and $(e)$ follows from the fact that $I_M-X^{-1}=(X-I_M)X^{-1}$ where $X={\rho }^{-1}H^{-1}_{11}{(I_{M}+{\rho }H_{11}H^{\dagger }_{11}+
({\rho }^{-\alpha}{H^{\dagger}_{21}}^{-1}H^{-1}_{21}+I_{M}
)^{-1})}{H_{11}^{\dagger}}^{-1}$. Now, dividing both sides by $\log \mathsf{SNR}$, we obtain \eqref{gdofeq9}.

\noindent {\bf \eqref{ro0eq10}$\rightarrow$\eqref{gdofeq10}:} This is obtained similarly to the last bound by exchanging $1$ and $2$ in the indices.

Combining the above results we obtain Theorem \ref{thm_gdof} results.

\end{appendices}

\bibliographystyle{IEEETran}
\bibliography{bib}

\end{document}